\documentclass[12pt]{article}

\usepackage{amssymb}
\usepackage{graphicx}
\usepackage{amsmath}
\usepackage{amsfonts,amsthm}

\newtheorem{theorem}{Theorem}
\newtheorem{corollary}{Corollary}

\newtheorem{proposition}{Proposition}
\newtheorem{remark}{Remark}

\newcommand{\half}{\small \mbox{$\frac{1}{2}$}}
\usepackage[dvipdfmx]{color}
\newcommand{\bm}[1]{\mbox{\boldmath $#1$}}
\newcommand{\argmin}{\mathop{\rm argmin}}
\newcommand{\argmax}{\mathop{\rm argmax}}

\usepackage{modnatbib}
\usepackage{url}

\usepackage{here}
\usepackage{setspace}

\begin{document}
\title{\bf Cumulant-based approximation for fast and efficient prediction for species distribution}

\author{Osamu Komori\\
Department of Science and Technology, Seikei University\\
 3-3-1 Kichijoji-kitamachi, Musashino, Tokyo 180-8633, Japan
\\
Yusuke Saigusa\\
Department of Biostatistics, School of Medicine, Yokohama City University\\
3-9 Fukuura, Kanazawa, Yokohama, Kanagawa 236-0004, Japan\\
 Shinto Eguchi\\
The Institute of Statistical Mathematics\\
 10-3 Midori-cho, Tachikawa, Tokyo 190-8562, Japan\\
 Yasuhiro Kubota\\
 Faculty of Science, University of the Ryukyus\\
 1 Senbaru, Nishihara, Okinawa 903-0213, Japan
}

\maketitle

\begin{abstract}
Species distribution modeling plays an important role in estimating the habitat suitability of species using environmental variables. For this purpose, Maxent and the Poisson point process are popular and powerful methods extensively employed across various ecological and biological sciences. However, the computational speed becomes prohibitively slow when using huge background datasets, which is often the case with fine-resolution data or global-scale estimations.
To address this problem, we propose a computationally efficient species distribution model using a cumulant-based approximation (CBA) applied to the loss function of $\gamma$-divergence. Additionally, we introduce a sequential estimating algorithm with an $L_1$ penalty to select important environmental variables closely associated with species distribution.
The regularized geometric-mean method, derived from the CBA, demonstrates high computational efficiency and estimation accuracy. Moreover, by applying CBA to Maxent, we establish that Maxent and Fisher linear discriminant analysis are equivalent under a normality assumption. This equivalence leads to an highly efficient computational method for estimating species distribution. The effectiveness of our proposed methods is illustrated through simulation studies and by analyzing data on 226 species from the National Centre for Ecological Analysis and Synthesis and 709 Japanese vascular plant species.
The computational efficiency of the proposed methods is significantly improved compared to Maxent, while maintaining comparable estimation accuracy. A R package {\tt CBA} is also prepared to provide all programming codes used in simulation studies and real data analysis.
\end{abstract}
\newpage

\section{Introduction}
Species distribution models (SDMs) are widely used to estimate the relationship between species and their associated environmental features, clarifying how and why species' ranges change over time \citep{Gaston2003}. As essential tools in ecology, climate change, and conservation biology, they support conservation decision-making when the collaboration between modelers and decision-makers is well established \citep{Guisan2005, Guisan2013}. The primary data used for SDMs are presence-only (PO) data, which lack information on species absence and can be collected from museum and herbarium databases at relatively low cost. In contrast, presence-absence (PA) data, which include locations of both presence and absence, are difficult to obtain as they require systematic surveys involving significant time and effort.

Among the SDMs for PO data, Maxent \citep{Phillips2006} is a powerful machine learning method used for various purposes such as conservation planning and investigating the effects of invasive species and climate changes \citep{Elith2011}. Designed to sequentially maximize likelihood based on the maximum entropy principle \citep{Dudik2004}, Maxent's performance ranks among the highest based on the extensive NCEAS data, which includes 226 species from six regions worldwide \citep{Elith2006}. Recent studies confirm that Maxent and boosted regression trees remain top performers compared to other machine learning methods like random forests, XGBoost, support vector machines, and so on \citep{Valavi2022}. Theoretical work has shown that estimation by Maxent is equivalent to the Poisson point process model (PPM) and logistic regression model \citep{Renner2013, Warton2010}, with \cite{Fithian2013} proposing infinitely weighted logistic regression as also equivalent to PPM. For a detailed review, see \citep{Komori2023}.

Recently, SDMs have been applied to large-scale datasets to estimate global species distributions \citep{Heshmati2019, Lee2021}. When using fine-grain background data, computational costs become prohibitively high \citep{Phillips2008}, often mitigated by employing random sampling of background data. The default sample size for random sampling in Maxent is set to 10,000 \citep{Hijmans2023}, although studies suggest that 86,000 points are needed for reliable results \citep{Renner2015}, with recent recommendations suggesting 50,000 based on area under the ROC curve (AUC) values \citep{Valavi2022}. Practical strategies include increasing the number of background points until model fitting and predictive performance stabilize \citep{Phillips2008}.

This paper introduces a different approach to improve computational efficiency in Maxent and, equivalently, in PPM, by using all locations of background data and applying the cumulant-based approximation (CBA) to the normalization constant of Maxent or the numerical integration of the intensity function of PPM. Known as the Thouless-Anderson-Palmer (TAP) approximation based on the Plefka expansion in Boltzmann Machine learning \citep{Kappen1997, Tanaka1998}, CBA also unveils a compelling relationship between Maxent and Fisher linear discriminant analysis, where the estimated coefficients of Maxent's parameters are equivalent to those of Fisher's under a normality assumption, resulting in an extremely speedy estimation of species distribution. This normality assumption aligns with the log Gaussian Cox process \citep{Moller1998}. Furthermore, we apply CBA to the loss function of $\gamma$-divergence \citep{Basu1998, Fujisawa2008}, deriving a new method called the regularized geometric-mean method (rGM), which offers both low computational costs and high prediction accuracy. The loss function of rGM includes a regularization term based on the sample variance of the background data, helping to avoid overfitting compared to the original geometric-mean method (GM) derived from geometric-mean divergence.

This paper is organized as follows: First, we introduce PPM and the $\gamma$-loss function, clarifying their relation to Maxent and its statistical properties. We then apply CBA to the $\gamma$-loss function to demonstrate the equivalence of Maxent (PPM) to Fisher linear discrimination under the normality assumption. Next, we show that the first-order CBA of the $\gamma$-loss function corresponds to that of GM, and the second-order CBA corresponds to that of rGM. We also formulate the estimating algorithm for rGM, incorporating an $L_1$ penalty for variable selection. Lastly, we evaluate the performance of our proposed methods across various simulation settings and real data analysis, including a dataset of Japanese vascular plants \citep{Kubota2015}. We conclude with a discussion on the relationship with existing methods and suggestions for further research. All R codes used in the simulation studies and real data analysis are provided in the {\tt CBA} package. A brief explanation of how to use {\tt CBA} is also provided in a supplementary file.

\section{Materials and Methods}
\subsection*{Poisson point process and its fast estimating algorithm}
In the Poisson point process regression model (PPM), we deal with presence-only data, which consists of $m$ mutually independent locations $\{x_1,\ldots,x_m\}$, and background data (also called pseudo-absences), consisting of locations $\{x_{m+1},\ldots, x_n\}$ in a study region $\mathcal A$. For a target species, we consider an intensity function $\lambda_0(x_i,\bm \alpha,c)$ parameterized by a slope vector $\bm \alpha$ and an intercept $c$. Typically, we assume a log-linear model such as
\begin{equation}\label{eq1}
    \log \lambda_0(x_i,\bm \alpha,c)=c+\bm\alpha^\top \bm f(x_i),
\end{equation}
where $\bm f(x_i)$ is a $p$-dimensional feature vector associated with the presence of the target species.
Then the log-likelihood of PPM \citep{Warton2010,Renner2013} is given as 
\begin{eqnarray}\label{eq2}
\ell_{\rm PPM}(\bm\alpha,c)&=&
\sum_{i=1}^m  \log \lambda_0(x_i,\bm\alpha,c)- \Lambda_0(\bm\alpha,c),
\end{eqnarray}
where $\Lambda_0(\bm\alpha,c)=\sum_{i=1}^n w_i \lambda_0(x_i,\bm\alpha,c)$ is an approximation of $\int_{\cal A}  \lambda_0(x ,\bm\alpha,c)dx$ using a quadrature weight $w_i$ for a location $x_i$. Hereafter, we standardize the weights and assume $\sum_{i=1}^n w_i=1$ without loss of generality. Equation (\ref{eq2}) can be regarded as a weighted Poisson likelihood \citep{Berman1992}, denoted as
\begin{equation}
    \sum_{i=1}^n w_i\{z_i\log \lambda_0(x_i,\bm \alpha,c)-\lambda_0(x_i,\bm \alpha,c)\},
\end{equation}
where $z_i=I(i\in\{1,\ldots,m\})/w_i$ and $I(\cdot)$ is an indicator function. Then the estimator of PPM is defined as
\begin{eqnarray}
 (\hat{\bm\alpha}_{\rm PPM},\hat{c}_{\rm PPM})=\argmax_{\boldmath\alpha,c} \ell_{\rm PPM}(\bm\alpha,c).
\end{eqnarray}
The runtime for seeking estimators in both PPM and Maxent \citep{Phillips2006} is prohibitively slow for large $n$ in the iteration process of learning algorithms \citep{Phillips2008}. \cite{Valavi2022} illustrate the relation between the runtime and $n$ together with the estimation accuracy using a species in Australia. The evaluation of $\Lambda_0(\bm\alpha,c)$ is time-consuming for PPM. To mitigate this issue, we take the CBA approach, resulting in $\Lambda_0(\bm\alpha,c)\approx \exp(c+\bm\alpha^\top \bar{\bm f}+\bm\alpha^\top S\bm\alpha)$, where $\bar {\bm f}$ and $S$ are the sample mean and variance defined as $\bar {\bm f}=\sum_{i=1}^n w_i \bm f(x_i)$ and $S=\sum_{i=1}^n w_i (\bm f(x_i)-\bar {\bm f})(\bm f(x_i)-\bar {\bm f})^\top$. This yields a simple approximation of $\ell_{\rm PPM}(\bm\alpha,c)$ as
\begin{equation}
\ell_{\rm PPM}(\bm\alpha,c)\approx m(c+\bm\alpha^\top\bar{\bm f}_m)-\exp(c+\bm\alpha^\top \bar{\bm f}+\bm\alpha^\top S\bm\alpha),
\end{equation}
where $\bar{\bm f}_m$ is the sample mean over presence locations $\bar{\bm f}_m=1/m\sum_{i=1}^m \bm f(x_i)$. We note that this approximation is defined only by sufficient statistics $\bar{\bm f}_m$, $\bar{\bm f}$ and $S$ without any further computations over $\{x_1,\ldots,x_n\}$ in the iteration process. See the following subsections for the technical discussion for the approximation.
In the log-linear model, we have a fast iterative estimation algorithm
\begin{equation}
\begin{bmatrix}
   \bm\alpha_{new}  \\
   c_{new}
\end{bmatrix}
\leftarrow
\begin{bmatrix}
     S^{-1}\left\{\frac{\displaystyle m}{\displaystyle\exp(c+\bm\alpha^\top \bar{\bm f}+\bm\alpha^\top S\bm\alpha)}\bar {\bm f}_m-\bar {\bm f}\right\}\vspace{0.3cm}\\
  \log m-(\bm\alpha^\top \bar{\bm f}+\bm\alpha^\top S\bm\alpha)
\end{bmatrix}
 \label{alpha.c.new}
\end{equation}
Note that $\exp(c+\bm\alpha^\top \bar{\bm f}+\bm\alpha^\top S\bm\alpha)$ is the second-order approximation of the expectation of $m$ in the Poisson point process; hence, the top-right-hand side in (\ref{alpha.c.new}) is almost the same as the coefficient of the Fisher linear discriminant analysis $\hat{\bm\alpha}_{\rm Fisher}=S^{-1}(\bar {\bm f}_m-\bar {\bm f})$ when the iteration properly converges. 
It is widely known that the estimation of the slope parameter $\bm \alpha$ in (\ref{eq1}) of PPM is equivalent to that of Maxent \citep{Renner2013}. Hence, equation (\ref{alpha.c.new}) also corresponds to the fast estimating algorithm for Maxent. In the subsequent subsection, we derive a faster and more efficient estimating method for Maxent, where no iteration is needed.

\subsection*{$\gamma$-loss function}
The $\gamma$-loss function for an intensity function $\lambda(x_i,\bm\alpha)$ \citep{Komori2023} is given as
\begin{eqnarray}\label{g_loss}
L_\gamma(\bm\alpha) &=& -\frac{1}{\gamma}\left[\frac{\sum_{i=1}^m \lambda(x_i,\bm\alpha)^\gamma}{\left\{\sum_{i=1}^n w_i \lambda(x_i, \bm\alpha)^{\gamma+1}\right\}^{\frac{\gamma}{\gamma+1}}}-m\right],
\end{eqnarray}
where $\gamma>-1$. The estimator is defined as
\begin{equation}
    \hat{\bm\alpha}_\gamma = \argmin_{\boldmath\alpha} L_\gamma(\bm\alpha).
\end{equation}
Note that we omit an intercept $c$ from the intensity function because it is canceled out in (\ref{g_loss}) when we assume the log-linear model in (\ref{eq1}). The statistical properties of the $\gamma$-loss function, particularly its impact on estimator performance and bias, are intensively investigated by \cite{Fujisawa2008}. As a special case, we have
\begin{eqnarray}
L_0(\bm\alpha) \equiv \lim_{\gamma\rightarrow0} L_\gamma(\bm\alpha) = -\sum_{i=1}^m \log\frac{\lambda(x_i,\bm\alpha)}{\Lambda(\bm\alpha)},
\end{eqnarray}
where $\Lambda(\bm\alpha) = \sum_{i=1}^n w_i \lambda(x_i, \bm\alpha)$.
This is equivalent to the negative log-likelihood of Maxent. Hence, similarly to the theorem of \cite{Renner2013}, under the assumption of a log-linear model for intensity functions for PPM and $L_0(\bm\alpha)$, we establish an important relation between $\hat{\bm\alpha}_{\rm PPM}$ and $\hat{\bm\alpha}_0$ as
\begin{eqnarray}\label{eq8}
\hat{\bm\alpha}_{\rm PPM} = \hat{\bm\alpha}_0, \ \hat c_{\rm PPM} = \log \frac{m}{\Lambda(\hat{\bm\alpha}_0)}.
\end{eqnarray}
See Appendix \ref{appendA} for the details of the derivation. From equation (\ref{eq8}), we observe that the estimated intensity function of PPM, i.e., $\lambda_0(x_i, \hat{\bm\alpha}_{\rm PPM}, \hat c_{\rm PPM})$, is completely reproduced by the estimator of the slope of Maxent, $\hat{\bm\alpha}_0$, and the number of presence locations, $m$. This relationship can be extended to that between the $\gamma$-loss function and the $\beta$-loss function derived from $\beta$-divergence \citep{Basu1998, Minami2002, Eguchi2022}. See the details in Appendix \ref{appendB}.

Next, we consider the statistical consistency of $\hat{\bm\alpha}_\gamma$ based on the $\gamma$-divergence defined as
\begin{eqnarray}
D_\gamma(\lambda_1, \lambda_2) = -\frac{1}{\gamma} \frac{\sum_{i=1}^n w_i \lambda_2(x_i)^\gamma \lambda_1(x_i)}{\left\{\sum_{i=1}^n w_i \lambda_2(x_i)^{\gamma+1}\right\}^{\frac{\gamma}{\gamma+1}}} + \frac{1}{\gamma} \left\{\sum_{i=1}^n w_i \lambda_1(x_i)^{\gamma+1}\right\}^{\frac{1}{\gamma+1}},
\end{eqnarray}
where $\lambda_1(x)$ and $\lambda_2(x)$ are intensity functions. Cf. \cite{Saigusa2024} for robust divergences.
\begin{proposition}\label{prop1}
If there exists $\bm\alpha^*$ such that a true intensity function of a Poisson point process is expressed as $\lambda^*(x) \propto \lambda(x, \bm\alpha^*)$, then the expected loss is expressed as
\begin{eqnarray}
{\mathbb E}\{L_\gamma(\bm\alpha)\} - {\mathbb E}\{L_\gamma(\bm\alpha^*)\} &=& D_\gamma(\lambda(\cdot, \bm\alpha^*), \lambda(\cdot, \bm\alpha)),
\end{eqnarray}
where ${\mathbb E}$ denotes the expectation under $\lambda^*(x)$.
Moreover, we have the unbiasedness of the estimating equation as
\begin{eqnarray}
\mathbb E\left[\frac{\partial}{\partial\bm\alpha}L_\gamma(\bm\alpha)\right] = 0.
\end{eqnarray}
\end{proposition}
See Appendix \ref{appendC} for the proof.
This observation implies that
\begin{equation}\label{consistency}
{\mathbb E}\{L_\gamma(\bm\alpha^*)\} = \min_{\boldmath\alpha}{\mathbb E}\{L_\gamma(\bm\alpha)\},
\end{equation}
which confirms the consistency of the estimator $\hat{\bm\alpha}_\gamma$ for $\bm\alpha^*$.

In the following theorem, we clarify the relationship between the $\gamma$-loss function and the cumulant of $\log \lambda(x, \bm\alpha)$ calculated by the entire locations. This theorem is key to deriving a computationally efficient estimating algorithm for the intensity function of $\lambda(x, \bm\alpha)$ in the subsequent sections.

\subsection*{Cumulant-based approximation (CBA)}

We consider the cumulant generating function of $\log \lambda(x,\bm\alpha)$, defined as
\begin{equation}
K(t)=\log\sum_{i=1}^n w_i\exp\{t\log \lambda(x_i,\bm\alpha)\}.
\end{equation}
It has the power series expansion
\begin{equation}
K(t)=\sum_{j=1}^\infty\kappa_j(\bm\alpha)\frac{t^j}{j!},
\end{equation}
where $\kappa_j(\bm\alpha)$ is the $j$th cumulant.
Because $\kappa_1(\bm \alpha)=\bm \alpha^\top \bar {\bm f}$ and $\kappa_2(\bm \alpha)=\bm\alpha^\top S\bm\alpha$, we have the second-order approximation of $K(t)$ as
\begin{equation}
K(t)\approx t\bm\alpha^\top\bar{\bm f}+\frac{t^2}2\bm\alpha^\top S\bm\alpha.\label{Kt.approx}
\end{equation}
If we assume normality for $\bm\alpha^\top{\bm f}(x)$, then
\begin{equation}
K(t)= t\bm\alpha^\top\bar{\bm f}+\frac{t^2}2\bm\alpha^\top S\bm\alpha
\end{equation}
for any real number $t$.

\begin{theorem}\label{thm4}
Under the log-linear model assumption, $L_\gamma(\bm\alpha)$ adopts a cumulant-based exponential form:
\begin{equation}
 L_\gamma(\bm\alpha)
=
- \frac{1}{\gamma}\Bigg[\sum_{i=1}^m\exp\bigg\{\gamma\bm\alpha^\top \bm f(x_i)- \gamma\sum_{j=1}^\infty\kappa_j(\bm\alpha)\frac{(\gamma+1)^{j-1}}{j!}\bigg\}-m\Bigg].
\end{equation}
\end{theorem}
See Appendix \ref{appendD} for the details of the proof.

\begin{remark}\label{rem1}
Under the log-linear model assumption, we have the first and second-order approximations of $L_\gamma(\bm\alpha)$ as follows:
\begin{align}
 L_\gamma^{(1)}(\bm\alpha)
=&
- \frac{1}{\gamma}\Bigg[\sum_{i=1}^m\exp\bigg\{\gamma\bm\alpha^\top(\bm f(x_i)-\bar {\bm f})\bigg\}-m\Bigg]\label{fapprox}\\
 L_\gamma^{(2)}(\bm\alpha)
=&
- \frac{1}{\gamma}\Bigg[\sum_{i=1}^m\exp\bigg\{\gamma\bigg(\bm\alpha^\top(\bm f(x_i)-\bar {\bm f})-\frac{\gamma+1}2\bm\alpha^\top S\bm\alpha\bigg)\bigg\}-m\Bigg].\label{sapprox}
\end{align}
\end{remark}
The second-order approximation is also employed for the Gibbs free energy in Boltzmann machine learning \citep{Kappen1997}. Equation (\ref{Kt.approx}) is closely related to Gaussian approximations similar to the Laplace approximation. If we choose a higher order of approximation, then the normalizing constant is better approximated but with more computational costs. Hence, we strike a balance between them in practice. We can show that the estimating equation of $L_\gamma^{(2)}(\bm\alpha)$ is unbiased under the normality assumption. See Appendix \ref{appendE}.
\subsection*{Equivalence of PPM to Fisher linear discrimination under normality assumption}
From Remark \ref{rem1} and equation (\ref{eq8}), we derive the following relationship between PPM and the Fisher linear function \citep{Fisher1936}.

\begin{remark}\label{rem2}
Under the log-linear model assumption, the second-order approximations of $L_0(\bm\alpha)$ are as follows:
\begin{eqnarray}
 L_0^{(2)}(\bm\alpha)
&=&m\left\{\bm\alpha^\top(\bar {\bm f}-\bar {\bm f}_m)+\frac12\bm\alpha^\top S\bm\alpha\right\}
\end{eqnarray}
Thus, we have:
\begin{eqnarray}
 \hat{\bm\alpha}_{\rm Fisher}&=&\argmin_{\boldmath\alpha}  L_0^{(2)}(\bm\alpha)\\
&=&S^{-1}(\bar {\bm f}_m-\bar {\bm f}).
\end{eqnarray}
\end{remark}
If feature vectors $\bm f(x_i)$ $(i=1,\ldots,n)$ are standardized preliminarily using $\bar {\bm f}$ and $S$ so that the sample mean and variance are $\bm 0$ and identity matrix $I$, then $\hat{\bm\alpha}_{\rm Fisher}=\bar {\bm f}_m$, which is simple and extremely computationally efficient.

If we make a clear distinction between presence and entire locations, then it is more natural to consider a two-class Fisher coefficient $(\pi_mS_m+S)^{-1}(\bar {\bm f}_m-\bar {\bm f})$,
where $\pi_m$ and $S_m$ are a proportion and a sample variance of the feature vector $f(x_i)$ for presence locations ($i=1,\ldots,m$), respectively. Hence $\hat{\bm\alpha}_{\rm Fisher}$ can be regarded as a special case, in which we put more importance on $S$ rather than $S_m$. Similar ideas are employed in generalized $t$-statistic \citep{Komori2015} and statistical methods for imbalanced data \citep{Komori2019}.

\begin{corollary}\label{coro1}
Under the assumption of the log-linear model and normality assumption for ${\bm\alpha}^\top \bm f(x)$ in the entire region: $\kappa_j(\bm\alpha)=0$ for $j=3,4,\ldots$, we have 
\begin{eqnarray}
\hat{\bm\alpha}_{\rm PPM}=\hat{\bm\alpha}_{\rm Fisher}.
\end{eqnarray}
\end{corollary}
\begin{proof}
It is obvious from equation (\ref{eq8}) and Remark \ref{rem2} because $L_0(\bm\alpha)=L_0^{(2)}(\bm\alpha)$ under the normality assumption.
\end{proof}

\subsection*{Geometric-mean divergence}
The geometric-mean (GM) divergence between intensity functions $\lambda_1(x)$ and $\lambda_2(x)$ is defined as
\begin{eqnarray}
D_{\rm GM}(\lambda_1,\lambda_2) &=& 
{\sum_{i=1}^n w_i{\lambda_2(x_i)^{-1} \lambda_1(x_i)}}
{\prod_{i=1}^n \lambda_2(x_i )^{w_i}}-
{\prod_{i=1}^n \lambda_1(x_i )^{w_i}}.
\end{eqnarray}
The arithmetic and geometric means of $\big\{\frac{\lambda_1(x_i)}{\lambda_2(x_i)}\big\}_{i=1}^n$ satisfy the following inequality:
\begin{eqnarray}
\sum_{i=1}^n \frac{\lambda_1(x_i)}{\lambda_2(x_i)}w_i
 \geq
 \prod_{k=1}^n \left\{\frac{\lambda_1(x_i)}{\lambda_2(x_i)}\right\}^{w_i}.
\end{eqnarray}
This implication directly leads to $D_{\rm GM}(\lambda_1,\lambda_2)\geq0$. The GM-loss function is given as
\begin{eqnarray}\label{GM_ori}
 L_{\rm GM}(\bm\alpha)
=\sum_{i=1}^m 
 \lambda(x_i,\bm\alpha)^{-1}  \prod_{i=1}^n \lambda(x_i,\bm\alpha)^{w_i},
\end{eqnarray}
and the estimator is defined by
\begin{equation}
    \hat{\bm\alpha}_{\rm GM}=\argmin_{\boldmath\alpha} L_{\rm GM}(\bm\alpha). 
\end{equation}
The expected loss function is expressed as
\begin{eqnarray}
\mathbb E [  L_{\rm GM}(\bm\alpha)]
=\sum_{i=1}^n w_i
 \lambda(x_i,\bm\alpha)^{-1} \lambda^*(x_i)  \prod_{i=1}^n \lambda(x_i,\bm\alpha)^{w_i},
\end{eqnarray}
where $\lambda^*(x)$ is the true intensity function of a Poisson point process. Then we have
\begin{eqnarray}
{\mathbb E}\{L_{\rm GM}(\bm\alpha)\}-{\mathbb E}\{L_{\rm GM}(\bm\alpha^*)\}
&=&D_{\rm GM}(\lambda(\cdot,\bm\alpha^*),\lambda(\cdot,\bm\alpha)).
\end{eqnarray}
Similarly to (\ref{consistency}), we conclude that
\begin{equation}
{\mathbb E}\{L_{\rm GM}(\bm\alpha^*)\}=\min_{\boldmath\alpha}{\mathbb E}\{L_{\rm GM}(\bm\alpha)\},
\end{equation}
demonstrating the consistency of the estimator $\hat{\bm\alpha}_{\rm GM}$ for $\bm\alpha^*$. Under the log-linear model, from (\ref{GM_ori}) we have a GM-loss function as
\begin{align}\label{GM}
 L_{\rm GM}(\bm\alpha)
=&\sum_{i=1}^m\exp\left\{-\bm\alpha^\top(\bm f(x_i)-\bar {\bm f})\right\},
\end{align}
which is equivalent to $L^{(1)}_{-1}(\bm \alpha)$ in (\ref{fapprox}). Hence, $L^{(2)}_{\gamma}(\bm \alpha)$ in (\ref{sapprox}) can be seen as the regularized GM loss function (rGM), with the regularization term $-\gamma(\gamma+1)\bm\alpha^\top S\bm\alpha/2$, and is positive if and only if $-1<\gamma<0$.

\subsection*{Estimating algorithm with $L_1$ penalty}

We consider the $L_\gamma^{(2)}(\bm\alpha)$ function penalized by an $L_1$-penalty, defined as
\begin{equation*}
    L_{\gamma,\tau}^{(2)}(\bm\alpha) = L_\gamma^{(2)}(\bm\alpha) + \frac{\tau}{\sqrt{m}} \sum_{j=1}^p s_j |\alpha_j|,
\end{equation*}
where $s_j$ is the sample standard deviation of $f_j(x_i)$ for presence locations ($i=1,\ldots,m$). For details on this type of penalty term, which places more penalty on $f_j(x_i)$ with larger variance, see \cite{Phillips2008}.
Then, for $\bm\alpha'=(\alpha_1,\ldots,\alpha_j+\delta,\ldots,\alpha_p)^\top$, we calculate the upper bound of the loss functions as
\begin{eqnarray}
    L_{\gamma,\tau}^{(2)}(\bm\alpha') - L_{\gamma,\tau}^{(2)}(\bm\alpha)
    \leq D_j^\tau(\delta, \bm\alpha).  
\end{eqnarray}
We propose the following estimating algorithm, similar to that of \cite{Dudik2004}, based on the \textit{soft thresholding operator} used in the gradient vector, as employed in \cite{Goeman2010}.

\noindent
\hrulefill
\begin{enumerate}
\item Set ${\bm\alpha^{(1)}}^\top = (0,\ldots,0)$.
\item For $t=1,\cdots,T$
\begin{enumerate}
\item[(a)] Choose $j^* \in \{1,\ldots,p\}$ and $\delta^* \in \mathbb R$ such that
\begin{equation*}
(j^*, \delta^*) = \underset{j,\delta}{\rm argmin}\ D_j^\tau(\delta, \bm\alpha^{(t)}).
\end{equation*} 
\item[(b)] Update $\bm\alpha^{(t)}$ as 
\begin{equation*}
\alpha^{(t+1)}_j = \left\{
    \begin{array}{l}
      \alpha^{(t)}_j + \delta^* \ {\rm if}\ j = j^* \\
      \alpha^{(t)}_j\hspace{1cm} {\rm otherwise}
    \end{array}
  \right.
\end{equation*}
\end{enumerate}
\item Output $\hat{\bm\alpha} = \bm\alpha^{(T+1)}$.
\end{enumerate}
\hrulefill

Steps (a) and (b) are iterated until convergence, where important features are selected during this process. Similar penalized loss functions and estimating algorithms are also considered for $L_\gamma(\bm \alpha)$, $L_{\rm GM}(\bm\alpha)$ and Fisher. See the details in Appendix \ref{appendF}.

\subsection*{Simulation studies}

We consider simulation settings based on a Cox process, where the intensity $\lambda(x,\bm\alpha)$ varies randomly. Initially, we define the elements of a correlation matrix \( R \) as follows:
\begin{equation*}
r_{jj'}=\left\{
\begin{array}{ll}
\rho & \text{if } j \neq j' \\
1 & \text{if } j = j' \quad (j, j' = 1, \ldots, p).
\end{array}
\right.
\end{equation*}
We then consider three cases:
\begin{enumerate}
\item Poisson case: We generate the dependent variable $u_j \sim \text{Unif}([0,1])$ using a Gaussian copula with the correlation matrix $R$. Subsequently, we generate $f_j(x_i) \sim \text{Pois}(\lambda_0)$ based on $u_j$ and the inverse of the distribution function for $j = 1, \ldots, p$, where $\lambda_0 = 3$.
\item Gaussian case: We generate $\bm f(x_i) \sim N(0, R)$.
\item Uniform case: We generate the dependent variable $f_j(x_i) \sim \text{Unif}([0,1])$ using a Gaussian copula with the correlation matrix $R$ for $j = 1, \ldots, p$.
\end{enumerate}
For each case, we calculate $\lambda(x_i, \bm \alpha^*)$ as
\begin{equation*}
\lambda(x_i, \bm \alpha^*) = \exp(\bm{\alpha^*}^\top \bm f(x_i)),
\end{equation*}
where $\alpha^*_j = (j-1)/\{10(p-1)\}$ for $j = 1, \ldots, p$. The second case corresponds to the log Gaussian Cox process \citep{Moller1998}.
Then, $m$ presence locations are generated by a multinomial distribution with probabilities $\lambda(x_i, \bm \alpha^*) / \sum_{i=1}^{n} \lambda(x_i, \bm \alpha^*)$.

The performance measures used in these simulation studies are the sum of squared errors $(\hat{\bm\alpha} - \bm\alpha^*)^\top(\hat{\bm\alpha} - \bm\alpha^*)$ and computational costs. These measures are calculated 100 times to compare Maxent, Gamma: $L_\gamma(\bm \alpha)$ with $\gamma=10^{-5}$, GM: $L_{\rm GM}(\bm \alpha)$, ${\rm rGM_{\gamma \approx 0}}$: $L_\gamma^{(2)}(\bm\alpha)$ with $\gamma=10^{-5}$, ${\rm rGM_{\gamma = -0.5}}$: $L_\gamma^{(2)}(\bm\alpha)$ with $\gamma=-0.5$, and Fisher. We omit the results for Gamma with $\gamma=-0.5$ due to its poor estimation accuracy and computational efficiency compared with the case of $\gamma=10^{-5}$. In the simulation studies, we compare the properties of the loss functions themselves, hence we apply no penalty term to all methods ($\tau=0$).

\subsection*{NCEAS data}

Data from the National Centre for Ecological Analysis and Synthesis (NCEAS) have been openly released by \cite{Elith2020}, which include presence-only (PO) and presence-absence (PA) data across six regions: Australian Wet Tropics (AWT), Canada (CAN), New South Wales (NSW), New Zealand (NZ), South America (SA), and Switzerland (SWI). The dataset encompasses records for 226 species, including birds (in AWT, CAN, NSW), bats (in NSW), plants (in AWT, NSW, NZ, SA, SWI), and reptiles (in NSW). The number of presence records in PO data ranges from 2 to 5822, and the number of environmental variables ranges from 11 to 21, including counts of levels for factor variables. For more details on the datasets, see \cite{Elith2020, Valavi2022}.

We use the NCEAS data to determine the optimal values of the tuning parameter $\tau$ for our proposed methods. Rather than determining $\tau$ from scratch, we establish the optimal multipliers based on the results of Maxent as detailed in \cite{Phillips2008}. Additionally, the performance of our proposed methods is evaluated using the same variables that were selected in \cite{Valavi2022}.

\subsection*{Japanese vascular plants data}
We analyzed data from 709 species of vascular plants from the Japanese Biodiversity Mapping Project (J-BMP) database. The study area $\mathcal A$ is divided into 4684 grid cells, each with a resolution of $10\times 10$ km, covering all presence locations for these species \citep{Kubota2015}. To minimize spatial clumping, duplications within each cell are removed, resulting in a range of presence locations from 4 (for {\it Cirsium chokaiense}) to 3077 (for {\it Pteridium aquilinum}). The environmental variables used to model the intensity function are sourced from Mesh Climate Data 2000 and the Bioclim database, as detailed in Table S. \ref{37vari}. Each variable is standardized to fall within the [0,1] range, aligning with the requirements of the Maxent algorithm.

We evaluate the performance of Maxent and our proposed methods based on several criteria: the training Area Under the Curve (AUC) calculated using presence and pseudo-absence locations, the test AUC using presence-absence locations as used in \cite{Kubota2015}, relative computational costs based on Fisher's method, and the number of environmental variables selected. Additionally, we compare habitat maps produced by the estimated intensity functions from both Maxent and our methods. The penalty terms $\tau$ for Maxent are set to the default values for linear features \citep{Phillips2008}, while those for our methods are determined based on the test AUC results from NCEAS data, which assist in estimating species distributions and calculating training and test AUCs.

\section{Results}
\subsection*{Simulation studies}
Figure \ref{fig1} demonstrates the results from the Poisson case in scenario 1, characterized by parameters $\rho=0.5$, $p=50$, $m=500$, and $n=10,000$. 
Panel (a) shows the estimation accuracy of Maxent and the proposed methods, with Gamma performing the best. It is important to note that with a smaller sample size ($m=500$), the likelihood-based method (Maxent) is not always the best in terms of squared errors. We also confirmed that the performances of Maxent and Gamma almost coincide as $m$ increases. The performance of GM is the worst in terms of squared errors primarily because it lacks a regularization term, resulting in $\hat{\bm\alpha}_{\rm GM}$ being much larger than the true value $\bm\alpha^*$. However, $\hat{\bm\alpha}_{\rm GM}$ is almost proportional to $\bm\alpha^*$, which leads to high values of training AUC (not shown here). Panel (b) shows relative computational costs compared with Fisher. Fisher is approximately 100 times faster than GM, 400 times faster than ${\rm rGM_{\gamma\approx0}}$, 3,000 times faster than ${\rm rGM_{\gamma=-0.5}}$, 10,000 times faster than Gamma, and 20,000 times faster than Maxent. Gamma is slightly faster than Maxent due to requiring fewer iterations to converge.

Panels (c) and (d) demonstrate results from the Poisson case with no correlation among $f_j(x)$'s ($\rho=0$). In this scenario, each environmental variable is independent and $p$ is large, resulting in $\bm\alpha^\top \bm f(x_i)$ being approximately normally distributed due to the central limit theorem. Hence, the second-order approximation methods such as ${\rm rGM_{\gamma\approx0}}$, ${\rm rGM_{\gamma=-0.5}}$, and Fisher perform well, as seen in panel (c). Fisher is approximately 10,000 times faster than Maxent, yet their performances measured by squared errors are at almost the same level. This trend is also observed in the Gaussian case and uniform case in Figures S\ref{fig_box05_g} and S\ref{fig_box05_u} in Appendix \ref{appendG}.

Overall, Gamma excels in terms of estimation accuracy; Fisher is the most efficient computationally with good estimation accuracy; ${\rm rGM_{\gamma\approx0}}$ and ${\rm rGM_{\gamma=-0.5}}$ are intermediate between the two methods.

\subsection*{NCEAS data}
The optimal values of $\tau$ for Gamma and ${\rm rGM_{\gamma\approx0}}$ are the same as those of Maxent. It is noteworthy that Gamma nearly equals Maxent because $\gamma=10^{-5}$, and ${\rm rGM_{\gamma\approx0}}$ is the second-order approximation of Gamma. The optimal values of $\tau$ for other methods are one-tenth of that for Maxent, suggesting that GM, Fisher, and ${\rm rGM_{\gamma=-0.5}}$ function effectively with less $L_1$ penalty compared to Maxent. Further details are provided in Figure S\ref{fig_multi} in Appendix \ref{appendH}.

The test AUCs for six regions are illustrated in Figure \ref{fig_elith}. For AWT, the best performing methods are Gamma and ${\rm rGM_{\gamma\approx0}}$. For CAN, NSW, NZ, and SWI, ${\rm rGM_{\gamma=-0.5}}$ shows the best performances. Maxent performs best for NSW and SA. Surprisingly, Fisher also yields the best results for CAN and shows comparable performance in other regions.

The relative computational costs are depicted in Figure S\ref{fig_elith_time}. Fisher is approximately 1000 times faster than Maxent, and ${\rm rGM_{\gamma=-0.5}}$ is about 10 times faster than Maxent. These results confirm the improvements in computational efficiency brought about by the cumulant-based approximation (CBA) in this dataset.

\subsection*{Japanese vascular plants data analysis}

Figure \ref{fig_habitat} illustrates the habitat maps produced by Maxent and our proposed methods for {\it Pteridium aquilinum}, a fern species that is distributed worldwide and noted for its high biological activity \citep{Vetter2009}. The maps show that regions in eastern Japan, including Tokyo (Kanto district), the Southwest Island (Kyushu district), and a central part of Japan (Kinki district), are estimated to be abundant. Conversely, the Hokkaido and Tohoku districts in northeast Japan are estimated to have lower abundance.

Figure \ref{fig_coef} illustrates the estimated coefficients $\bm{\hat{\alpha}}$ for Maxent and our proposed methods. Variables that significantly impact the habitat include potential evapotranspiration (PET), land area, bio\_5 (maximum temperature of the warmest month), and sunshine. PET, indicative of water loss through transpiration and evaporation, suggests water requirements. Large land areas may experience various disturbances, such as fires, storms, or human activities, which could favor the bracken fern, known for its ability to quickly colonize disturbed soils \citep{Stewart2007}. The negative impact of soil pH on {\it Pteridium aquilinum}, also noted in Southwest Asia \citep{Amouzgar2020}, underscores environmental influences on fern distribution.

Figure \ref{fig_path} traces the coefficients' paths as the $L_1$-penalty term $\tau$ varies. Interestingly, Maxent, GM, and ${\rm rGM_{\gamma=-0.5}}$ show similar trends, as illustrated by the estimated coefficients in Figure \ref{fig_coef}. Influential variables such as land area, PET, bio\_5, radiation, and pH remain significant even with a large $\tau$. Conversely, panel (d) shows that Fisher's coefficients paths shrink to 0 even when $\tau$ is small, highlighting a substantial difference from Maxent. This discrepancy suggests that the normality assumption for $\hat{\bm \alpha}_{\rm Fisher}^\top \bm f(x)$ in Corollary \ref{coro1} does not hold, as indicated by a Lilliefors' test for normality p-value of less than $0.001$ \citep{Dallal1986}.

Table \ref{Vascular_result79(ALL)).txt} summarizes the training and test AUC, relative computational costs, and the number of selected variables for 709 vascular plants analyzed using Maxent and our proposed methods. ${\rm rGM_{\gamma=-0.5}}$ achieves the highest training and test AUC, and is nearly seven times more computationally efficient than Maxent. Fisher is the most efficient, being nearly 200 times faster in parameter estimation than Maxent, albeit with slightly inferior accuracy. GM ranks highest in terms of model simplicity, followed by Gamma and Fisher. More detailed results for vascular plant data are listed in the tables in Appendix \ref{appendI}.

\section{Discussion}
In Bayesian inference, variational Bayes is used to enhance computational efficiency by maximizing the variational free energy to approximate the posterior distribution based on the variational density \citep{Friston2007}. This density is simplified under the independence assumption of parameters in a mean field approach combined with the Laplace approximation. The latter provides a local Gaussian approximation for the posterior distribution and aids in evaluating the marginal likelihood using a maximum a posteriori (MAP) estimate.

Unlike traditional Bayesian approaches that focus on reducing computational costs, our proposed method approximates the calculation of a normalizing constant using the cumulants of the logarithm of the intensity function. This approximation technique is also explored in \cite{Takenouchi2017}, employing an empirical localization trick to limit the summation range needed for the normalizing constant. In SDM contexts, such approximations resemble random sampling or target-group sampling of background data, where only a subset of total locations is utilized to estimate species distributions \citep{Phillips2008}. Our CBA method differs in that it becomes exact when the normality assumption on the logarithm of the intensity function holds, aligning with scenarios such as the log Gaussian Cox process \citep{Moller1998}. Furthermore, normality is asymptotically guaranteed even with discrete environmental variables when \(p\) is sufficiently large, as demonstrated in our simulation studies. Empirical results show that Maxent and its second-order approximations (Fisher) as well as ${\rm rGM_{\gamma=-0.5}}$ produce similar habitat maps, as depicted in Figure \ref{fig_habitat} and Figures S\ref{fig_habitatS1}, S\ref{fig_habitatS2}, and S\ref{fig_habitatS3} in Appendix \ref{appendI}. This similarity is quantitatively supported by the low Jeffreys divergence between estimated species distributions by Maxent and other methods, detailed in Table S\ref{Vascular_result79(jd)}. It may be beneficial to use the estimate obtained by Fisher as a starting point for Maxent estimation, similar to using contrastive divergence followed by a maximum likelihood approach to refine solutions \citep{Carreira2005}.

The challenge of sample selection bias in SDM using PO data, which reflects environmental suitability and sampling intensity, is significant. Accurate information on sampling bias is often unavailable \citep{Phillips2009}. One straightforward method is target-group background sampling, where background locations are sampled similarly to PO data \citep{Dudik2005}, although this requires additional information about the target species group observed by similar methods. Another strategy is to use PA data alongside PO data to adjust for sampling bias \citep{Fithian2015}, which also necessitates extra information. Alternatively, \cite{Komori2020} proposes using quasi-linear modeling in Poisson point processes, where variables associated with both species suitability and sampling bias are modeled separately in different clusters, effectively addressing both suitability and bias. Combining quasi-linear modeling with CBA could be a promising approach to manage sampling bias while enhancing computational efficiency.

An important extension for Poisson point processes, or Maxent, defined over the study region $\mathcal A$, is to include observed times of presence data to capture the temporal trends of the intensity function. In this spatio-temporal point process \citep{Diggle2010}, the intensity function's summation extends over locations and observation times, significantly increasing computational demands. The CBA approach thus becomes increasingly crucial for improving computational efficiency in the analysis of spatio-temporal Poisson point processes.

\section*{Conflict of Interest statement}
The authors declare no potential conflict of interest.
\section*{Author Contributions}
OK wrote the first draft of manuscript; OK analyzed the data;
SE and OK proposed the methodology, and YS and YK edited the final
version of the manuscript.
\section*{Acknowledgements}
Financial support was provided by the Japan Society for the
Promotion of Science KAKENHI Grant Number JP22K11938.
\section*{Data Availability}
NCEAS data is available in \cite{Elith2020,Van1974}.
The datasets of Japanese vascular plants are available at Japan Biodiversity Mapping Project (J-BMP) database (https://biodiversity-map.thinknature-japan.com/en/) upon request.

\bibliographystyle{jae} 

\newpage
\begin{figure}[H]
 \begin{minipage}{0.5\hsize}
  \begin{center}
   \includegraphics[width=8cm,height=8cm]{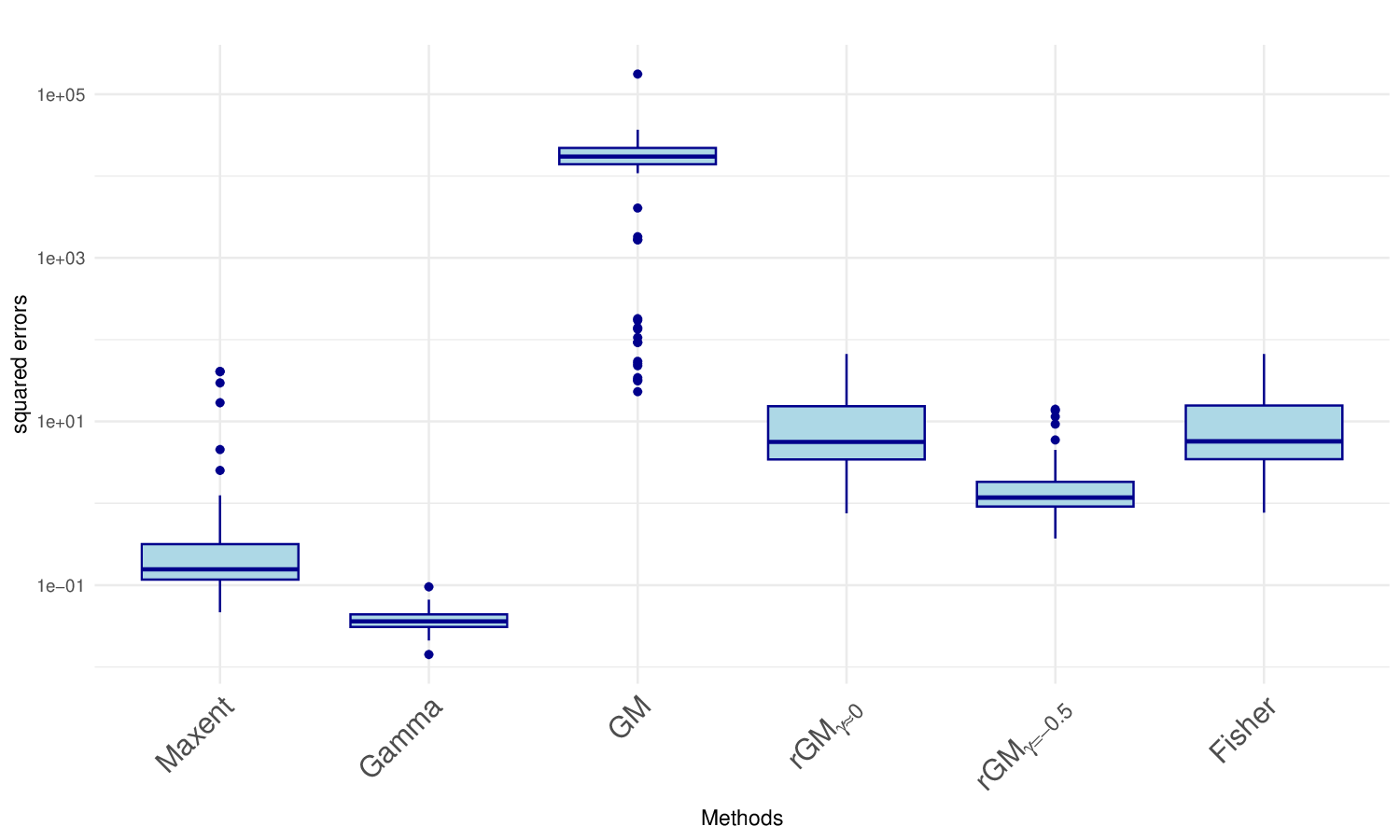}\\
(a) squared errors
  \end{center}
 \end{minipage}
  \begin{minipage}{0.5\hsize}
  \begin{center}
   \includegraphics[width=8cm,height=8cm]{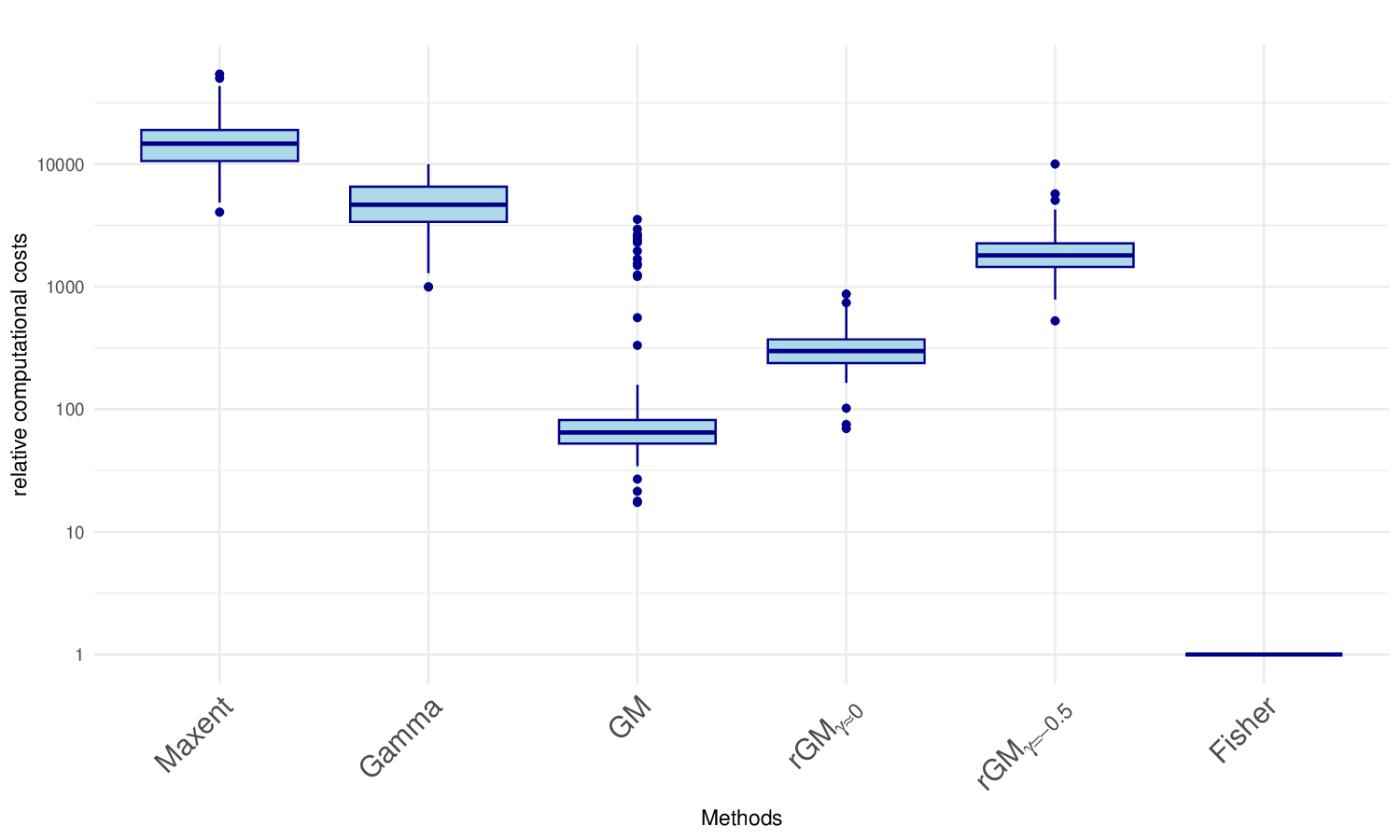}\\
(b) relative computational costs
  \end{center}
 \end{minipage}
  \begin{center}
correlation setting $\rho=0.5$
  \end{center}
  
  \begin{minipage}{0.5\hsize}
  \begin{center}
  \includegraphics[width=8cm,height=8cm]{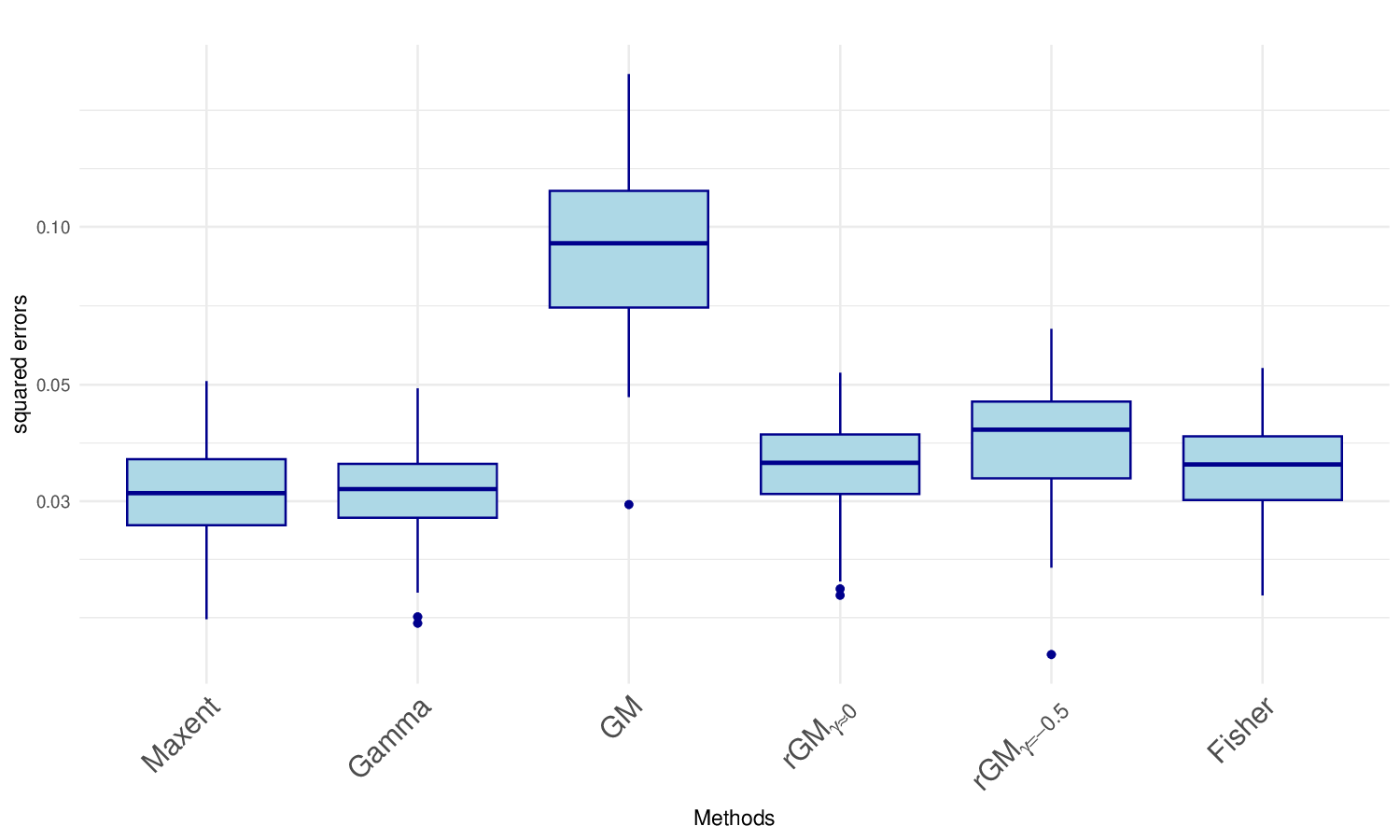}\\
(c) squared errors
  \end{center}
 \end{minipage}
 \begin{minipage}{0.5\hsize}
  \begin{center}
   \includegraphics[width=8cm,height=8cm]{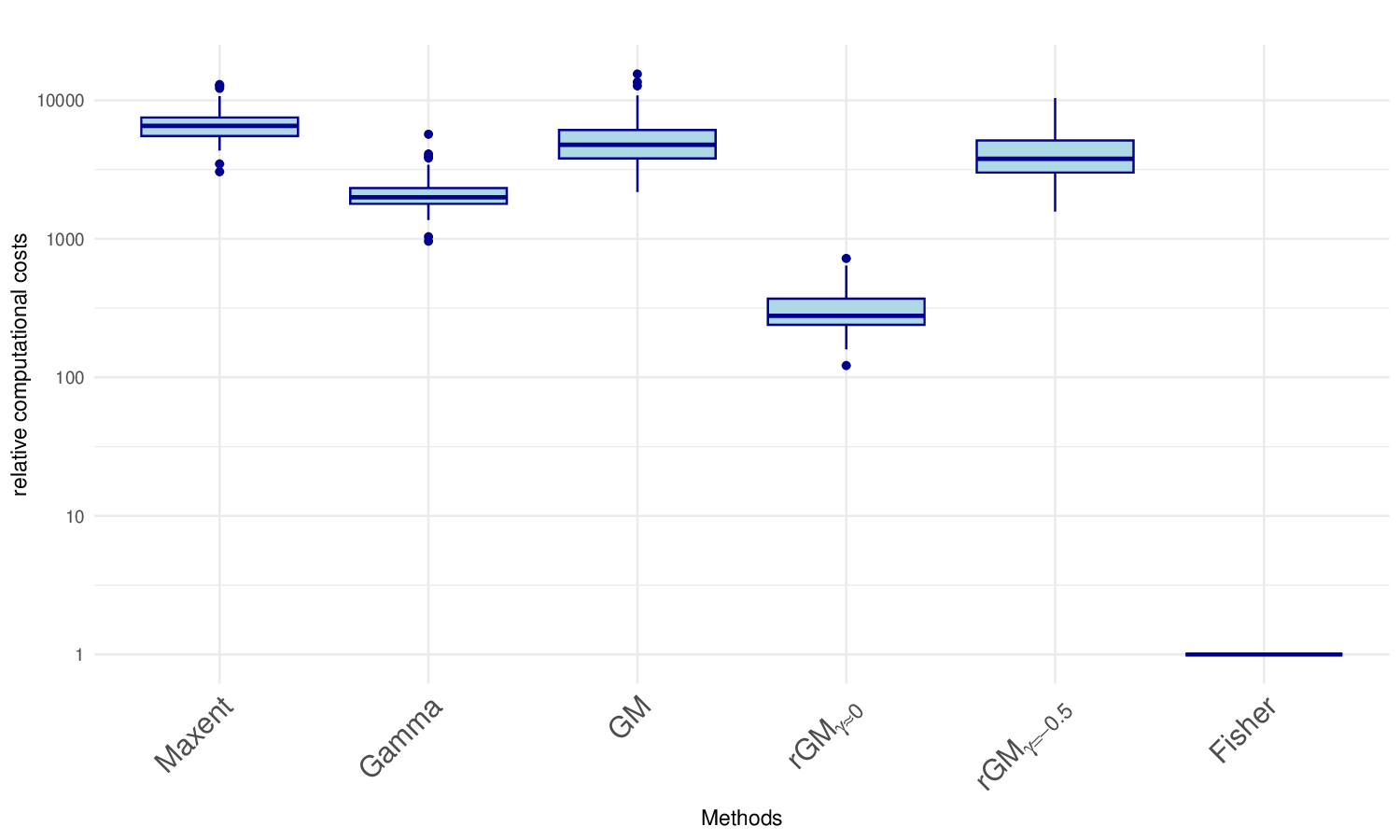}\\
(d) relative computational costs
  \end{center}
 \end{minipage}
 \begin{center}
correlation setting $\rho=0$
  \end{center}
  
\caption{Comparison of performances in Poisson case with $p=50$, $m=500$ and $n=10,000$}\label{fig1}
\end{figure}

\begin{figure}[H]
 \begin{minipage}{0.5\hsize}
  \begin{center}
   \includegraphics[width=8cm,height=6cm]{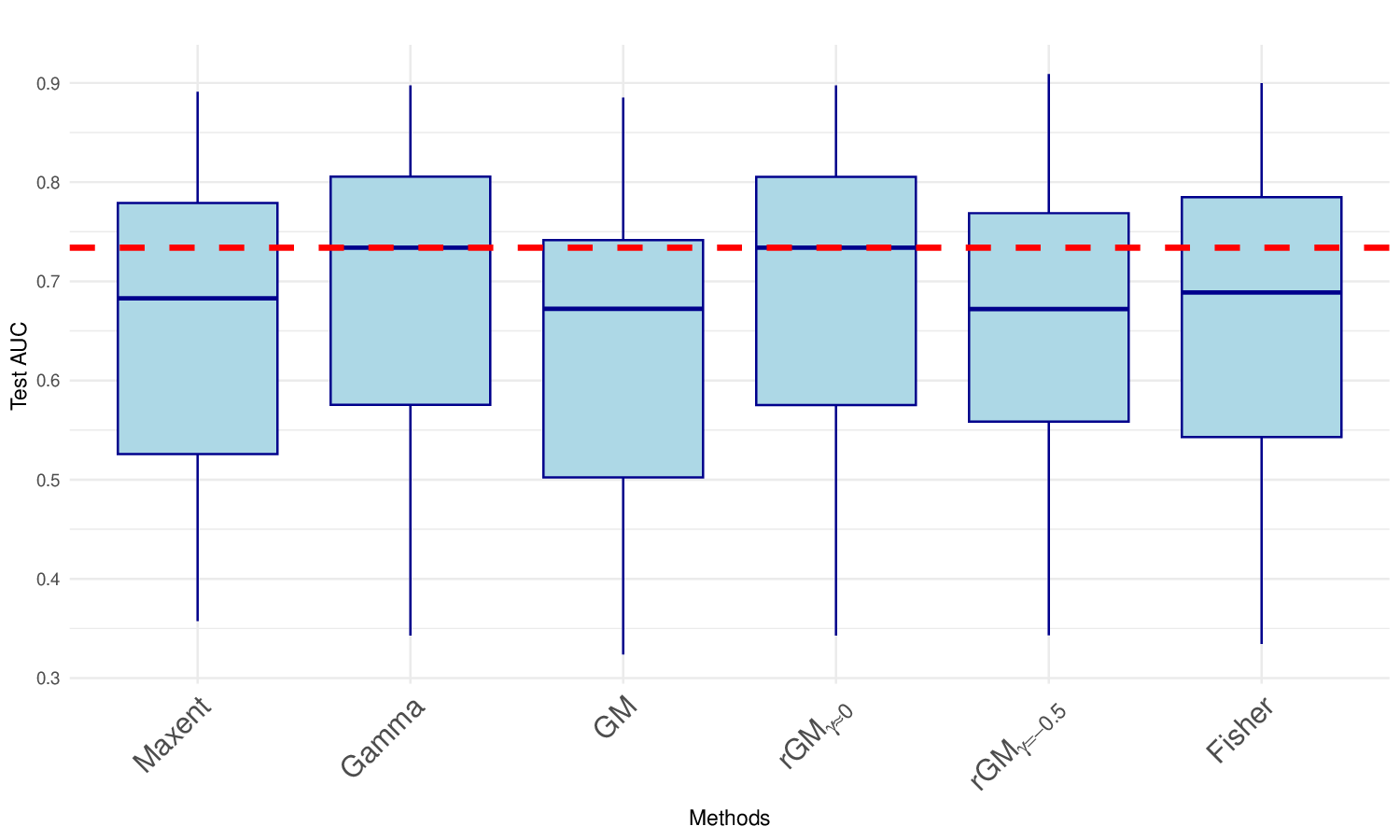}\\
(a) AWT ($p=8$, $9\leq m\leq 265$)
  \end{center}
 \end{minipage}
 \begin{minipage}{0.5\hsize}
  \begin{center}
  \includegraphics[width=8cm,height=6cm]{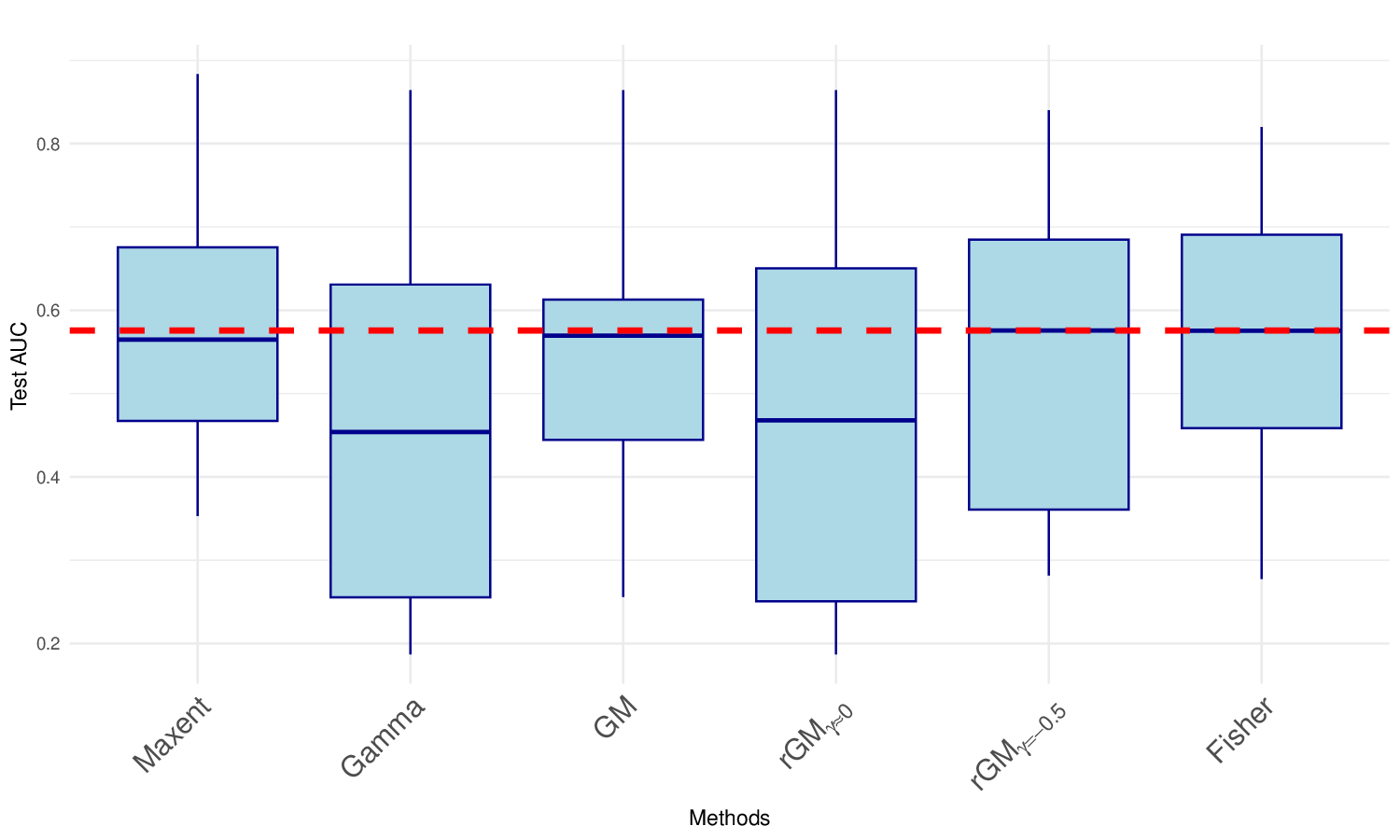}\\
(b)CAN ($p=12$, $16\leq m\leq 740$)
  \end{center}
 \end{minipage}
 \begin{minipage}{0.5\hsize}

 \end{minipage}

 \begin{minipage}{0.5\hsize}
  \begin{center}
   \includegraphics[width=8cm,height=6cm]{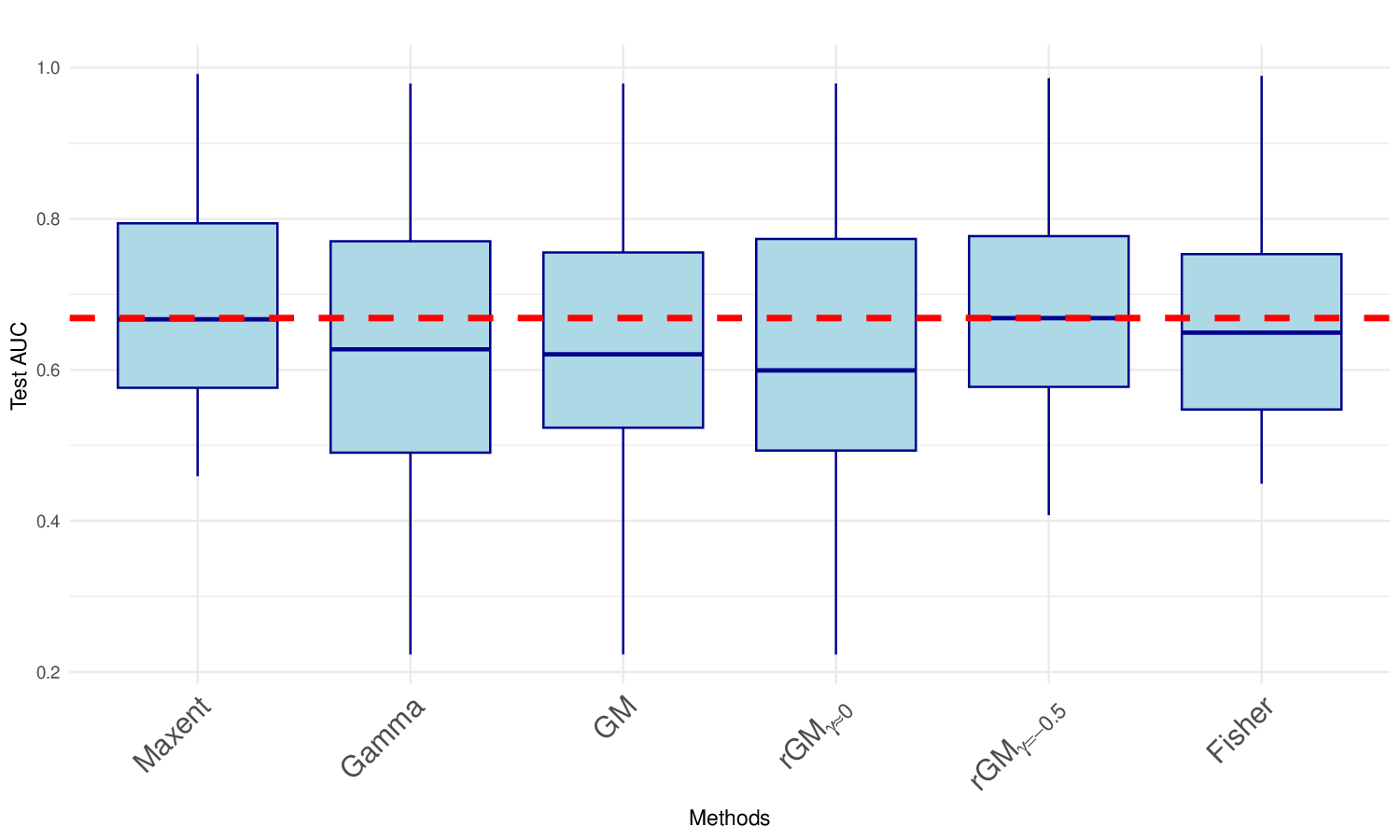}\\
(c) NSW ($p=20$, $2\leq m \leq 426$)
  \end{center}
 \end{minipage}
  \begin{minipage}{0.5\hsize}
  \begin{center}
   \includegraphics[width=8cm,height=6cm]{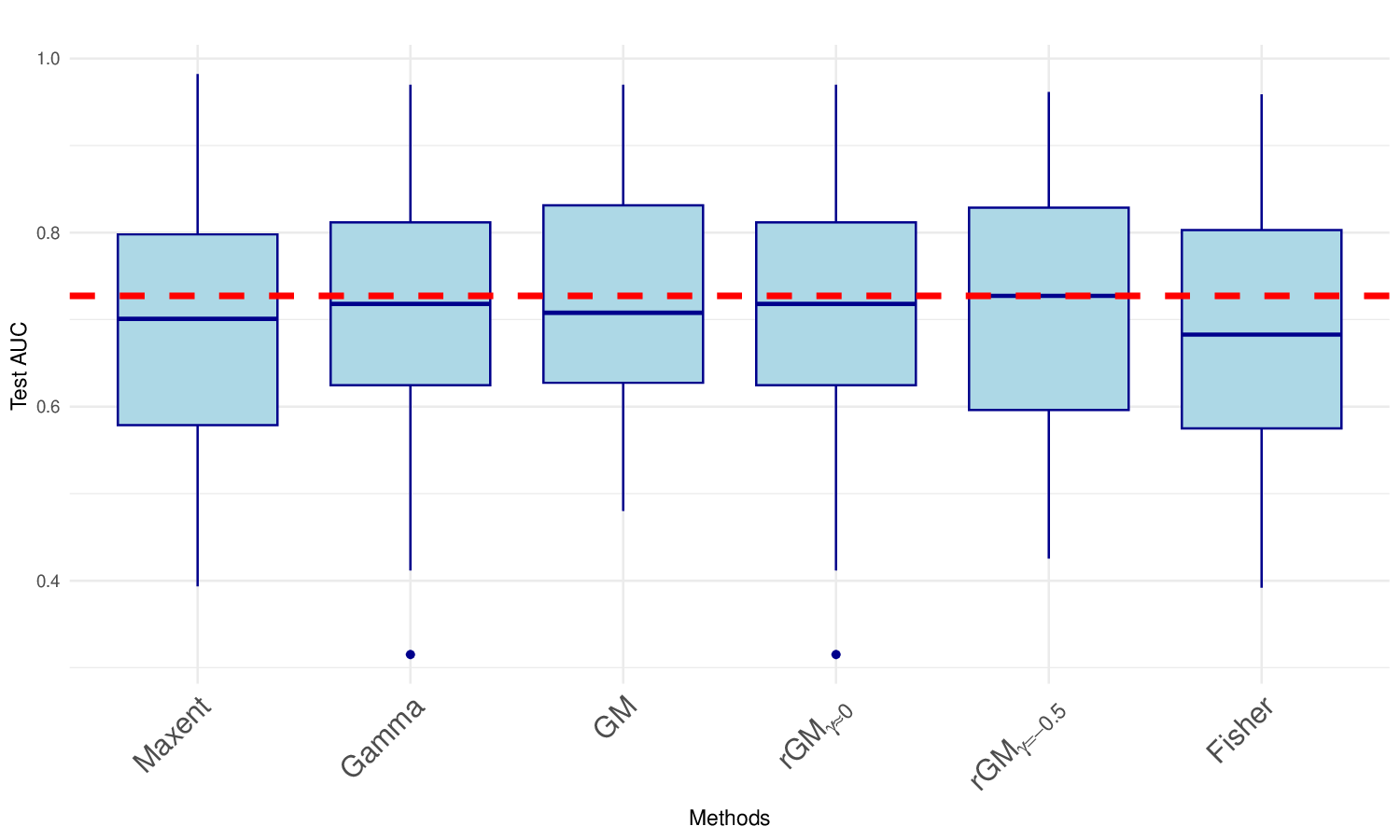}\\
(d) NZ ($p=15$, $18\leq m \leq 211$)
  \end{center}
 \end{minipage}

\begin{minipage}{0.5\hsize}
  \begin{center}
   \includegraphics[width=8cm,height=6cm]{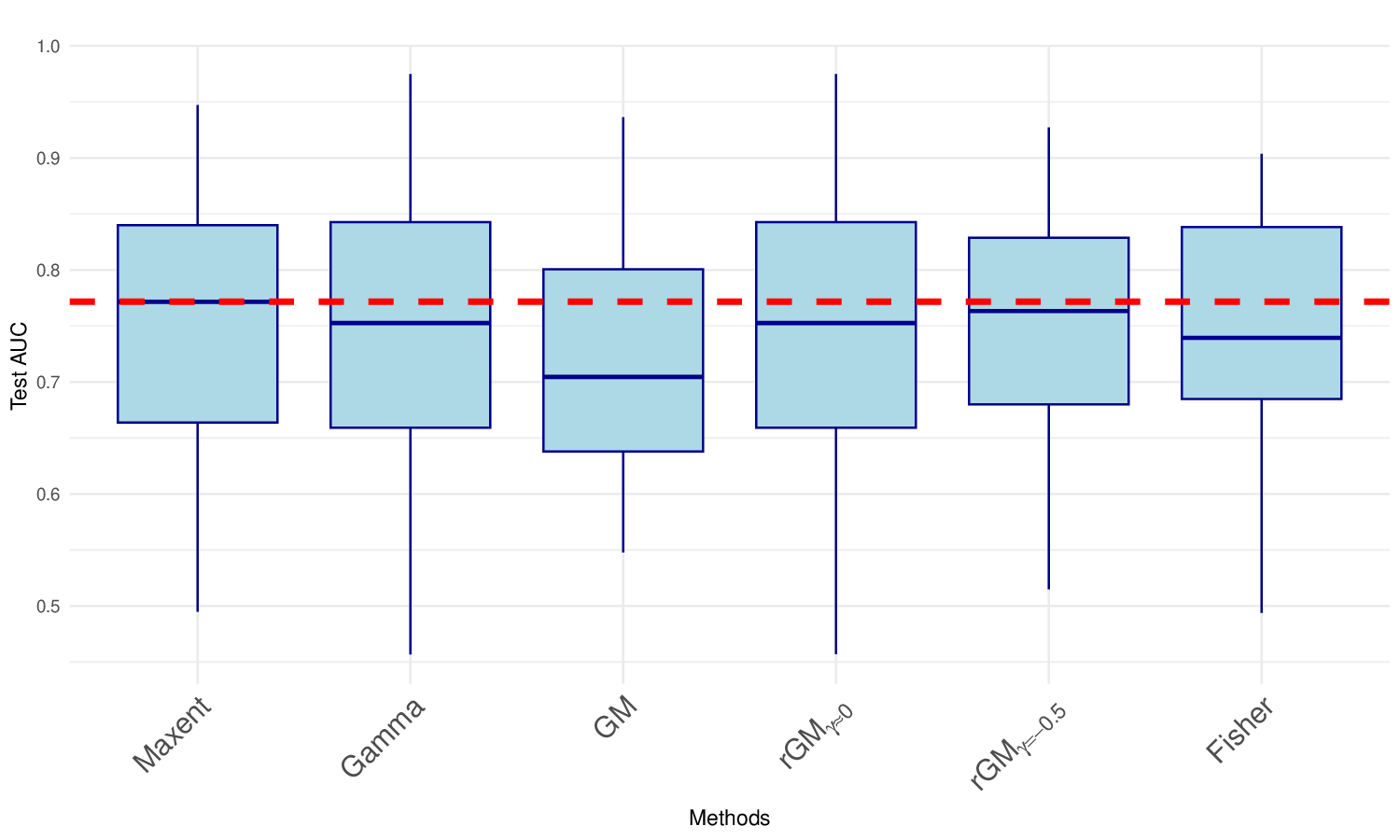}\\
(e) SA ($p=8$, $17\leq m \leq 216$)
  \end{center}
 \end{minipage}
  \begin{minipage}{0.5\hsize}
  \begin{center}
   \includegraphics[width=8cm,height=6cm]{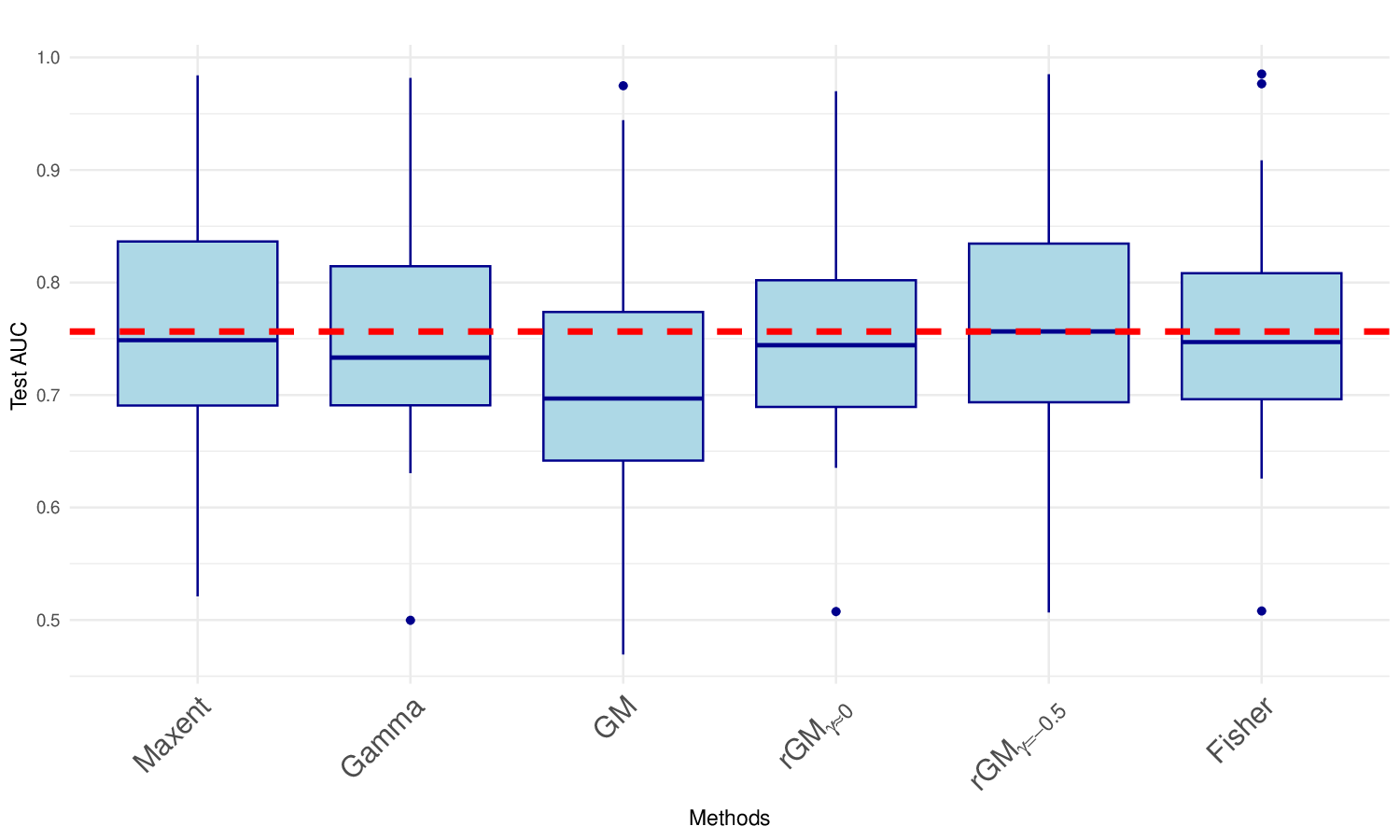}\\
(f) SWI ($p=14$, $36\leq m \leq 5822$)
  \end{center}
 \end{minipage}

\caption{Comparison of performances in terms of test AUC based on NCEAS data with background size $n=10000$}\label{fig_elith}
\end{figure}

\begin{figure}[H]
 \begin{minipage}{0.5\hsize}
  \begin{center}
   \includegraphics[width=8cm,height=10cm]{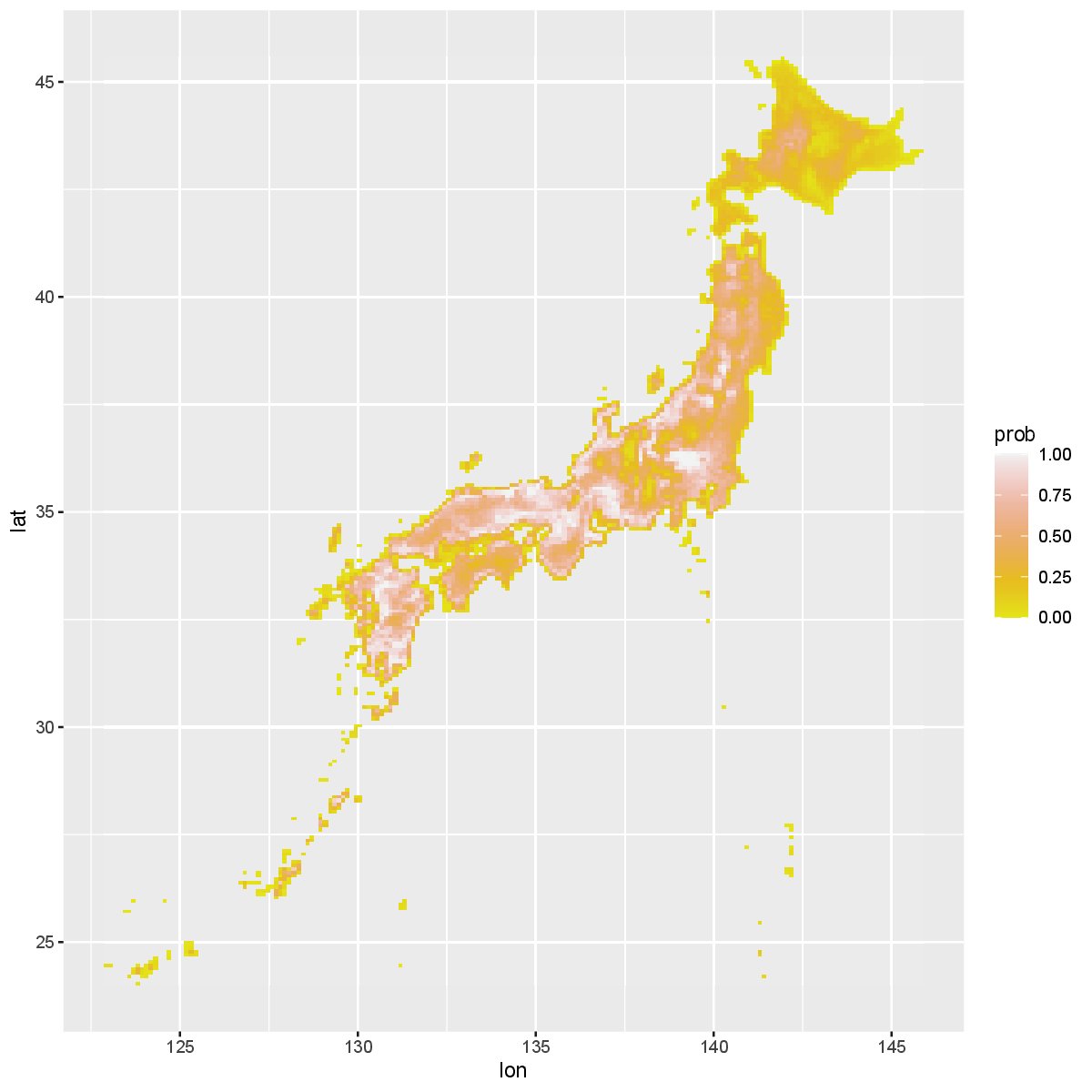}\\
(a) Maxent
  \end{center}
 \end{minipage}
 \begin{minipage}{0.5\hsize}
  \begin{center}
   \includegraphics[width=8cm,height=10cm]{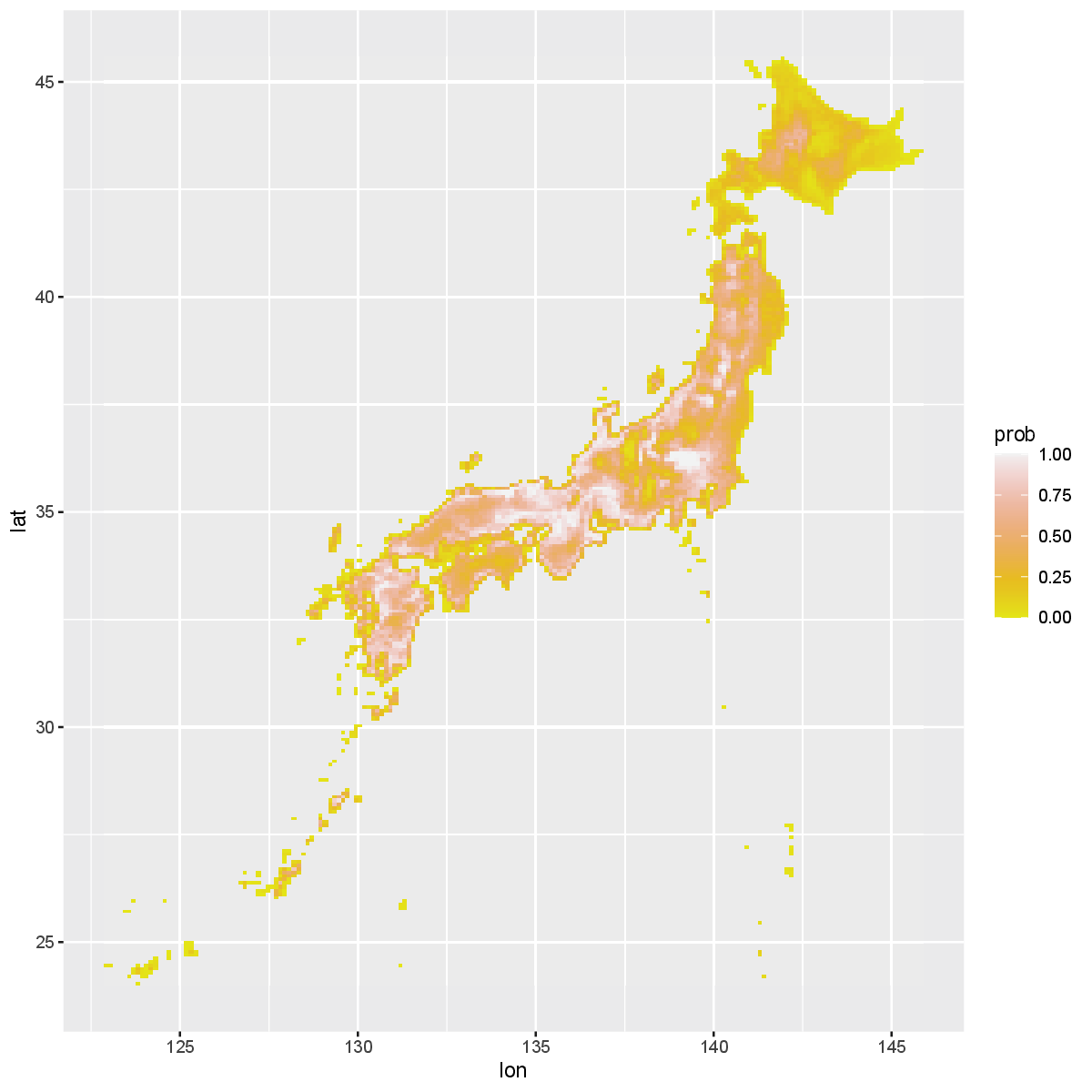}\\
(b) $\rm{rGM_{\gamma=-0.5}}$
  \end{center}
 \end{minipage}
 \begin{minipage}{0.5\hsize}

 \end{minipage}

 \begin{minipage}{0.5\hsize}
  \begin{center}
   \includegraphics[width=8cm,height=10cm]{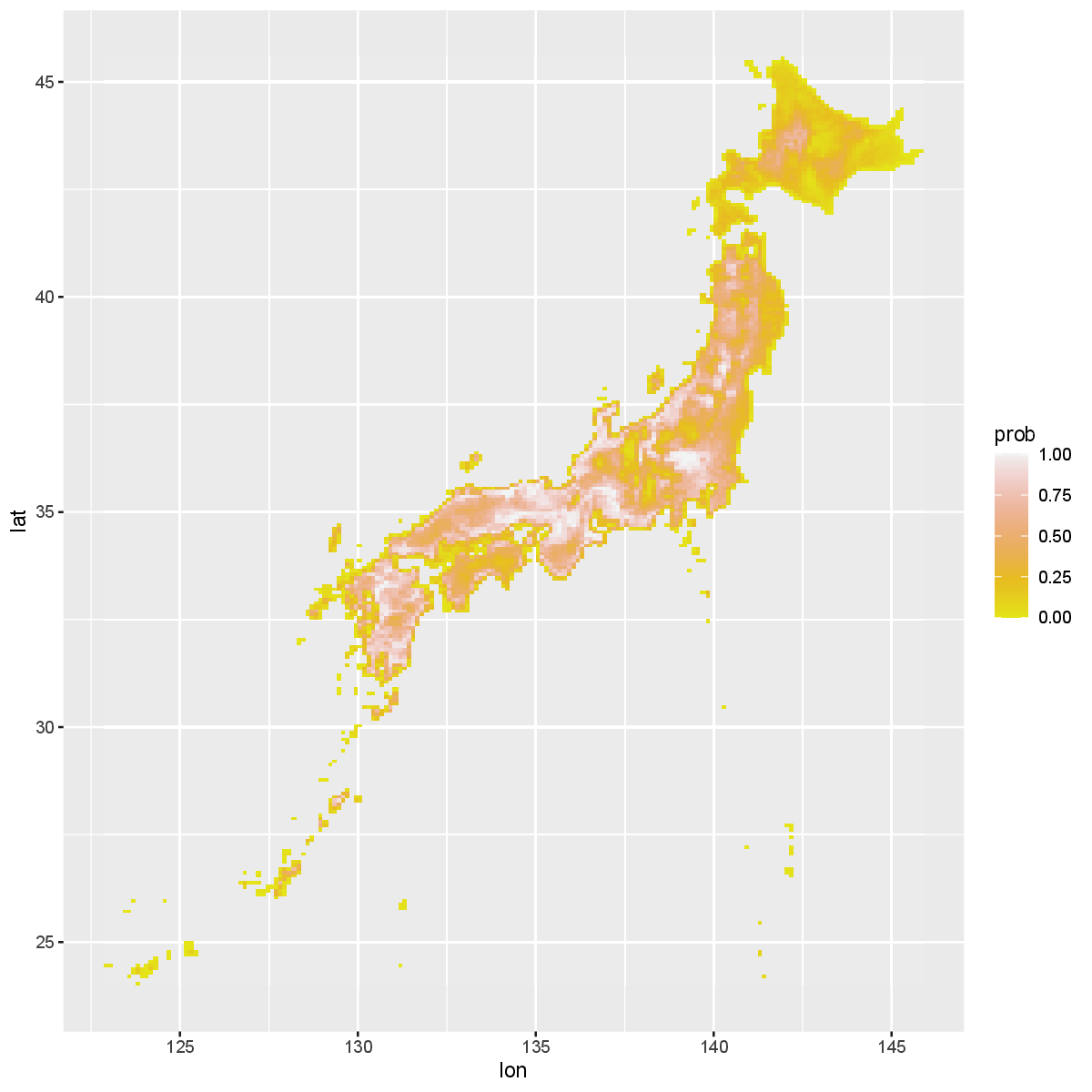}\\
(c) GM
  \end{center}
 \end{minipage}
  \begin{minipage}{0.5\hsize}
  \begin{center}
   \includegraphics[width=8cm,height=10cm]{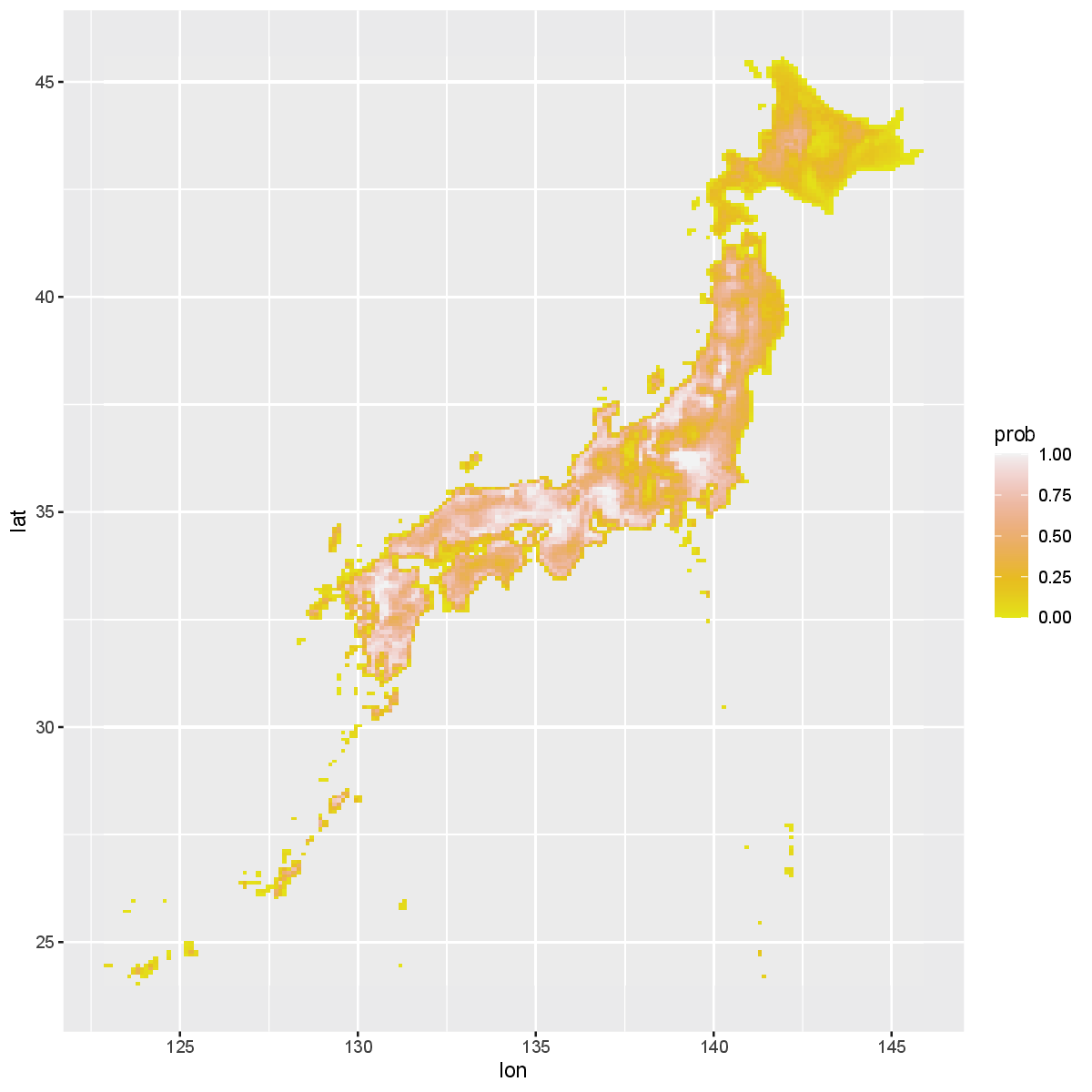}\\
(d) Fisher
  \end{center}
 \end{minipage}

\caption{Estimated habitat maps for {\it Pteridium aquilinum}}\label{fig_habitat}
\end{figure}

\begin{figure}[H]

  \begin{center}
   \includegraphics[width=18cm,height=20cm]{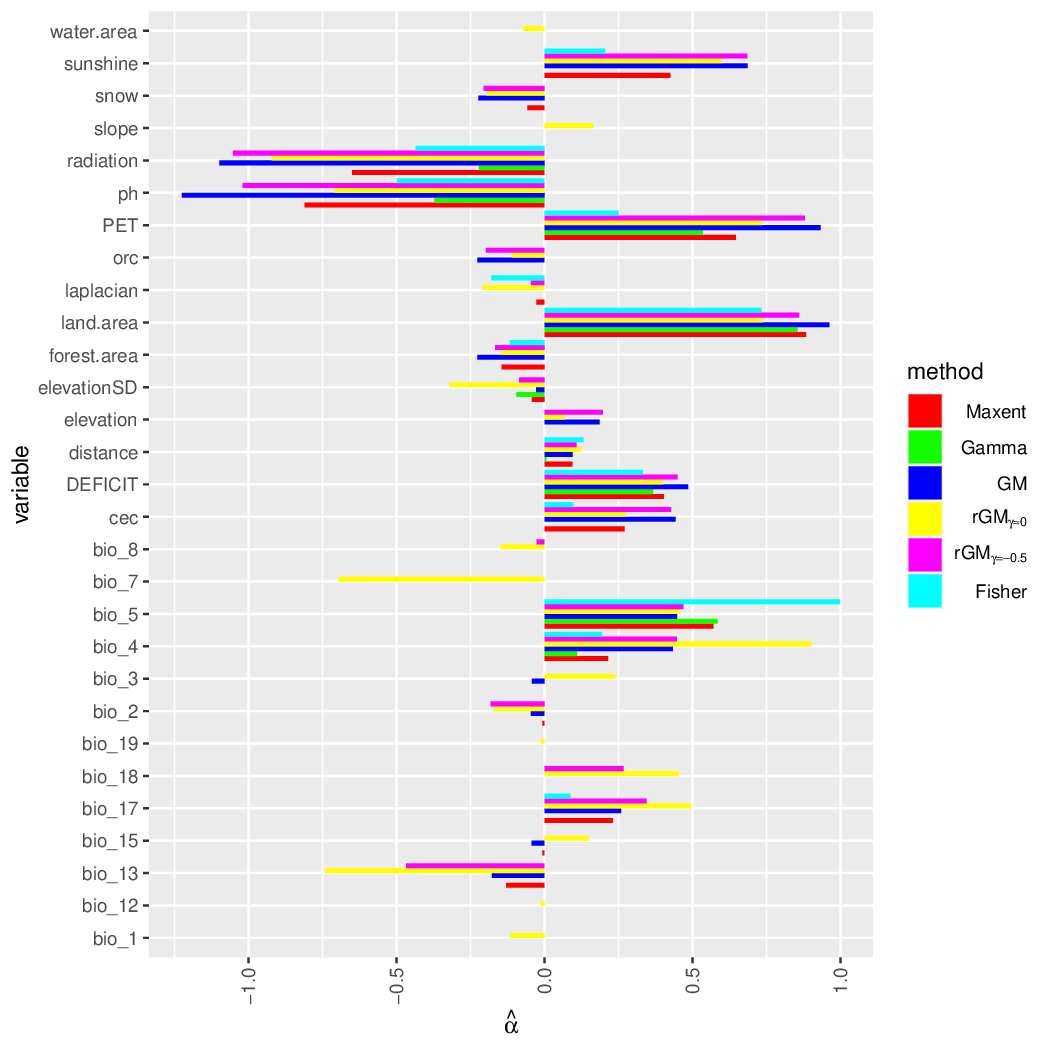}
  \end{center}
 
\caption{Estimated coefficients for {\it Pteridium aquilinum}}\label{fig_coef}

\end{figure}

\begin{figure}[H]
 \begin{minipage}{0.5\hsize}
  \begin{center}
   \includegraphics[width=8cm,height=10cm]{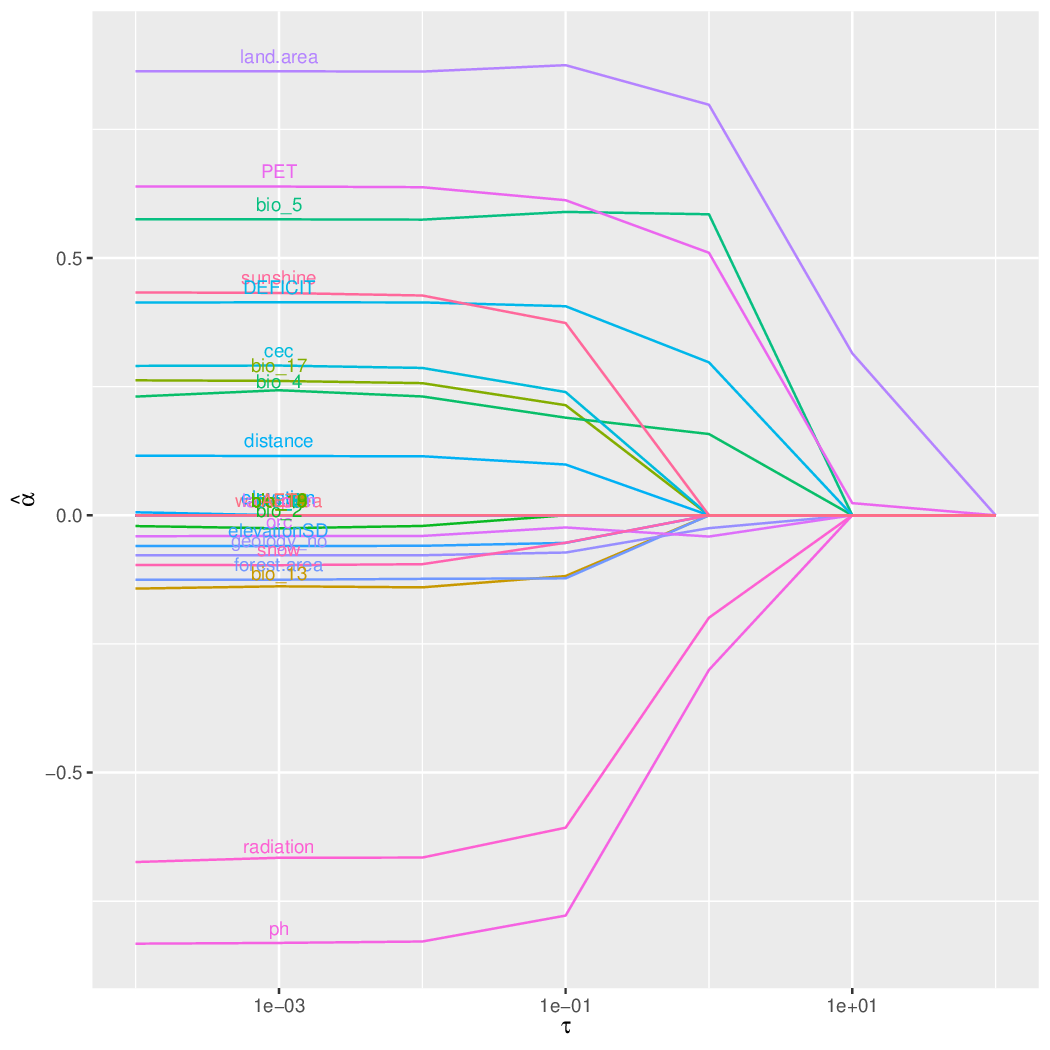}\\
(a) Maxent
  \end{center}
 \end{minipage}
 \begin{minipage}{0.5\hsize}
  \begin{center}
  \includegraphics[width=8cm,height=10cm]{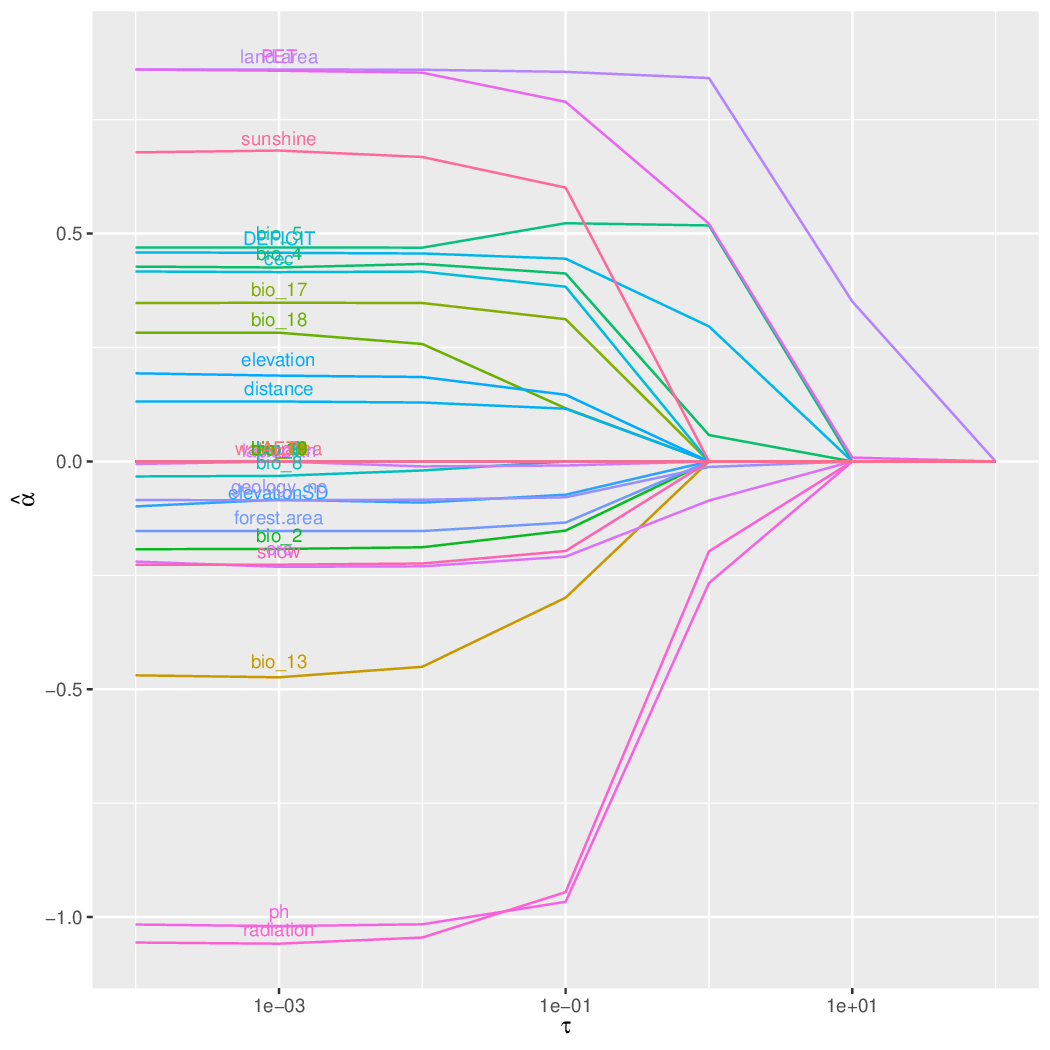}\\
(b)${\rm rGM_{\gamma=-0.5}}$
  \end{center}
 \end{minipage}
 \begin{minipage}{0.5\hsize}

 \end{minipage}

 \begin{minipage}{0.5\hsize}
  \begin{center}
   \includegraphics[width=8cm,height=10cm]{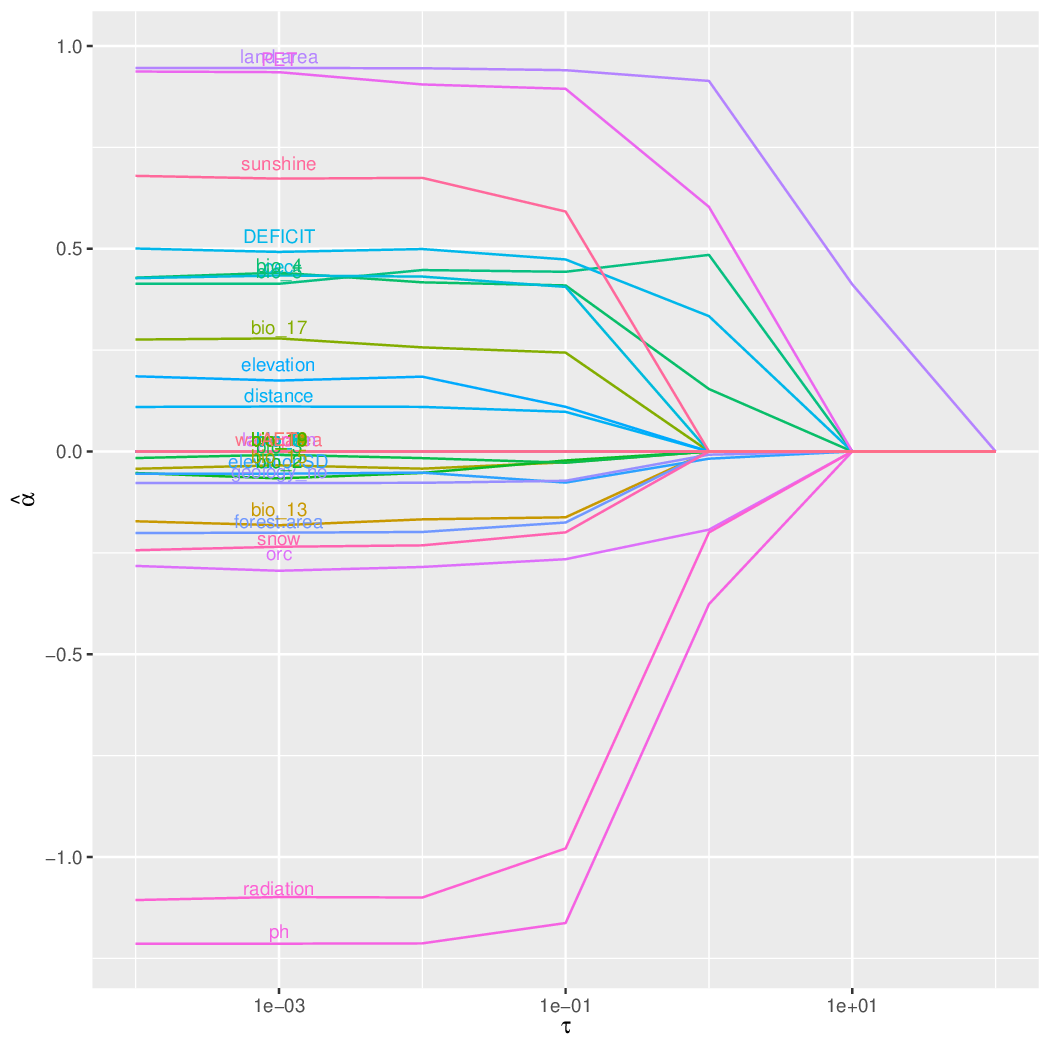}\\
(c) GM
  \end{center}
 \end{minipage}
  \begin{minipage}{0.5\hsize}
  \begin{center}
   \includegraphics[width=8cm,height=10cm]{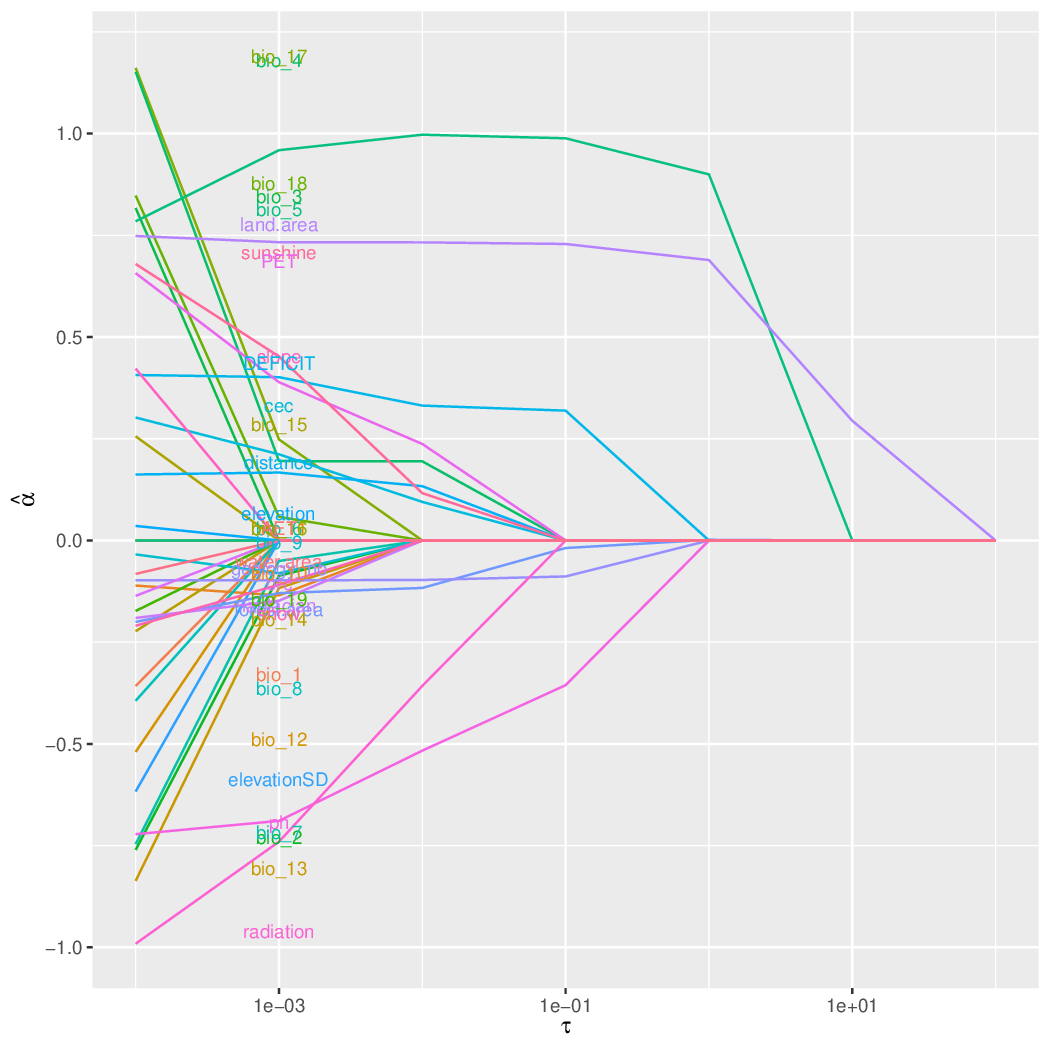}\\
(d) Fisher
  \end{center}
 \end{minipage}

\caption{Estimated coefficients' paths for {\it Pteridium aquilinum}}\label{fig_path}
\end{figure}

\begin{table}[H]
\caption{Mean values of performance measures for Japanese vascular plants\label{Vascular_result79(ALL)).txt}} 
\begin{center}
\begin{tabular}{lllllll}
\hline\hline
\multicolumn{1}{c}{}&\multicolumn{1}{c}{Maxent}&\multicolumn{1}{c}{Gamma}&\multicolumn{1}{c}{GM}&\multicolumn{1}{c}{$\rm{rGM_{\gamma\approx 0}}$}&\multicolumn{1}{c}{$\rm{rGM_{\gamma=-0.5}}$}&\multicolumn{1}{c}{Fisher}\tabularnewline
\hline
training AUC&0.904&0.901&0.833&0.9&\bf{0.905}&0.892\tabularnewline
test AUC&0.863&0.862&0.781&0.858&\bf{0.864}&0.854\tabularnewline
comp. cost&197.6&124&28.5&37.8&39.4&\bf{1}\tabularnewline
No of variables&17.5&14.5&\bf{9.1}&26.5&20.5&15.1\tabularnewline
\hline
\end{tabular}\end{center}
\end{table}

\clearpage
\newpage

\section*{Supporting Information for ``Cumulant-based approximation for fast and efficient prediction for species distribution"}

\author{Osamu Komori\\
Department of Science and Technology, Seikei University\\
 3-3-1 Kichijoji-kitamachi, Musashino, Tokyo 180-8633, Japan
\\
Yusuke Saigusa\\
Department of Biostatistics, School of Medicine, Yokohama City University\\
3-9 Fukuura, Kanazawa, Yokohama, Kanagawa 236-0004, Japan\\
 Shinto Eguchi\\
The Institute of Statistical Mathematics\\
 10-3 Midori-cho, Tachikawa, Tokyo 190-8562, Japan\\
 Yasuhiro Kubota}\\
 Faculty of Science, University of the Ryukyus\\
 1 Senbaru, Nishihara, Okinawa 903-0213, Japan

\appendix
 \renewcommand{\figurename}{Figure.S}
\renewcommand{\tablename}{Table.S}
\setcounter{figure}{0}
\setcounter{table}{0}
\def\thesection{\Alph{section}}

\section{Derivation of equation (\ref{eq8})}\label{appendA}

The estimation equations for PPM and $ L_0(\bm\alpha)$ under the log-linear model are given as
\begin{eqnarray}
\frac{\partial}{\partial\bm\alpha} \ell_{\rm PPM}(\bm\alpha,c)
&=&\sum_{i=1}^n\bm f(x_i)\bigg\{z_i-w_i\exp(\bm\alpha^\top \bm f(x_i)+c)\bigg\}=0\label{ppm}\\
\frac{\partial}{\partial\bm\alpha} L_0(\bm\alpha)
&=&\sum_{i=1}^n\bm f(x_i)\bigg\{z_i-\frac{m w_i}{\Lambda(\bm\alpha)}\exp(\bm\alpha^\top \bm f(x_i))\bigg\}=0,\label{gamma0}
\end{eqnarray}
Comparing with (\ref{ppm}) and (\ref{gamma0}), we have
\begin{eqnarray}
\hat{\bm \alpha}_{\rm PPM}^\top \bm f(x_i)+\hat c_{\rm PPM}=\hat{\bm \alpha}_0^\top \bm f(x_i)+\log \frac{m}{\Lambda(\hat{\bm\alpha}_0)},\ (i=1,\ldots,n).
\end{eqnarray}
Hence we have $\hat{\bm\alpha}_{\rm PPM}=\hat{\bm\alpha}_0$ and $\hat c_{\rm PPM}=\log \{m/\Lambda(\hat{\bm\alpha}_0)\}$.

\section{Relation between $\gamma$-loss and $\beta$-loss functions}\label{appendB}

The relationship between the estimators of PPM and Maxent can be generalized to that between the estimators derived from $\beta$-loss function and $\gamma$-loss function as follows. The $\beta$-loss function for an intensity function $\lambda( x_i,\bm\alpha)$ \citep{Komori2023} is given as
\begin{equation}
 L^{(\beta)}(\bm\alpha,c)=-\sum_{i=1}^m\frac{\lambda(x_i,\bm\alpha,c)^\beta-1}{\beta}+\frac{1}{1+\beta}\sum_{i=1}^nw_i\lambda(x_i,\bm\alpha,c)^{1+\beta},
\end{equation}
where $\beta>-1$ and the estimator of $ L^{(\beta)}(\bm\alpha,c)$ is defined as
\begin{eqnarray}
 (\hat{\bm\alpha}^{(\beta)},\hat{c}^{(\beta)})=\argmin_{\boldmath\alpha,c}  L^{(\beta)}(\bm\alpha,c).
\end{eqnarray}
The $\beta$-loss function is widely employed in the robust parameter estimation \citep{Basu1998,Minami2002,Eguchi2022} as well as estimation of species distribution \citep{Komori2019}. Then we have the following theorem.

\begin{theorem}\label{thm2}
Under an assumption of log-linear model for intensity functions for $\beta$-loss and $\gamma$-loss functions, we have
\begin{eqnarray}\label{eq11}
\hat{\bm\alpha}^{(\gamma)}=\hat{\bm\alpha}_\gamma,\ \hat c^{(\gamma)}=\log\frac{\sum_{i=1}^m\exp\{\gamma\hat{\bm\alpha}_\gamma^\top \bm f(x_i)\}}{\sum_{i=1}^nw_i\exp\{(1+\gamma)\hat{\bm\alpha}_\gamma^\top \bm f(x_i)\}}.
\end{eqnarray}

\end{theorem}
\begin{proof}

Under the log-linear model we have

\begin{align}
 L^{(\beta)}(\bm\alpha,c)=-\sum_{i=1}^m\frac{\exp\{\beta(\bm\alpha^\top \bm f(x_i)+c)\}-1}{\beta}+\frac{1}{1+\beta}\sum_{i=1}^nw_i\exp\{(1+\beta)(\bm\alpha^\top  \bm f(x_i)+c)\}
\end{align}
and 
\begin{align}\label{beta_loss}
\frac{\partial}{\partial\bm\alpha} L^{(\beta)}(\bm\alpha,c)=
-\sum_{i=1}^n\bm f(x_i)\exp(\beta\bm\alpha^\top \bm f(x_i))\{z_i-\exp(c)w_i\exp(\bm\alpha^\top \bm f(x_i))\}=0.
\end{align}
Similarly, we have
\begin{align}\label{gamma_loss}
&\frac{\partial}{\partial\bm\alpha} L_\gamma(\bm\alpha)\nonumber\\
=&\sum_{i=1}^n\bm f(x_i)\exp(\gamma\bm\alpha^\top \bm f(x_i))\bigg\{z_i-\frac{\sum_{i=1}^m\exp(\gamma\bm\alpha^\top \bm f(x_i))}{\sum_{i=1}^nw_i\exp((\gamma+1)\bm\alpha^\top \bm f(x_i))}w_i\exp(\bm\alpha^\top \bm f(x_i))\bigg\}=0.
\end{align}
Hence comparing with (\ref{beta_loss}) and (\ref{gamma_loss}), we have

\begin{eqnarray}
{\hat{\bm\alpha}^{(\gamma)\top}} \bm f(x_i)+\hat c^{(\gamma)}=
\hat{\bm\alpha}_\gamma^\top \bm f(x_i)+
\log\frac{\sum_{i=1}^m\exp\{\gamma\hat{\bm\alpha}_\gamma^\top \bm f(x_i)\}}{\sum_{i=1}^nw_i\exp\{(1+\gamma)\hat{\bm\alpha}_\gamma^\top \bm f(x_i)\}},
\end{eqnarray}
which leads to
\begin{eqnarray}
\hat{\bm\alpha}^{(\gamma)}=\hat{\bm\alpha}_\gamma,\ \hat c^{(\gamma)}=\log\frac{\sum_{i=1}^m\exp\{\gamma\hat{\bm\alpha}_\gamma^\top \bm f(x_i)\}}{\sum_{i=1}^nw_i\exp\{(1+\gamma)\hat{\bm\alpha}_\gamma^\top \bm f(x_i)\}}.
\end{eqnarray}
\end{proof}
From Theorem \ref{thm2}, when $\beta=\gamma$ we observe that the estimated intensity function derived from $\beta$-loss function, i.e. $\lambda(x_i,\hat{\bm\alpha}^{(\gamma)},\hat c^{(\gamma)})$,  is completely reproduced by the estimator of slope derived from $\gamma$-loss function $\hat{\bm\alpha}_\gamma$ and the number of presence locations $m$. Note that equation (\ref{eq11}) in Theorem \ref{thm2} is reduced to equation (\ref{eq8}) as $\gamma\to0$ because $L^{(0)}(\bm\alpha,c)=-\ell_{\rm PPM}(\bm\alpha,c)$.

\section{Proof of Proposition \ref{prop1}}\label{appendC}

\begin{proof}
Let $\lambda^*(x)=c^*\lambda(x,\bm\alpha^*)$ be a true intensity function of Poisson point process with a positive constant $c^*$, then 
we observe that
the expected loss is written by
\begin{eqnarray}
{\mathbb E}\{L_\gamma(\bm\alpha)\}
&=&- \frac{1}{\gamma}\left[\frac{{\mathbb E}\Big[\sum_{i=1}^m 
  { \lambda(x_i,\bm\alpha)^\gamma}\Big]}
{\Big\{\sum_{i=1}^n  w_i\lambda(x_i ,\bm\alpha) ^{\gamma+1}\Big\}^{\frac{\gamma}{\gamma+1}}}-{\mathbb E}[m]\right]\label{sum}\\
&=&- \frac{c^*}{\gamma}\Bigg[\frac{\sum_{i=1}^n
  w_i \lambda(x_i,\bm\alpha)^\gamma\lambda(x_i,\bm\alpha^*)}
{\Big\{\sum_{i=1}^n w_i  \lambda(x_i ,\bm\alpha) ^{\gamma+1}\Big\}^{\frac{\gamma}{\gamma+1}}}-\sum_{i=1}^n w_i \lambda(x_i,\bm\alpha^*)\Bigg].
\end{eqnarray}
We note that equation \eqref{sum} is derived by the formula of random sums \citep{Streit2010}.
Hence we have
 \begin{eqnarray}\nonumber
&&\hspace{-10mm}{\mathbb E}\{L_\gamma(\bm\alpha)\}-{\mathbb E}\{L_\gamma(\bm\alpha^*)\}\\[6mm]\label{beta}
&&\hspace{-15mm}=- \frac{c^*}{\gamma}\left[\frac{\sum_{i=1}^nw_i
  { \lambda(x_i,\bm\alpha)^\gamma}\lambda(x_i,\bm\alpha^*)}
{\Big\{\sum_{i=1}^nw_i  \lambda(x_i ,\bm\alpha) ^{\gamma+1}\Big\}^{\frac{\gamma}{\gamma+1}}}
-{\Big\{\sum_{i=1}^nw_i  \lambda(x_i ,\bm\alpha^*) ^{\gamma+1}\Big\}^{\frac{1}{\gamma+1}}}\right]
\end{eqnarray}
which is nothing but the $\gamma$-power divergence $c^*D_\gamma(\lambda(\cdot,\bm\alpha^*),\lambda(\cdot,\bm\alpha))$. We know that\\ $D_\gamma(\lambda(\cdot,\bm\alpha^*),\lambda(\cdot,\bm\alpha))\geq0$ with equality if and only if $\bm\alpha=\bm\alpha^*$ as one of basic properties for the $\gamma$-power divergence.

Similarly, we confirm of unbiasedness of the estimating equation as
\begin{eqnarray}\label{est-eq}
&&\mathbb E\bigg[\frac{\partial}{\partial\bm\alpha}L_\gamma(\bm\alpha)\bigg]\\
&&\propto\mathbb E\bigg[
- { \sum_{i=1}^m { \lambda(x_i,\bm\alpha) ^\gamma} \frac{\partial}{\partial\boldmath\alpha}\log  \lambda(x_i ,\bm\alpha) }
+{ \sum_{i=1}^m { \lambda(x_i,\bm\alpha) ^\gamma} }
\frac{\sum_{j=1}^n w_j\lambda(x_j,\bm\alpha) ^{\gamma+1} \frac{\partial}{\partial\boldmath\alpha}\log  \lambda(x_j ,\bm\alpha) }{\sum_{j=1}^nw_j  \lambda(x_j ,\bm\alpha) ^{\gamma+1}}\bigg],\nonumber
\end{eqnarray}
which vanishes when the expectation is taken by a Poisson point process with the intensity function $\lambda(x,\bm\alpha)$. 
\end{proof}

\section{Proof of Theorem \ref{thm4}}\label{appendD}

\begin{proof}
From (\ref{g_loss}), we have
\begin{eqnarray}
 L_\gamma(\bm\alpha)
=- \frac{1}{\gamma}\left[\frac{\displaystyle\sum_{i=1}^m 
{ \lambda(x_i,\bm\alpha)^\gamma }}{ \Big\{\displaystyle\sum_{i=1}^n w_i  \lambda(x_i ,\bm\alpha) ^{\gamma+1}\Big\}^{\frac{\gamma}{\gamma+1}}}-m\right].
\end{eqnarray}
The logarithm of the denominator becomes  
\begin{eqnarray}
\frac{\gamma}{\gamma+1}\log\bigg[\sum_{i=1}^n w_i  \lambda(x_i ,\bm\alpha) ^{\gamma+1}\bigg]
=\frac{\gamma}{\gamma+1} K(\gamma+1)
=\gamma\sum_{j=1}^\infty\kappa_j(\bm\alpha)\frac{(\gamma+1)^{j-1}}{j!}.
\end{eqnarray}
When $\log \lambda(x_i,\bm\alpha)=\bm\alpha^\top \bm f(x_i)$, we have
\begin{equation}
 L_\gamma(\bm\alpha)
=
- \frac{1}{\gamma}\Bigg[\sum_{i=1}^m\exp\bigg\{\gamma\bm\alpha^\top \bm f(x_i)- \gamma\sum_{j=1}^\infty\kappa_j(\bm\alpha)\frac{(\gamma+1)^{j-1}}{j!}\bigg\}-m\Bigg].
\end{equation}

\end{proof}

\section{Consistency of estimating function of $L_\gamma^{(2)}(\bm\alpha)$}\label{appendE}

We check the expected estimating function under normality assumption $\bm f(X)\sim {\cal N}(\bar {\bm f},S)$ for presence locations. This leads to $\bm\alpha^\top\bm f \bm(X)\sim {\cal N}(\bm\alpha^\top \bar {\bm f},\bm\alpha^\top S\bm\alpha).$  
Note that
\begin{equation}
\int \exp\{(\gamma+1)\bm\alpha^\top\bm f\}d {\cal N}(\bar {\bm f},{S})(\bm f)
=\exp\{(\gamma+1)\bm\alpha^\top \bar {\bm f}+\half(\gamma+1)^2  \bm\alpha^\top {S} \bm\alpha\},
\end{equation}
which leads to
\begin{eqnarray}\nonumber
\frac{\partial}{\partial\bm\alpha}L_\gamma^{(2)}(\bm\alpha)
=
- \sum_{i=1}^m\exp\big\{\gamma[\bm\alpha^\top( \bm f(x_i)-  \bar {\bm f})-\half(\gamma+1)  \bm\alpha^\top {S} \bm\alpha]\big\} \big(\bm f(x_i)- {\bar {\bm f}}-(\gamma+1){S}\bm\alpha\big).
\end{eqnarray}
Hence, we get
\begin{eqnarray}\nonumber
\mathbb E_{\boldmath\alpha}\Big[\frac{\partial}{\partial\bm\alpha}L_\gamma^{(2)}(\bm\alpha)\Big]
&\propto& \int \exp\{(\gamma+1)\bm\alpha^\top\bm f\}\big(\bm f- {\bar {\bm f}}-(\gamma+1){S}\bm\alpha\big)d {\cal N}(\bar {\bm f},{S})(\bm f)
\\[3mm]&\propto& \int \big(\bm f- {\bar {\bm f}}-(\gamma+1){S}\bm\alpha\big)
d {\cal N}(\bar {\bm f}+(\gamma+1){S}\bm\alpha,{S})(\bm f),
\end{eqnarray}
where $\mathbb E_{\boldmath\alpha}$ is the expectation regarding Poisson point process with a intensity $\exp(\bm\alpha^\top \bm f)$.
This becomes a zero vector since
the mean vector of the normal distribution ${\cal N}(\bar {\bm f}+(\gamma+1){S}\bm\alpha,{S})$ is equal to $\bar {\bm f}+(\gamma+1){S}\bm\alpha$.
Therefore, we conclude that $\frac{\partial}{\partial\boldmath\alpha}L_\gamma^{(2)}(\bm\alpha)$
is unbiased.

\section{Details of estimating algorithm for the proposed methods}\label{appendF}

For $\bm\alpha'=(\alpha_1,\ldots,\alpha_j+\delta,\ldots,\alpha_p)^\top$ and $a_j\leq f_j(x_i)-\bar{f_j}\leq b_j$, we have
\begin{eqnarray}
   L_{\gamma,\tau}^{(2)}(\bm\alpha')-L_{\gamma,\tau}^{(2)}(\bm\alpha)
    \leq
    D_j(\delta,\bm\alpha)+\tau_j(|\alpha_j+\delta|-|\alpha_j|)
    \equiv D_j^\tau(\delta,\bm\alpha)\label{ineq},  
\end{eqnarray}
 where $\tau_j=\tau\sigma_j/\sqrt{m}$ and 
 \begin{align}
 &D_j(\delta,\bm\alpha)\\
 =&\exp\left[-\gamma\bm\alpha^\top \bm{\bar f}-\frac{\gamma(\gamma+1)}{2}\bm\alpha^\top S\bm\alpha-\{\gamma\bar{f_j}+\gamma(\gamma+1)\bm s_j^\top\bm\alpha-\frac{\gamma(\gamma+1)}{2}s_{jj}\delta\}\delta\right]\nonumber\\
 &\hspace{1cm}\times\sum_{i=1}^m\exp\{\gamma\bm\alpha^\top\bm f(x_i)\}\left[\frac{\exp(\gamma b_j\delta)-\exp(\gamma a_j\delta)}{b_j-a_j}f_j(x_i)+\frac{b_j\exp(\gamma a_j\delta)-a_j\exp(\gamma b_j\delta)}{b_j-a_j}\right]\nonumber\\
 &\hspace{1cm}-L_\gamma^{(2)}(\bm\alpha)\\
=&\exp\left[-\frac12\gamma\{2\bm\alpha^\top\bm{\bar f}+(1+\gamma)\alpha^\top S\alpha\}+\frac{\{b_j\gamma-\bar f_j\gamma-\gamma(1+\gamma)\bm s_j^\top \bm \alpha\}^2}{2\gamma(1+\gamma)s_{jj}}\right]\nonumber\\
&\times\frac{\sum_{i=1}^m\exp(\gamma\bm\alpha^\top \bm{\bar f}(x_i))(f_j(x_i)-a_j)}{b_j-a_j}\exp\left[-\frac12\gamma(\gamma+1)s_{jj}\left\{\delta-\frac{(a_j-\bar f_j)\gamma-\gamma(\gamma+1)\bm s_j^\top \bm \alpha}{\gamma(\gamma+1)s_{jj}}\right\}^2\right]\nonumber\\
+&\exp\left[-\frac12\gamma\{2\bm\alpha^\top\bm{\bar f}+(1+\gamma)\alpha^\top S\alpha\}+\frac{\{a_j\gamma-\bar f_j\gamma-\gamma(1+\gamma)\bm s_j^\top \bm \alpha\}^2}{2\gamma(1+\gamma)s_{jj}}\right]\nonumber\\
 &\times\frac{\sum_{i=1}^m\exp(\gamma\bm\alpha^\top \bm{\bar f}(x_i))(b_j-f_j(x_i))}{b_j-a_j}\exp\left[-\frac12\gamma(\gamma+1)s_{jj}\left\{\delta-\frac{(b_j-\bar f_j)\gamma-\gamma(\gamma+1)\bm s_j^\top \bm \alpha}{\gamma(\gamma+1)s_{jj}}\right\}^2\right]\nonumber\\
 &-L_\gamma^{(2)}(\bm\alpha)\nonumber,
 \end{align}
 and $s_{jj}$ is the $j$th diagonal element of $S$ and $\bm s_j$ is the $j$th column vector of $S$.
 The inequality (\ref{ineq}) comes from $\exp(\gamma\delta f_j(x_i))\leq\{\exp(\gamma\delta b_j)-\exp(\gamma\delta a_j)\}/(b_j-a_j)f_j(x_i)+\{b_j\exp(\gamma a_j\delta)-a_j\exp(\gamma b_j\delta)\}/(b_j-a_j)$ for $a_j\leq f_j(x_i)\leq b_j$. Note that the calculation of $D_j(\delta,\bm\alpha)$ is efficient because it has order of $m$ time complexity rather than locations of the entire region $n$. This technique is also employed in the estimating algorithm of Maxent \citep{Dudik2004}. The soft thresholding operator \citep{Goeman2010} is defined as
 \begin{eqnarray}
&&\frac{\partial}{\partial\delta}D_j^\tau(\delta,\bm\alpha)\\
&=&\left\{
\begin{array}{ll}\displaystyle
\frac{\partial}{\partial\delta}D_j(\delta,\bm\alpha)+\tau_j{\rm sign}(\alpha_j+\delta) & if\ \alpha_j+\delta\not=0 \\
\displaystyle
\frac{\partial}{\partial\delta}D_j(\delta,\bm\alpha)-\tau_j{\rm sign}\left(\frac{\partial}{\partial\delta}D_j(\delta,\bm\alpha)\right) & if\ \alpha_j+\delta=0\ and\ |\frac{\partial}{\partial\delta}D_j(\delta,\bm\alpha)|>\tau_j \\
0&otherwise
\end{array}
\right.
\end{eqnarray}
 Hence $\delta^*$ is calculated using one-step Newton-Raphson method as
 \begin{equation}
\delta^*=-\frac{\frac{\partial}{\partial\delta}D_j^\tau(0,\bm\alpha)}{\frac{\partial^2}{\partial\delta^2}D_j(0,\bm\alpha)}.
 \end{equation}
 Note that $\delta^*$ for GM-loss function is obtained by setting $S=0$ in the above procedure.
 
For $\gamma$-loss function, we consider
\begin{equation}
L_{\gamma}^\tau(\bm\alpha)=\log L_{\gamma}(\bm\alpha)+\frac{\tau\gamma}{\sqrt{m}}\sum_{j=1}^p\sigma_j|\alpha_j|.
\end{equation}
Then we have
\begin{eqnarray}
   L_{\gamma}^\tau(\bm\alpha')-L_{\gamma}^\tau(\bm\alpha)
    \leq
    D_{\gamma,j}(\delta,\bm\alpha)+\tau_j(|\alpha_j+\delta|-|\alpha_j|)
    \equiv D_{\gamma,j}^\tau(\delta,\bm\alpha),  
\end{eqnarray}
 where
 \begin{align}
 &D_{\gamma,j}(\delta,\bm\alpha)\\
=&\log\left[\frac{\exp(\gamma\delta b_j)-\exp(\gamma\delta a_j)}{b_j-a_j}\sum_{i=1}^m\exp(\gamma\bm\alpha^\top \bm f(x_i))f_j(x_i)\right.\nonumber\\
&\hspace{1cm}+\left.\frac{b_j\exp(\gamma\delta a_j)-a_j\exp(\gamma\delta b_j)}{b_j-a_j}\sum_{i=1}^m\exp(\gamma\bm\alpha^\top \bm f(x_i))\right]\nonumber\\
&-\frac{\gamma}{\gamma+1}\log\left[\frac{\exp\{(\gamma+1)\delta b_j\}-\exp\{(\gamma+1)\delta a_j\}}{b_j-a_j}\sum_{i=1}^m w_i\exp\{(\gamma+1)\bm\alpha^\top \bm f(x_i)\}f_j(x_i)\right.\nonumber\\
&\hspace{1cm}+\left.\frac{b_j\exp\{(\gamma+1)\delta a_j\}-a_j\exp\{(\gamma+1)\delta b_j\}}{b_j-a_j}\sum_{i=1}^m w_i\exp\{(\gamma+1)\bm\alpha^\top \bm f(x_i)\}\right]-L_{\gamma}(\bm\alpha).\nonumber
 \end{align}
 The soft thresholding operator for $\gamma$-loss function is defined as
 \begin{eqnarray}
&&\frac{\partial}{\partial\delta}D_{\gamma,j}^\tau(\delta,\bm\alpha)\\
&=&\left\{
\begin{array}{ll}\displaystyle
\frac{\partial}{\partial\delta}D_{\gamma,j}(\delta,\bm\alpha)+\tau_j{\rm sign}(\alpha_j+\delta) & if\ \alpha_j+\delta\not=0 \\
\displaystyle
\frac{\partial}{\partial\delta}D_{\gamma,j}(\delta,\bm\alpha)-\tau_j{\rm sign}\left(\frac{\partial}{\partial\delta}D_{\gamma,j}(\delta,\bm\alpha)\right) & if\ \alpha_j+\delta=0\ and\ |\frac{\partial}{\partial\delta}D_{\gamma,j}(\delta,\bm\alpha)|>\tau_j \\
0&otherwise
\end{array}
\right.
\end{eqnarray}
 Hence $\delta^*$ for $\gamma$-loss function is calculated using one-step Newton-Raphson method as
 \begin{equation}
\delta^*=-\frac{\frac{\partial}{\partial\delta}D_{\gamma,j}^\tau(0,\bm\alpha)}{\frac{\partial^2}{\partial\delta^2}D_{\gamma,j}(0,\bm\alpha)}.
 \end{equation}

 For Fisher, we consider
\begin{equation}
L_{0,\tau}^{(2)}(\bm\alpha)=L_{0}^{(2)}(\bm\alpha)+\frac{\tau\gamma}{\sqrt{m}}\sum_{j=1}^p\sigma_j|\alpha_j|.
\end{equation}
Then we have
\begin{eqnarray}
L_{0,\tau}^{(2)}(\bm\alpha')-L_{0,\tau}^{(2)}(\bm\alpha)
\leq
    D_{0,j}(\delta,\bm\alpha)+\tau_j(|\alpha_j+\delta|-|\alpha_j|)
    \equiv D_{0,j}^\tau(\delta,\bm\alpha),  
\end{eqnarray}
 where
 \begin{align}
 &D_{0,j}(\delta,\bm\alpha)\\
 =&\frac12s_{jj}+\left(\delta+\frac{\bar f_j-\bar f_{m,j}+\bm s_j^\top\bm\alpha}{s_{jj}}\right)^2-\frac{(\bar f_j-\bar f_{m,j}+\bm s_j^\top\bm\alpha)^2}{s_{jj}}+\bm\alpha^\top(\bar{\bm f}-\bar{\bm f}_m)+\frac12\bm\alpha^\top S\bm\alpha-L_{0,\tau}^{(2)}(\bm\alpha)
 \end{align}

 The soft thresholding operator for Fisher is defined as
 \begin{eqnarray}
&&\frac{\partial}{\partial\delta}D_{0,j}^\tau(\delta,\bm\alpha)\\
&=&\left\{
\begin{array}{ll}\displaystyle
\frac{\partial}{\partial\delta}D_{0,j}(\delta,\bm\alpha)+\tau_j{\rm sign}(\alpha_j+\delta) & if\ \alpha_j+\delta\not=0 \\
\displaystyle
\frac{\partial}{\partial\delta}D_{0,j}(\delta,\bm\alpha)-\tau_j{\rm sign}\left(\frac{\partial}{\partial\delta}D_{0,j}(\delta,\bm\alpha)\right) & if\ \alpha_j+\delta=0\ and\ |\frac{\partial}{\partial\delta}D_{0,j}(\delta,\bm\alpha)|>\tau_j \\
0&otherwise
\end{array}
\right.
\end{eqnarray}
 Hence $\delta^*$ for Fisher is calculated using one-step method as
 \begin{equation}
\delta^*=-\frac{\frac{\partial}{\partial\delta}D_{0,j}^\tau(0,\bm\alpha)}{\frac{\partial^2}{\partial\delta^2}D_{0,j}(0,\bm\alpha)}.
 \end{equation}

\section{Additional results of simulation studies}\label{appendG}

\begin{figure}[H]
 \begin{minipage}{0.5\hsize}
  \begin{center}
    \includegraphics[width=8cm,height=8cm]{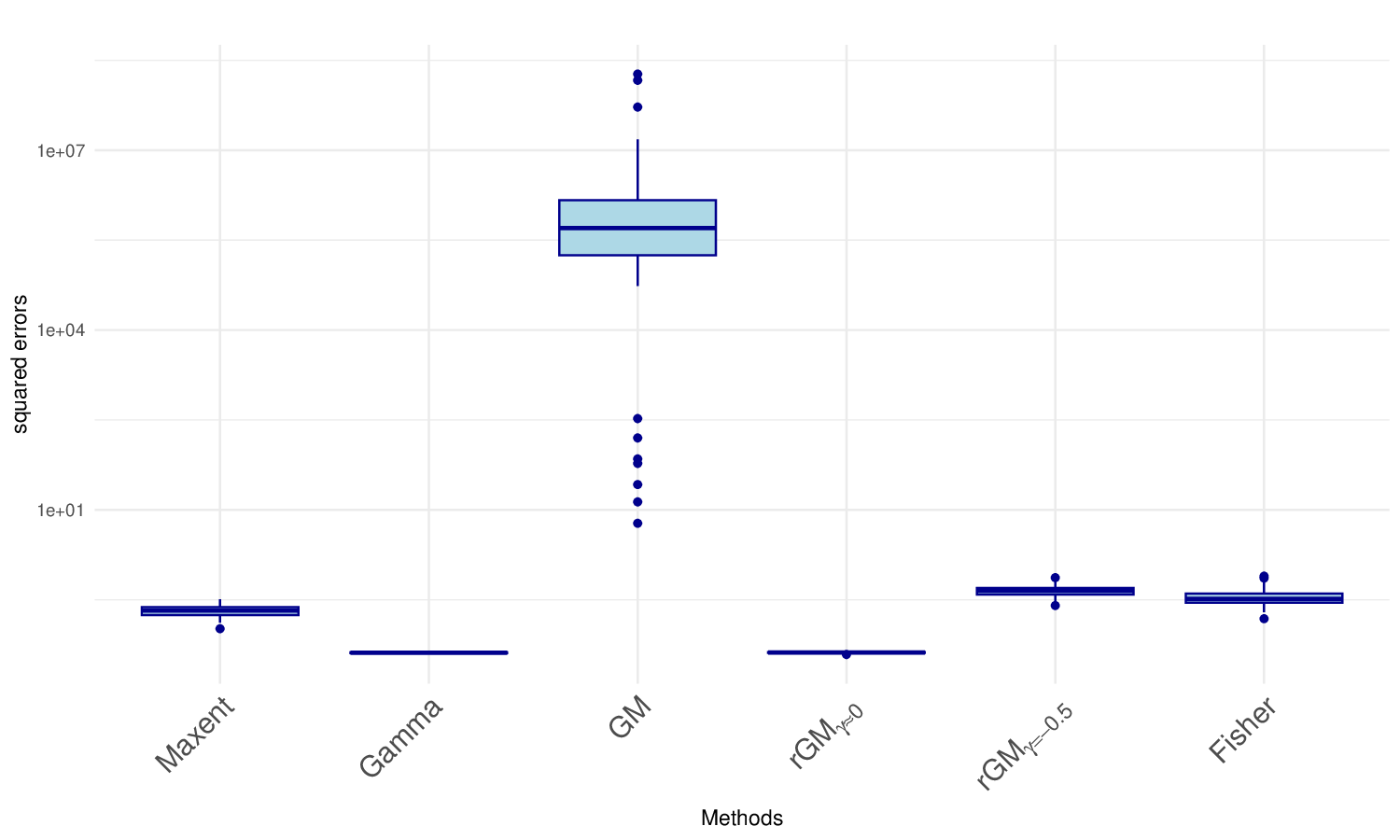}\\
(a) squared errors
  \end{center}
 \end{minipage}
 \begin{minipage}{0.5\hsize}
  \begin{center}
  \includegraphics[width=8cm,height=8cm]{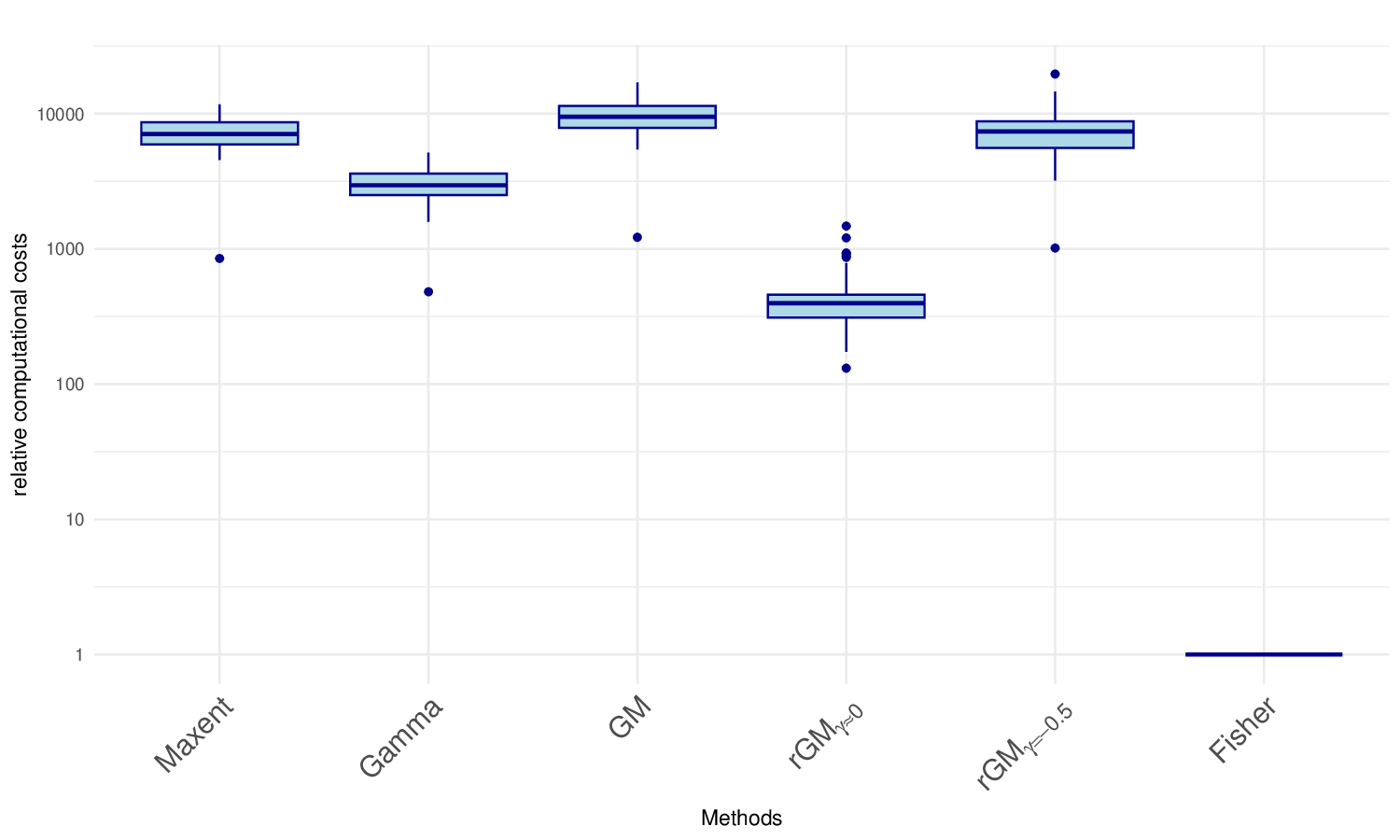}\\
(b) relative computational costs
  \end{center}
 \end{minipage}
 \begin{minipage}{0.5\hsize}

 \end{minipage}
\begin{center}
correlation setting $\rho=0.5$
  \end{center}
 \begin{minipage}{0.5\hsize}
  \begin{center}
  \includegraphics[width=8cm,height=8cm]{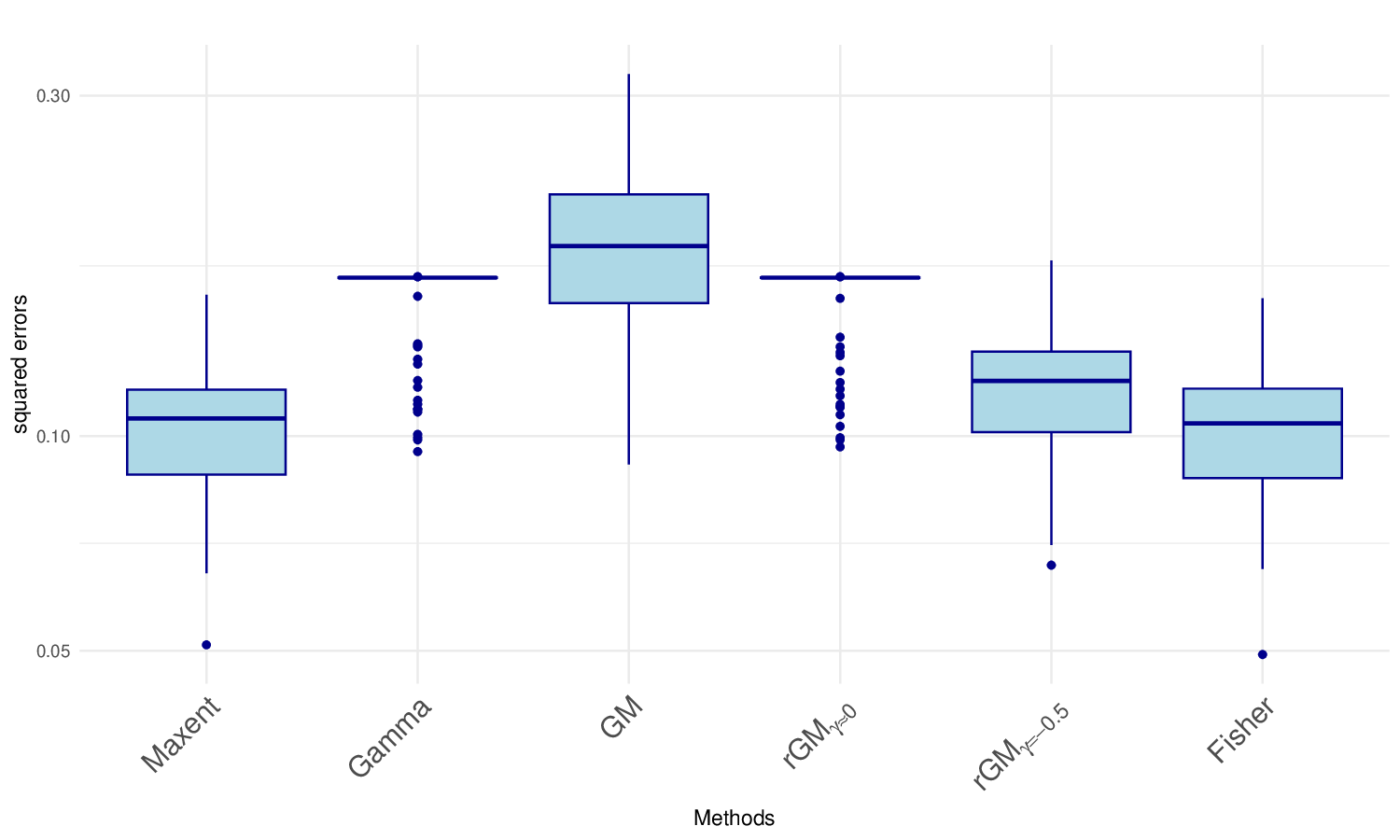}\\
(c) squared errors
  \end{center}
 \end{minipage}
  \begin{minipage}{0.5\hsize}
  \begin{center}
   \includegraphics[width=8cm,height=8cm]{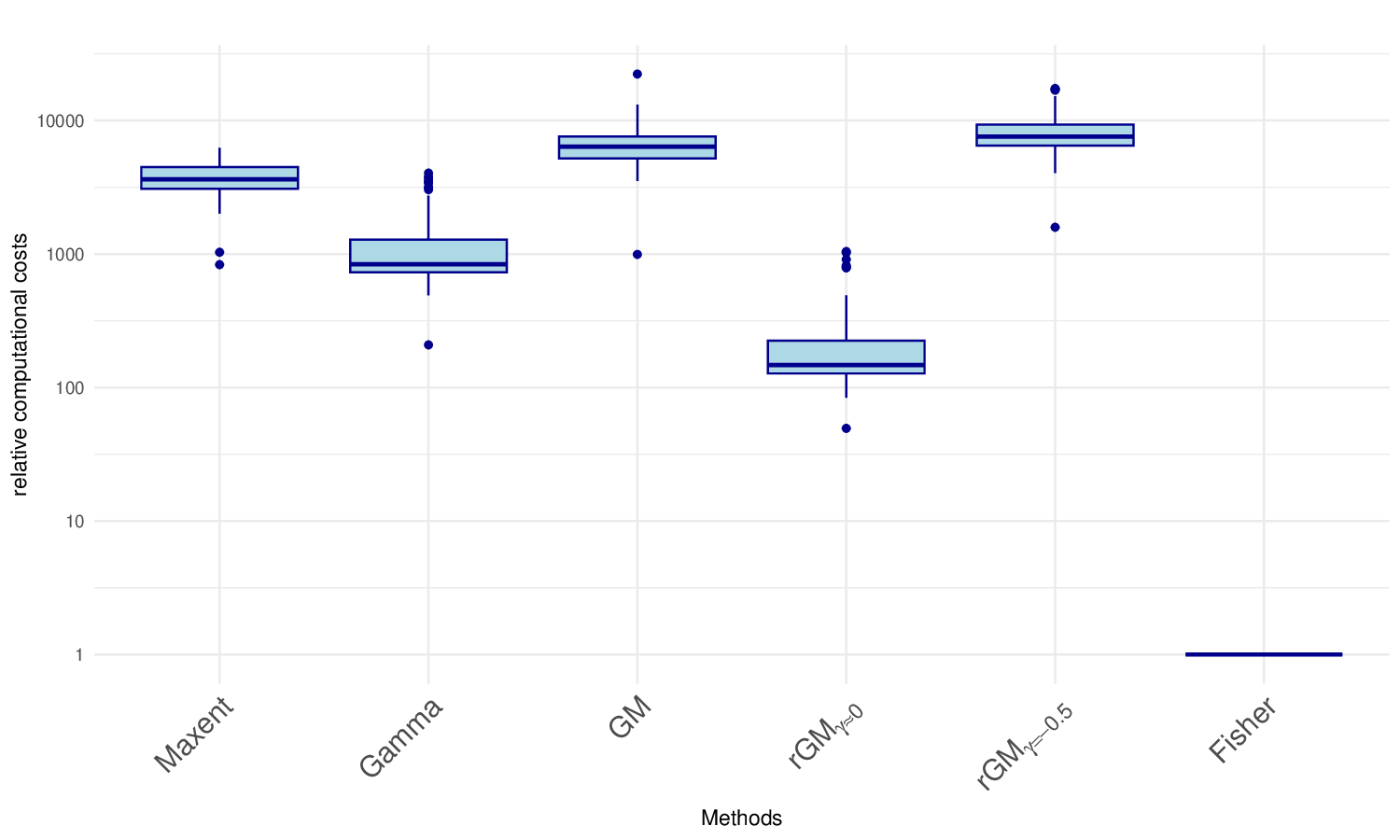}\\
(d) relative computational costs
  \end{center}
 
 \end{minipage}
 \begin{center}
correlation setting $\rho=0$
  \end{center}

\caption{Comparison of performances in Gaussian case with $p=50$, $m=500$ and $n=10000$}\label{fig_box05_g}
\end{figure}

\begin{figure}[H]
 \begin{minipage}{0.5\hsize}
  \begin{center}
   \includegraphics[width=8cm,height=9cm]{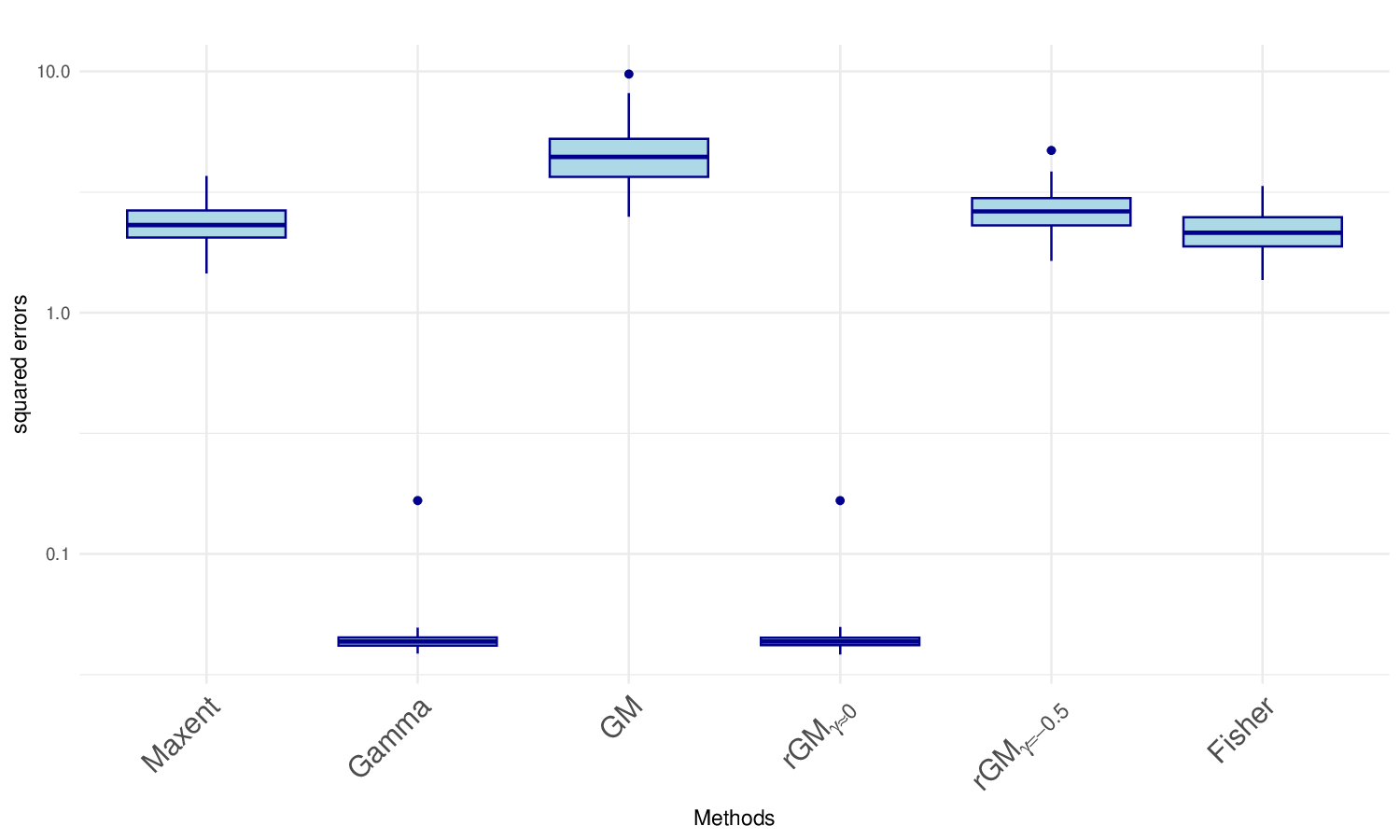}\\
(a) squared errors
  \end{center}
 \end{minipage}
 \begin{minipage}{0.5\hsize}
  \begin{center}
     \includegraphics[width=8cm,height=9cm]{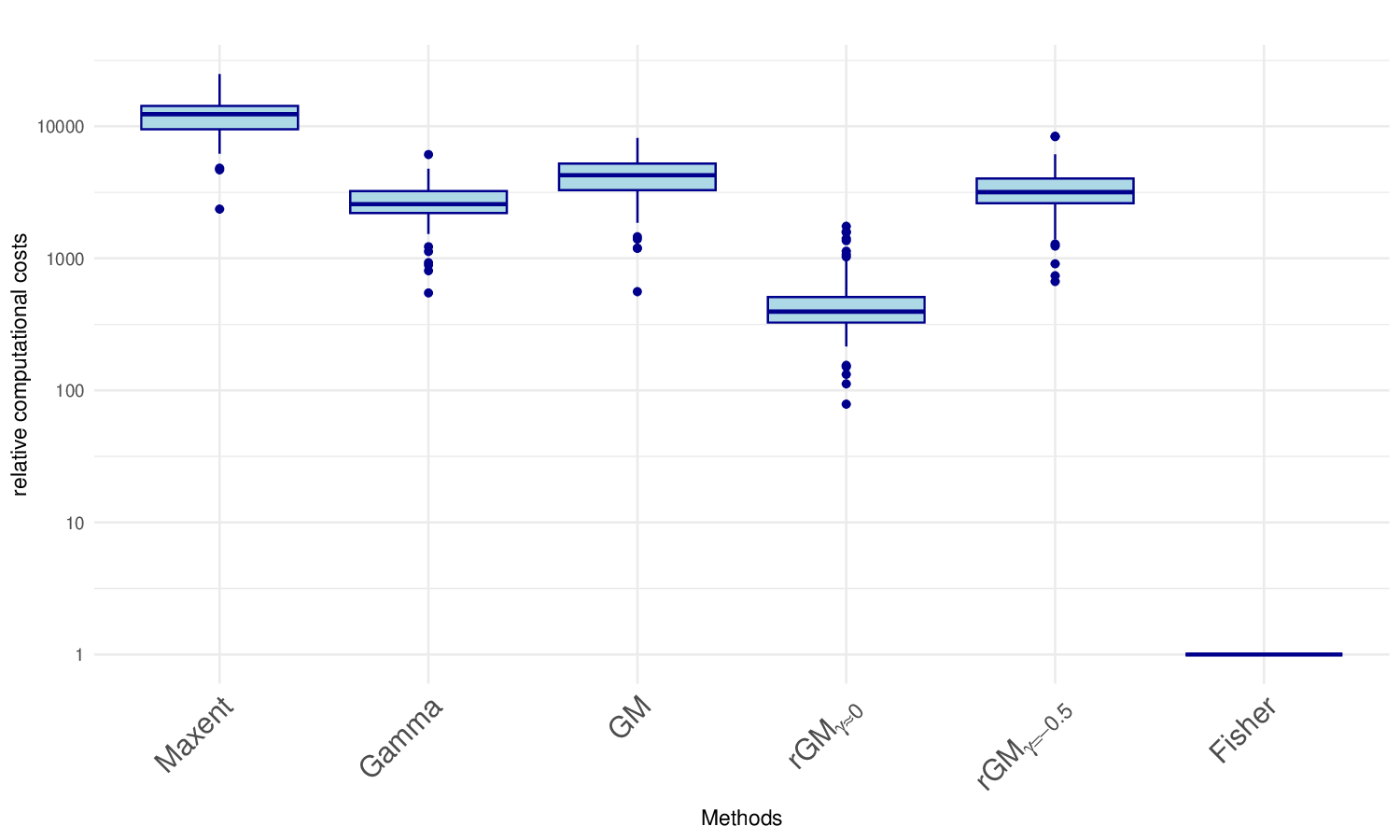}\\
(b) relative computational costs
  \end{center}
 \end{minipage}
 \begin{center}
correlation setting $\rho=0.5$
  \end{center}

 \begin{minipage}{0.5\hsize}
  \begin{center}
   \includegraphics[width=8cm,height=9cm]{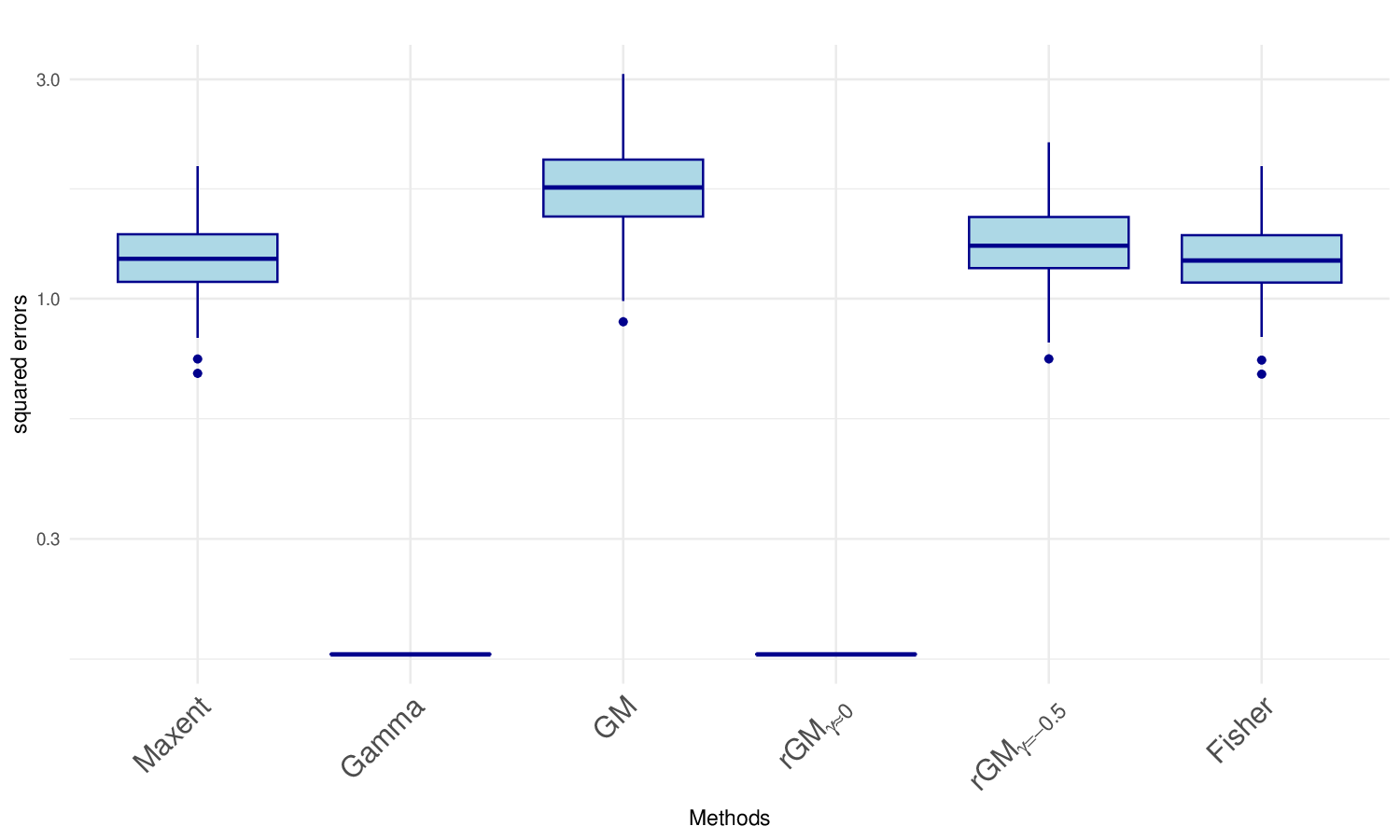}\\
(c) squared errors
  \end{center}
 \end{minipage}
  \begin{minipage}{0.5\hsize}
  \begin{center}
   \includegraphics[width=8cm,height=9cm]{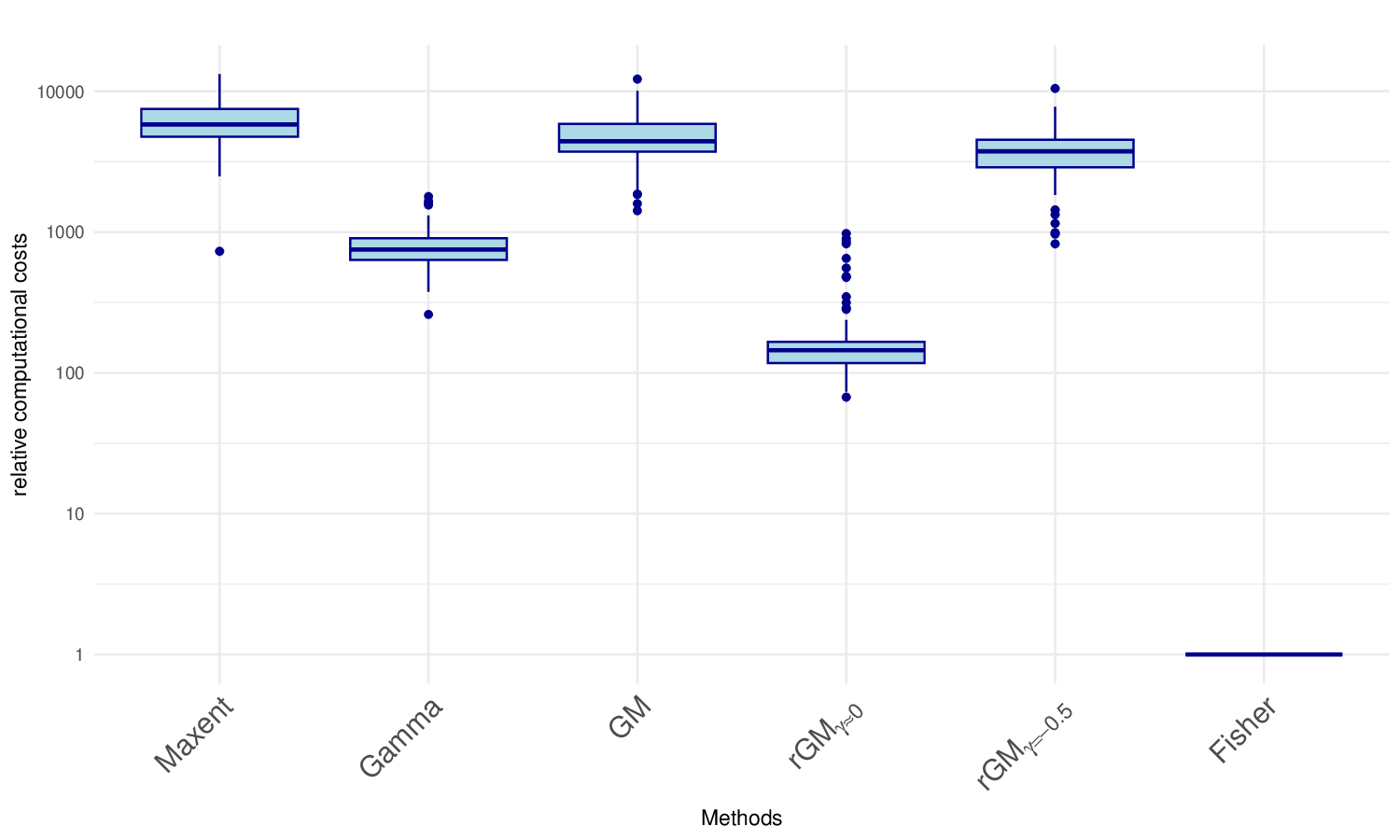}\\
(d) relative computational costs
  \end{center}
 \end{minipage}
 \begin{center}
correlation setting $\rho=0$
  \end{center}

\caption{Comparison of performances in Uniform case with $p=50$, $m=500$ and $n=10000$}\label{fig_box05_u}
\end{figure}

\section{Analysis of NCEAS data}\label{appendH}
\vspace{-0.5cm}

\begin{figure}[H]

  \begin{center}
   \includegraphics[width=15cm,height=12cm]{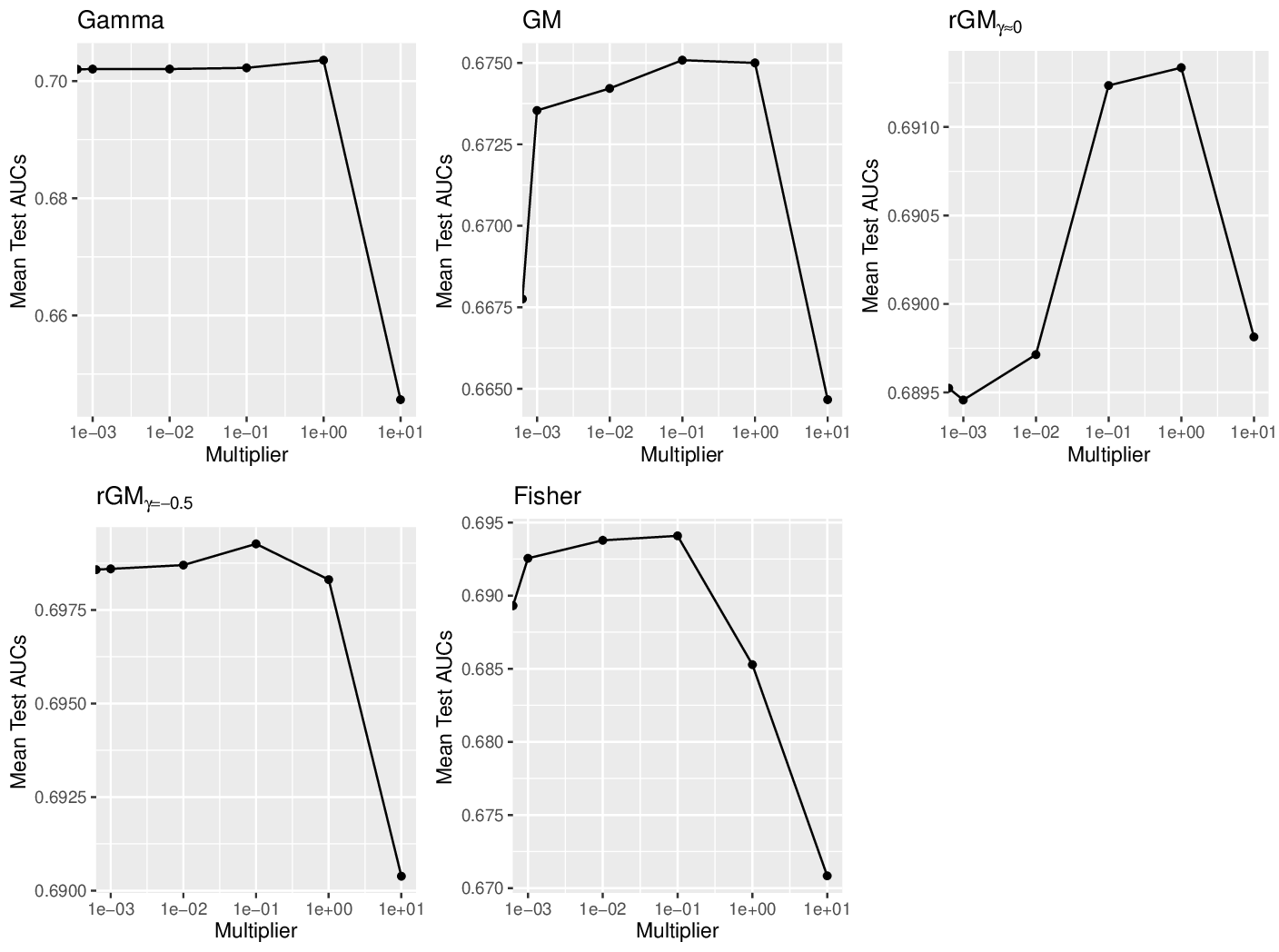}
  \end{center}
  \caption{Plots of mean test AUCs against multipliers of the optimal penalty term $\tau$ of linear features of Maxent.}\label{fig_multi}
\end{figure}

\newpage

\begin{figure}[H]
 \begin{minipage}{0.5\hsize}
  \begin{center}
   \includegraphics[width=8cm,height=6cm]{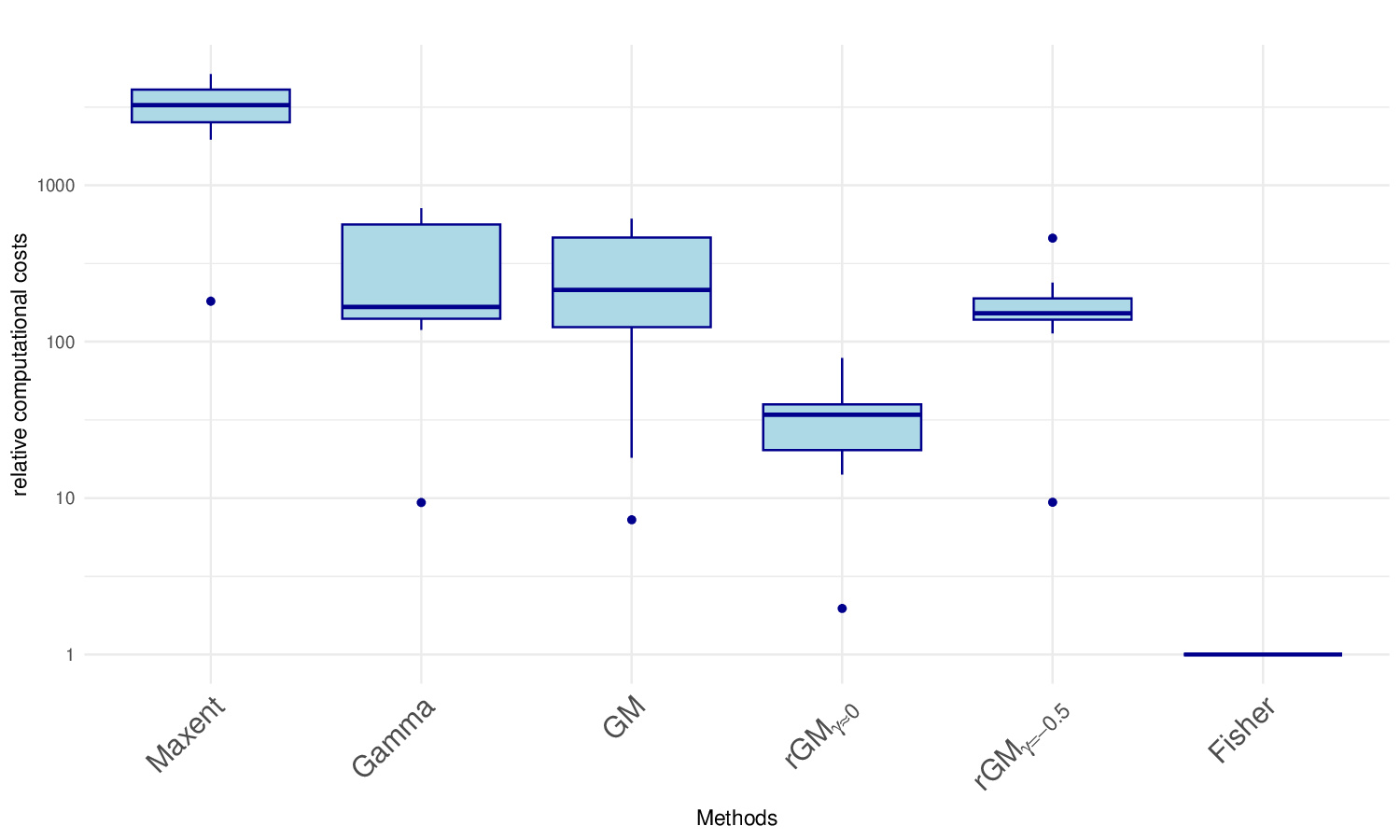}\\
(a) AWT ($p=8$, $9\leq m\leq 265$)
  \end{center}
 \end{minipage}
 \begin{minipage}{0.5\hsize}
  \begin{center}
  \includegraphics[width=8cm,height=6cm]{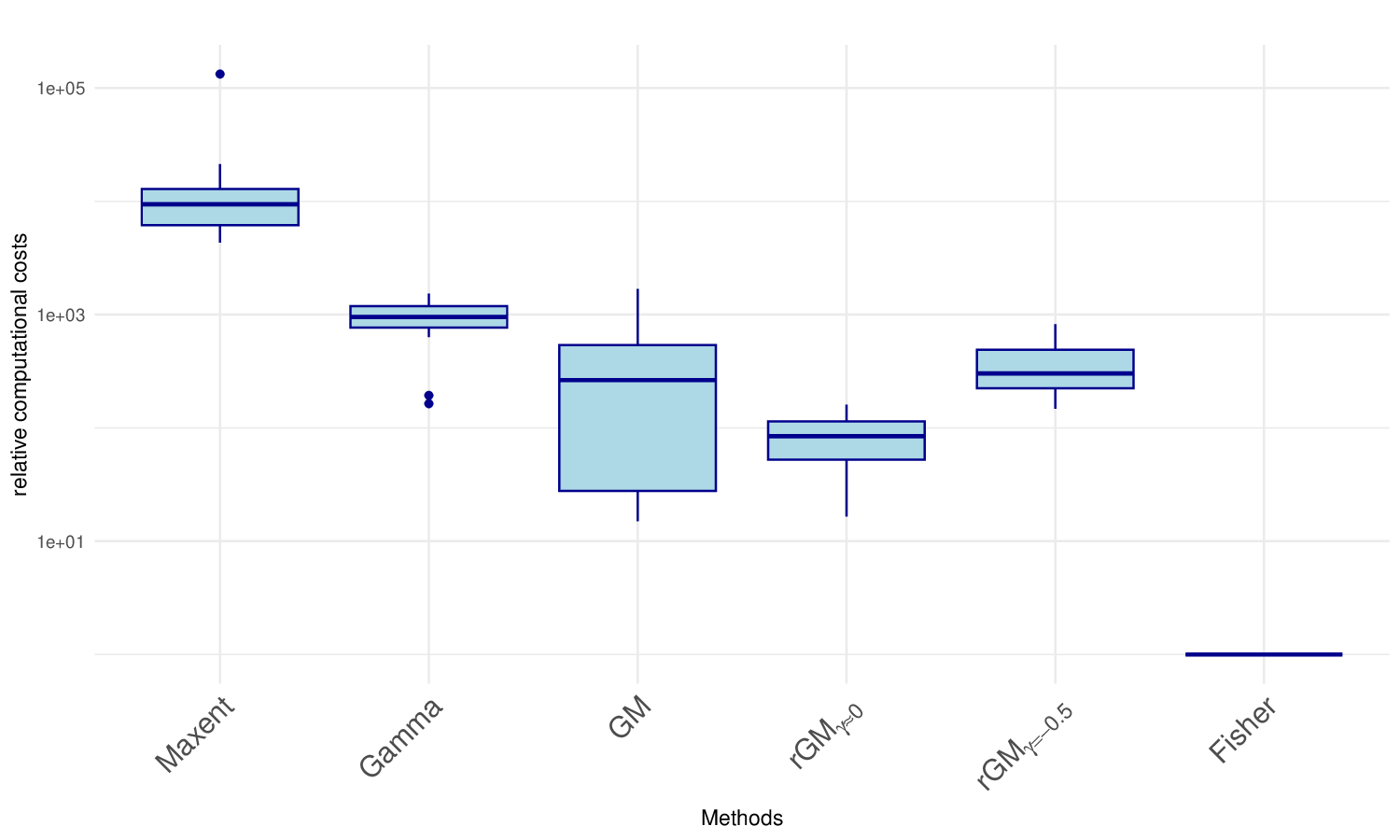}\\
(b)CAN ($p=12$, $16\leq m\leq 740$)
  \end{center}
 \end{minipage}
 \begin{minipage}{0.5\hsize}

 \end{minipage}

 \begin{minipage}{0.5\hsize}
  \begin{center}
   \includegraphics[width=8cm,height=6cm]{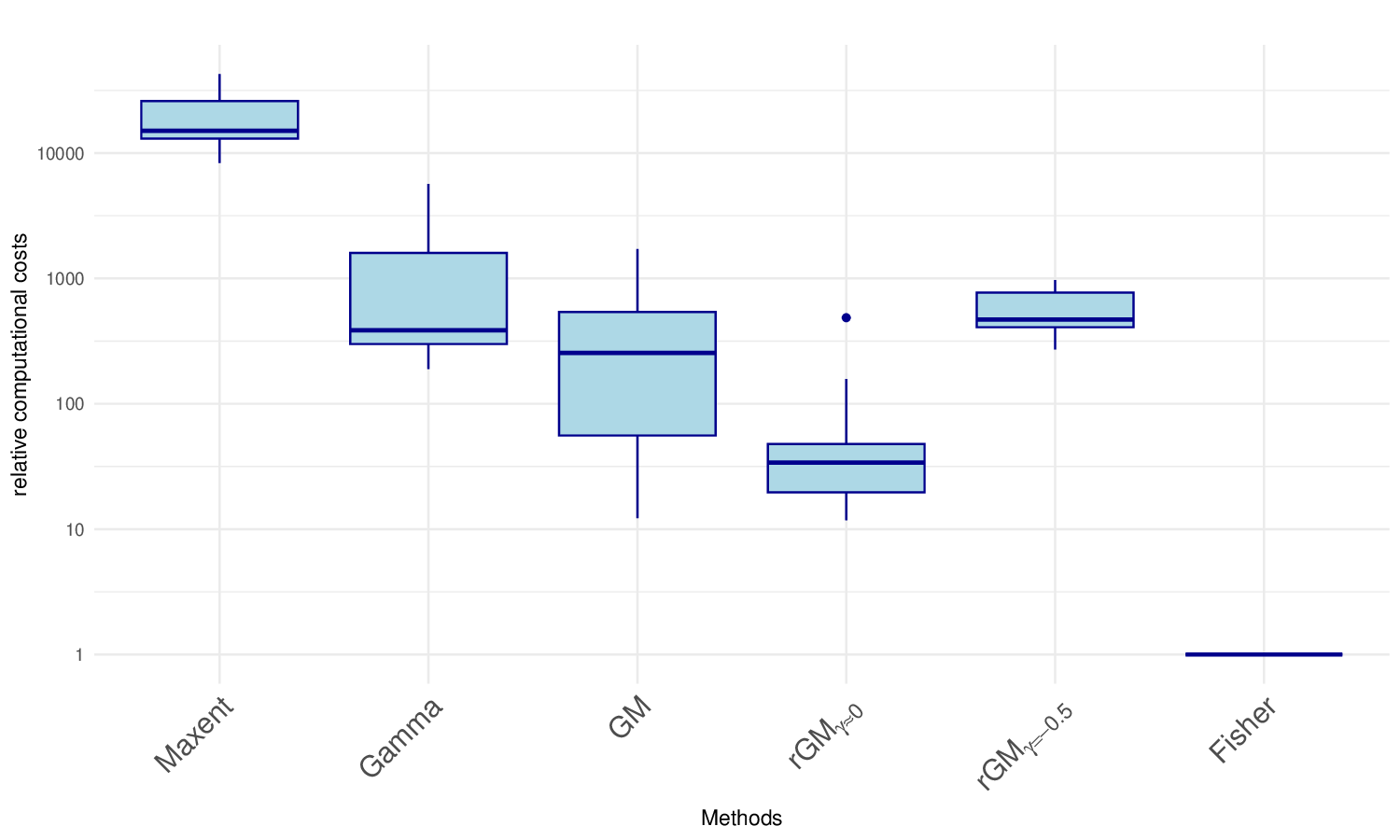}\\
(c) NSW ($p=20$, $2\leq m \leq 426$)
  \end{center}
 \end{minipage}
  \begin{minipage}{0.5\hsize}
  \begin{center}
   \includegraphics[width=8cm,height=6cm]{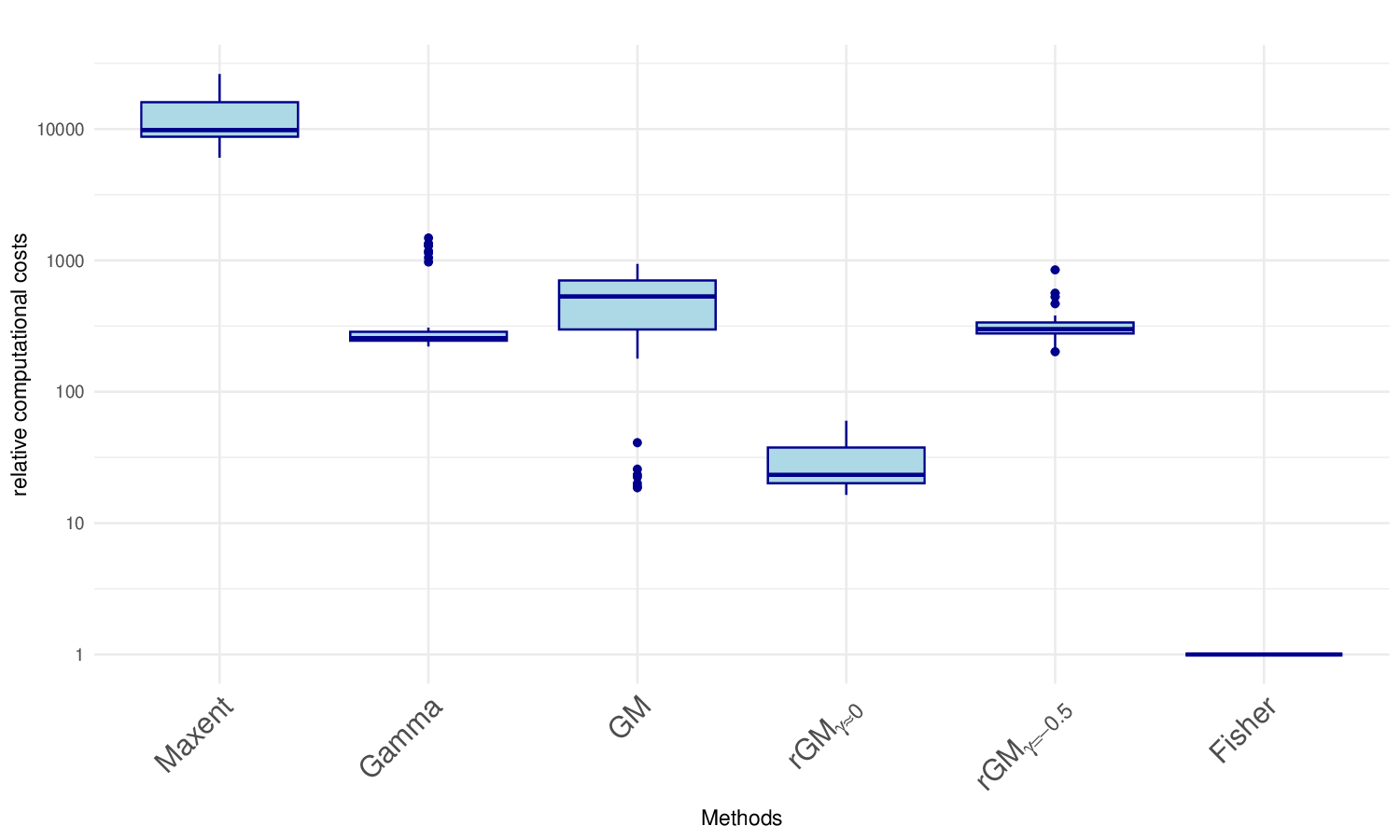}\\
(d) NZ ($p=15$, $18\leq m \leq 211$)
  \end{center}
 \end{minipage}

\begin{minipage}{0.5\hsize}
  \begin{center}
   \includegraphics[width=8cm,height=6cm]{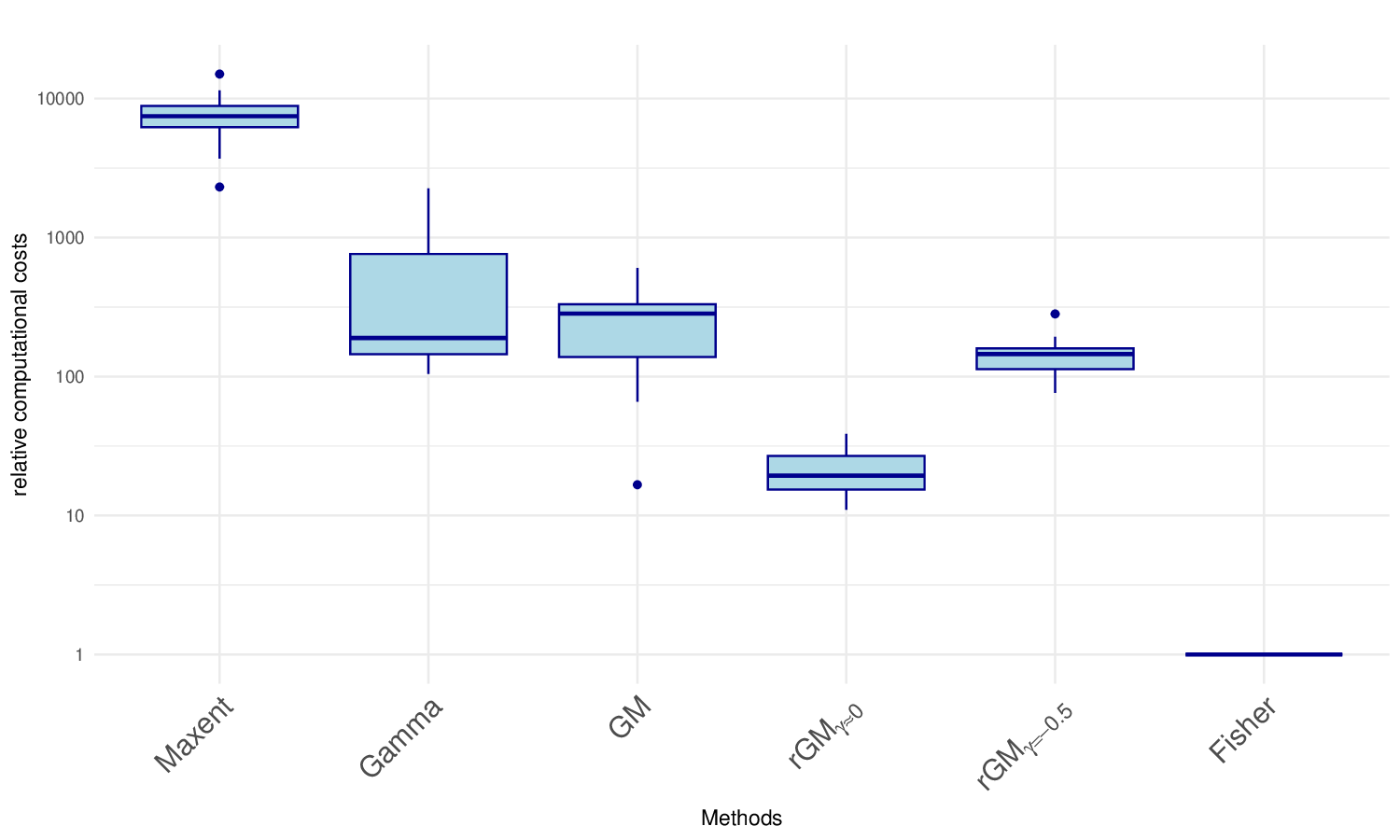}\\
(e) SA ($p=8$, $17\leq m \leq 216$)
  \end{center}
 \end{minipage}
  \begin{minipage}{0.5\hsize}
  \begin{center}
   \includegraphics[width=8cm,height=6cm]{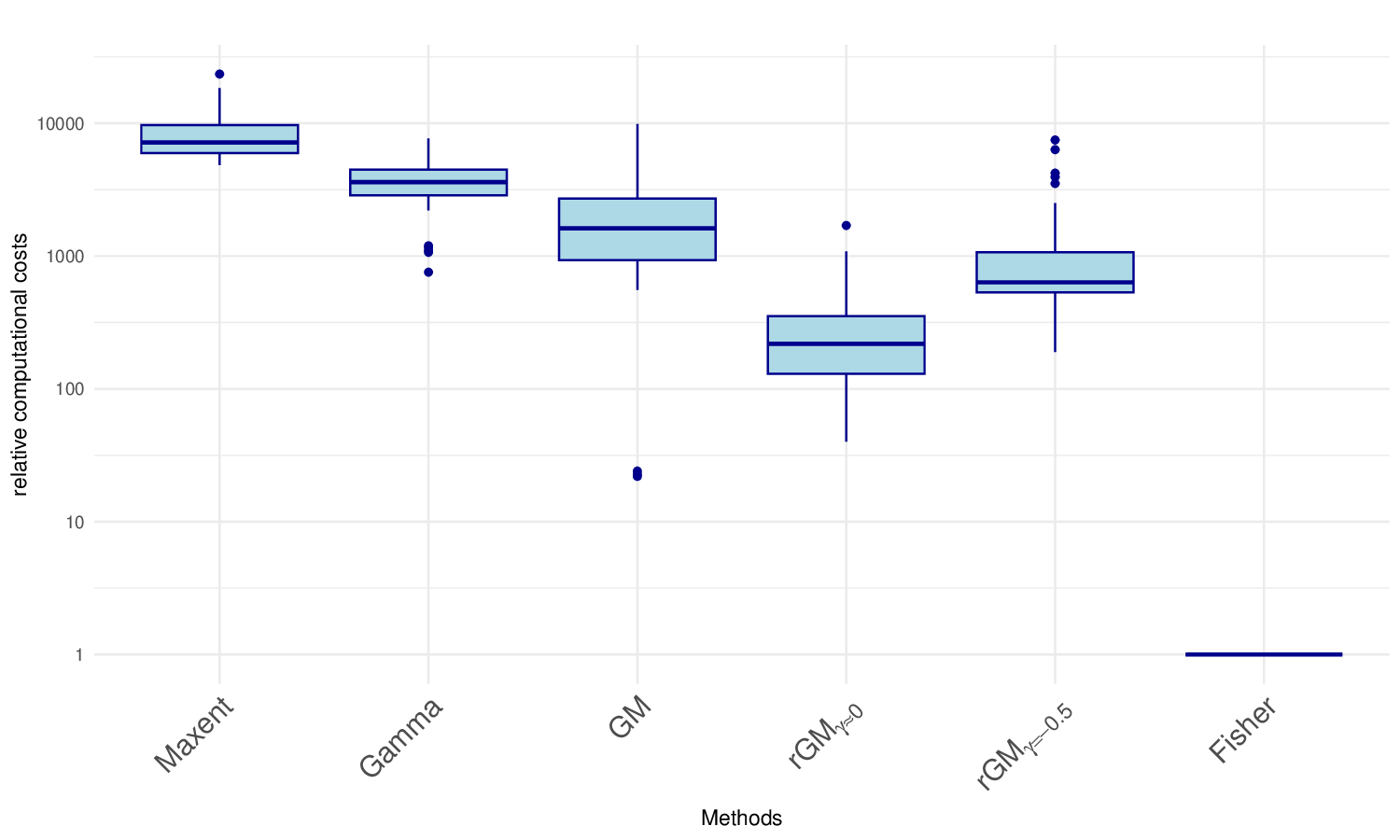}\\
(f) SWI ($p=14$, $36\leq m \leq 5822$)
  \end{center}
 \end{minipage}

\caption{Comparison of performances in terms of relative computational costs based on NCEAS data with background size $n=10000$}\label{fig_elith_time}
\end{figure}

\section{Analysis of Japanese vascular plants data}\label{appendI}
\begin{spacing}{0.9}
\begin{table}[H]

\centering
\caption{37 environmental variables used in data analysis} \label{37vari}
\vspace{0.5cm}
\begin{tabular}{rll}
  \hline
 & name & contents \\ 
  \hline
1 & land.area & land area   \\ 
  2 & forest.area & forest area \\ 
  3 & water.area & inland waters area   \\ 
  4 & elevation & elevation \\ 
  5 & elevationSD & standard deviation of elevation \\ 
  6 & laplacian & laplacian \\ 
  7 & slope & slope \\ 
  8 & radiation & annual average of radiation \\ 
  9 & sunshine & annual amount of sunshine \\ 
  10 & snow & height of snow \\ 
  11 & orc & organic carbon \\ 
  12 & cec & cation exchange capacity \\ 
  13 & ph & pH of soil \\ 
  14 & bio\_1 & annual mean temperature \\ 
  15 & bio\_2 & mean diurnal range (mean of monthly (max temp - min temp)) \\ 
  16 & bio\_3 & isothermality (bio\_2/bio\_7)  \\ 
  17 & bio\_4 & temperature seasonality  \\ 
  18 & bio\_5 & max temperature of warmest month \\ 
  19 & bio\_6 & min temperature of coldest month \\ 
  20 & bio\_7 & temperature annual range (bio\_5-bio\_6) \\ 
  21 & bio\_8 & mean temperature of wettest quarter \\ 
  22 & bio\_9 & mean temperature of driest quarter \\ 
  23 & bio\_10 & mean temperature of warmest quarter \\ 
  24 & bio\_11 & mean temperature of coldest quarter \\ 
  25 & bio\_12 & annual precipitation \\ 
  26 & bio\_13 & precipitation of wettest month \\ 
  27 & bio\_14 & precipitation of eriest month \\ 
  28 & bio\_15 & precipitation seasonality (coefficient of variation) \\ 
  29 & bio\_16 & precipitation of wettest quarter \\ 
  30 & bio\_17 & precipitation of driest quarter \\ 
  31 & bio\_18 & precipitation of warmest quarter \\ 
  32 & bio\_19 & precipitation of coldest quarter \\ 
  33 & AET & actual evapotranspiration amount \\ 
  34 & PET & potential evapotranspiration amount \\ 
  35 & DEFICIT & PET-AET  \\ 
  36 & distance & distance to the sea coast \\ 
   \hline
\end{tabular}

\label{tab.vari}
\end{table}
\end{spacing}

\begin{table}[H]
\caption{Mean values of test AUC for Japanese vascular plants\label{Vascular_result79(AUCt)}} 
\begin{center}
\begin{tabular}{lllllllll}
\hline\hline
\multicolumn{1}{c}{$p$}&\multicolumn{1}{c}{$n$}&\multicolumn{1}{c}{$m$}&\multicolumn{1}{c}{Maxent}&\multicolumn{1}{c}{Gamma}&\multicolumn{1}{c}{GM}&\multicolumn{1}{c}{$\rm{rGM_{\gamma\approx 0}}$}&\multicolumn{1}{c}{$\rm{rGM_{\gamma=-0.5}}$}&\multicolumn{1}{c}{Fisher}\tabularnewline
\hline
36&4684&[4,17)&0.968&\bf{0.969}&0.812&0.962&0.964&0.955\tabularnewline
&&[17,37)&0.955&0.955&0.802&\bf{0.966}&\bf{0.966}&0.96\tabularnewline
&&[37,58)&\bf{0.963}&0.96&0.81&0.955&0.957&0.952\tabularnewline
&&[58,101)&0.949&0.95&0.854&0.944&\bf{0.951}&0.941\tabularnewline
&&[101,170)&0.905&0.905&0.836&0.895&\bf{0.907}&0.895\tabularnewline
&&[170,306)&\bf{0.867}&0.864&0.794&0.852&\bf{0.867}&0.848\tabularnewline
&&[306,422)&\bf{0.838}&0.836&0.785&0.828&0.836&0.824\tabularnewline
&&[422,626)&\bf{0.771}&0.768&0.723&0.762&0.768&0.759\tabularnewline
&&[626,890)&0.729&0.727&0.701&0.727&\bf{0.732}&0.724\tabularnewline
&&[890,3077]&0.708&0.706&0.689&0.705&\bf{0.709}&0.7\tabularnewline
\hline
\end{tabular}\end{center}
\end{table}

\begin{table}[H]
\caption{Mean values of training AUC for Japanese vascular plants\label{Vascular_result79(AUC)}} 
\begin{center}
\begin{tabular}{lllllllll}
\hline\hline
\multicolumn{1}{c}{$p$}&\multicolumn{1}{c}{$n$}&\multicolumn{1}{c}{$m$}&\multicolumn{1}{c}{Maxent}&\multicolumn{1}{c}{Gamma}&\multicolumn{1}{c}{GM}&\multicolumn{1}{c}{$\rm{rGM_{\gamma\approx 0}}$}&\multicolumn{1}{c}{$\rm{rGM_{\gamma=-0.5}}$}&\multicolumn{1}{c}{Fisher}\tabularnewline
\hline
36&4684&[4,17)&\bf{0.993}&0.99&0.873&0.989&0.989&0.984\tabularnewline
&&[17,37)&\bf{0.988}&0.983&0.868&0.978&0.98&0.973\tabularnewline
&&[37,58)&\bf{0.983}&0.979&0.88&0.972&0.977&0.968\tabularnewline
&&[58,101)&\bf{0.972}&0.97&0.9&0.955&0.969&0.952\tabularnewline
&&[101,170)&\bf{0.932}&0.929&0.863&0.916&\bf{0.932}&0.911\tabularnewline
&&[170,306)&0.9&0.897&0.828&0.89&\bf{0.904}&0.883\tabularnewline
&&[306,422)&0.867&0.864&0.812&0.862&\bf{0.871}&0.855\tabularnewline
&&[422,626)&0.817&0.813&0.776&0.823&\bf{0.824}&0.811\tabularnewline
&&[626,890)&0.803&0.798&0.765&\bf{0.811}&0.81&0.8\tabularnewline
&&[890,3077]&0.791&0.786&0.766&\bf{0.803}&0.798&0.791\tabularnewline
\hline
\end{tabular}\end{center}
\end{table}

\begin{table}[H]
\caption{Mean values of relative computational costs for Japanese vascular plants\label{Vascular_result79(time)}} 
\begin{center}
\begin{tabular}{lllllllll}
\hline\hline
\multicolumn{1}{c}{$p$}&\multicolumn{1}{c}{$n$}&\multicolumn{1}{c}{$m$}&\multicolumn{1}{c}{Maxent}&\multicolumn{1}{c}{Gamma}&\multicolumn{1}{c}{GM}&\multicolumn{1}{c}{$\rm{rGM_{\gamma\approx 0}}$}&\multicolumn{1}{c}{$\rm{rGM_{\gamma=-0.5}}$}&\multicolumn{1}{c}{Fisher}\tabularnewline
\hline
36&4684&[4,17)&296.9&86&\bf{0.8}&3.7&1.8&1\tabularnewline
&&[17,37)&253.4&123.9&\bf{0.8}&12.3&5.3&1\tabularnewline
&&[37,58)&217.4&128.4&\bf{0.8}&16.7&8.9&1\tabularnewline
&&[58,101)&188.8&122.6&1.1&19.1&16.1&\bf{1}\tabularnewline
&&[101,170)&165&132.2&3.6&21.3&21.6&\bf{1}\tabularnewline
&&[170,306)&153.4&121.8&10.3&27.9&27.8&\bf{1}\tabularnewline
&&[306,422)&155&122.5&21.9&36.5&38.2&\bf{1}\tabularnewline
&&[422,626)&167.7&130.2&41.2&46.1&48.7&\bf{1}\tabularnewline
&&[626,890)&181.7&141.5&64.5&66.4&75.2&\bf{1}\tabularnewline
&&[890,3077]&201.1&131.2&139.3&127.2&149.3&\bf{1}\tabularnewline
\hline
\end{tabular}\end{center}
\end{table}

\begin{table}[H]
\caption{Mean values of number of selected variables for Japanese vascular plants\label{Vascular_result79(n.vari)}} 
\begin{center}
\begin{tabular}{lllllllll}
\hline\hline
\multicolumn{1}{c}{$p$}&\multicolumn{1}{c}{$n$}&\multicolumn{1}{c}{$m$}&\multicolumn{1}{c}{Maxent}&\multicolumn{1}{c}{Gamma}&\multicolumn{1}{c}{GM}&\multicolumn{1}{c}{$\rm{rGM_{\gamma\approx 0}}$}&\multicolumn{1}{c}{$\rm{rGM_{\gamma=-0.5}}$}&\multicolumn{1}{c}{Fisher}\tabularnewline
\hline
36&4684&[4,17)&9.7&7.2&\bf{1.2}&13.4&6.4&5.9\tabularnewline
&&[17,37)&14.9&10.6&\bf{1.3}&21.9&11.8&10\tabularnewline
&&[37,58)&16.7&12.8&\bf{1.3}&25.8&14.6&13.2\tabularnewline
&&[58,101)&18.3&14.5&\bf{2.4}&28.1&21.1&15.5\tabularnewline
&&[101,170)&18.9&16.3&\bf{5.5}&29.3&25.7&17.2\tabularnewline
&&[170,306)&18.9&16.7&\bf{10.3}&29.4&25.5&17.6\tabularnewline
&&[306,422)&19.1&16.4&\bf{13.4}&29.3&25.4&17.4\tabularnewline
&&[422,626)&19.2&\bf{16.9}&17.7&29&24.8&17.8\tabularnewline
&&[626,890)&19.3&\bf{16.9}&18.8&29.3&24.5&18.5\tabularnewline
&&[890,3077]&19.4&\bf{16.5}&19.1&28.9&24.3&17.9\tabularnewline
\hline
\end{tabular}\end{center}
\end{table}

\begin{table}[H]
\caption{Mean values of Jeffreys divergence from Maxent for Japanese vascular plants\label{Vascular_result79(jd)}} 
\begin{center}
\begin{tabular}{lllllllll}
\hline\hline
\multicolumn{1}{c}{$p$}&\multicolumn{1}{c}{$n$}&\multicolumn{1}{c}{$m$}&\multicolumn{1}{c}{Maxent}&\multicolumn{1}{c}{Gamma}&\multicolumn{1}{c}{GM}&\multicolumn{1}{c}{$\rm{rGM_{\gamma\approx 0}}$}&\multicolumn{1}{c}{$\rm{rGM_{\gamma=-0.5}}$}&\multicolumn{1}{c}{Fisher}\tabularnewline
\hline
36&4684&[4,17)&*&\bf{0.317}&7.282&3.15&3.566&3.289\tabularnewline
&&[17,37)&*&\bf{0.271}&7.236&6.057&7.813&6.782\tabularnewline
&&[37,58)&*&\bf{0.109}&7.126&6.517&10.423&5.623\tabularnewline
&&[58,101)&*&\bf{0.038}&6.384&6.314&5.907&4.645\tabularnewline
&&[101,170)&*&\bf{0.011}&5.799&3.777&1.353&2.989\tabularnewline
&&[170,306)&*&\bf{0.008}&6.611&1.69&0.459&1.289\tabularnewline
&&[306,422)&*&\bf{0.007}&6.251&1.057&0.265&0.82\tabularnewline
&&[422,626)&*&\bf{0.006}&5.182&0.358&0.103&0.268\tabularnewline
&&[626,890)&*&\bf{0.005}&4.826&0.179&0.065&0.132\tabularnewline
&&[890,3077]&*&\bf{0.005}&3.521&0.101&0.033&0.082\tabularnewline
\hline
\end{tabular}\end{center}
\end{table}

\begin{figure}[H]
 \begin{minipage}{0.5\hsize}
  \begin{center}
   \includegraphics[width=8cm,height=10cm]{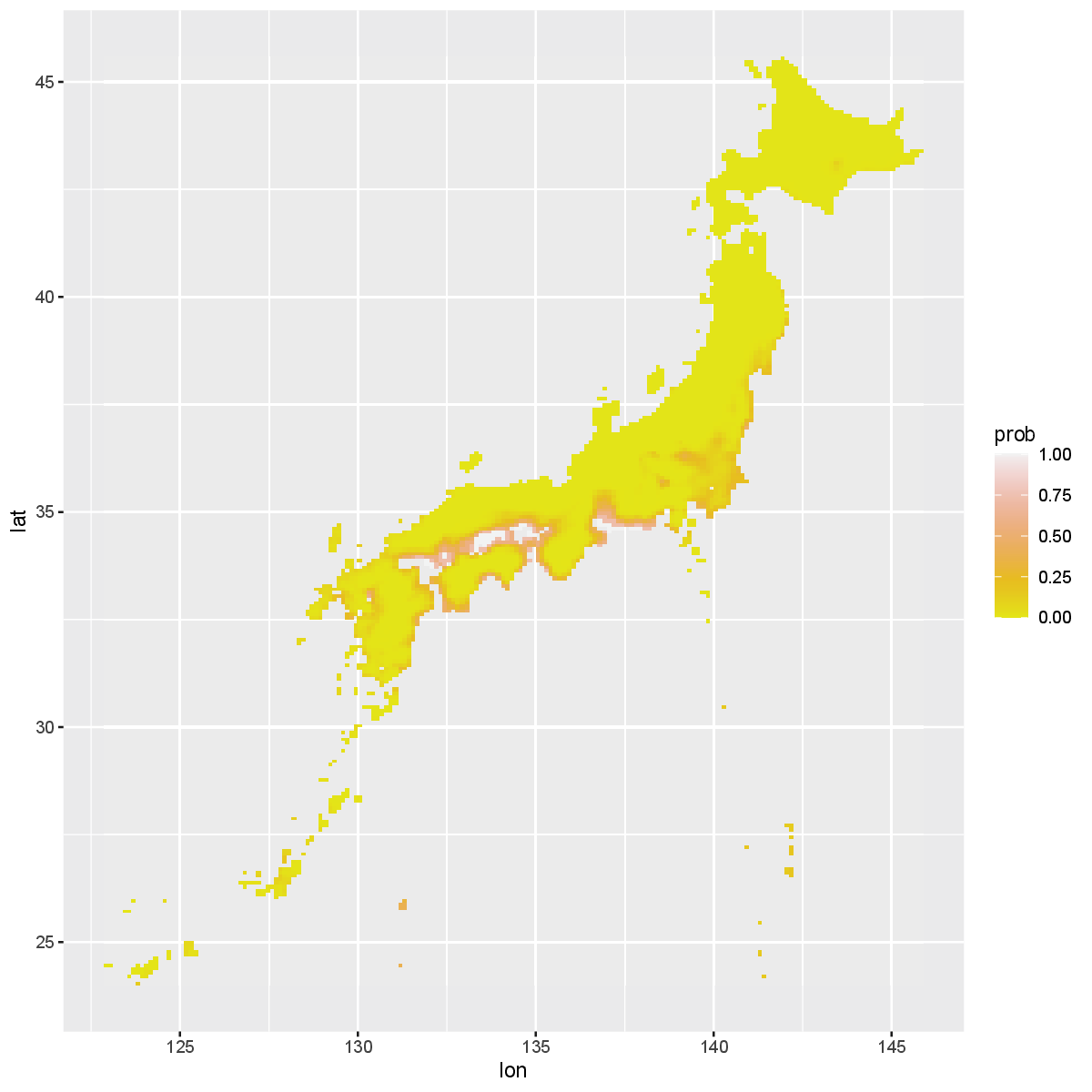}\\
(a) Maxent
  \end{center}
 \end{minipage}
 \begin{minipage}{0.5\hsize}
  \begin{center}
   \includegraphics[width=8cm,height=10cm]{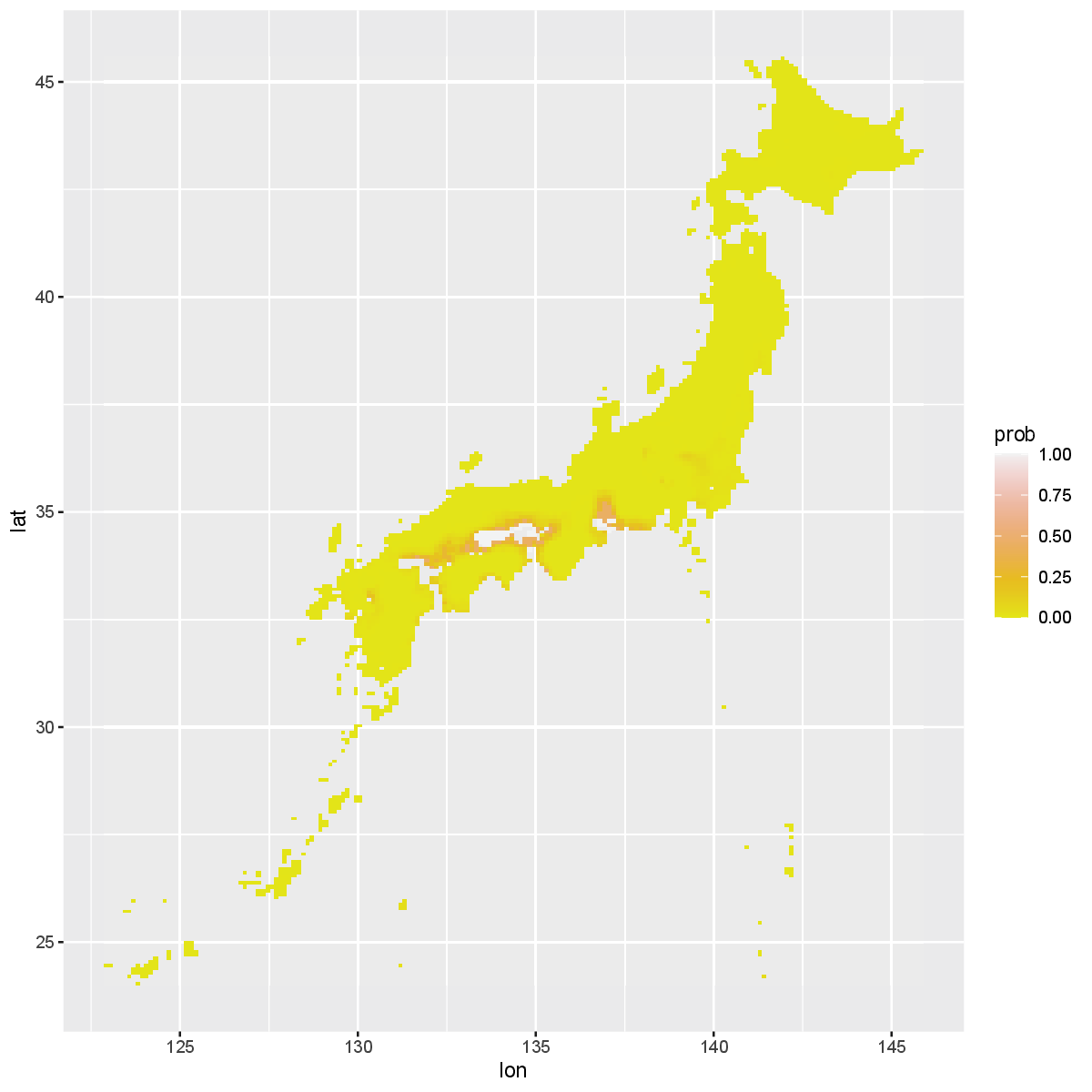}\\
(b) $\rm{rGM_{\gamma=-0.5}}$
  \end{center}
 \end{minipage}
 \begin{minipage}{0.5\hsize}

 \end{minipage}

 \begin{minipage}{0.5\hsize}
  \begin{center}
   \includegraphics[width=8cm,height=10cm]{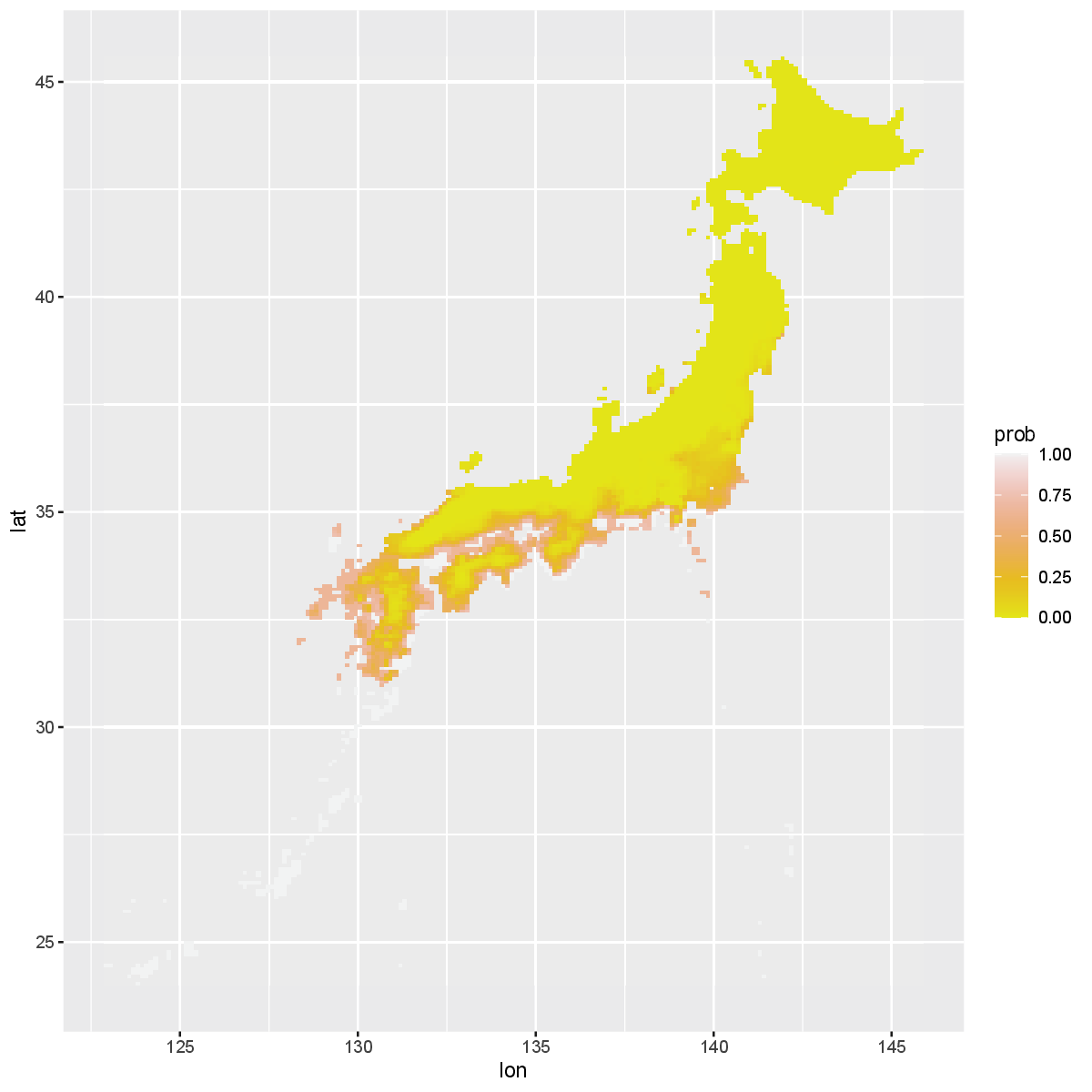}\\
(c) GM
  \end{center}
 \end{minipage}
  \begin{minipage}{0.5\hsize}
  \begin{center}
   \includegraphics[width=8cm,height=10cm]{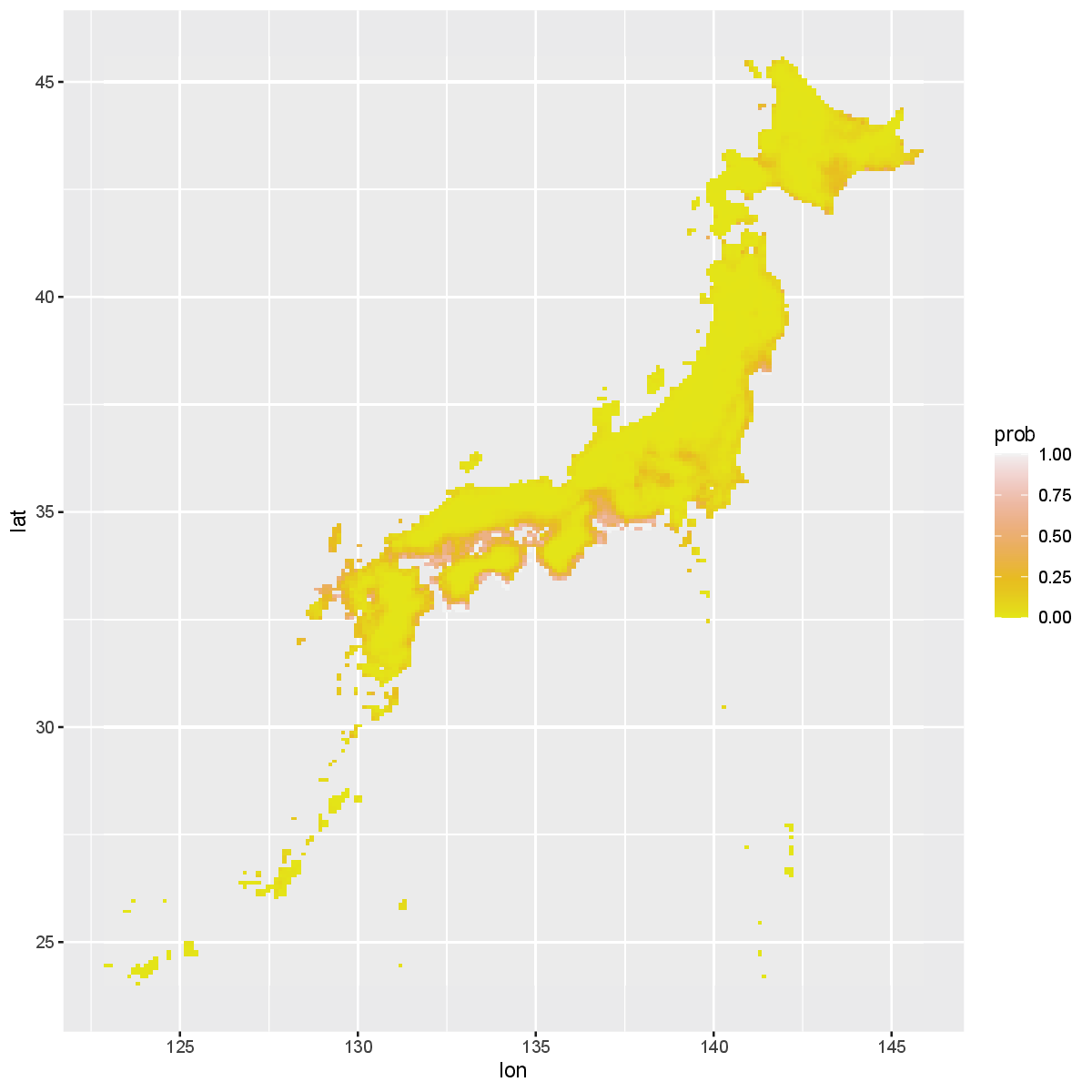}\\
(d) Fisher
  \end{center}
 \end{minipage}

\caption{Estimated habitat maps for {\it Hypericum tosaense} ($m=5$)}\label{fig_habitatS1}
\end{figure}

\begin{figure}[H]
 \begin{minipage}{0.5\hsize}
  \begin{center}
   \includegraphics[width=8cm,height=10cm]{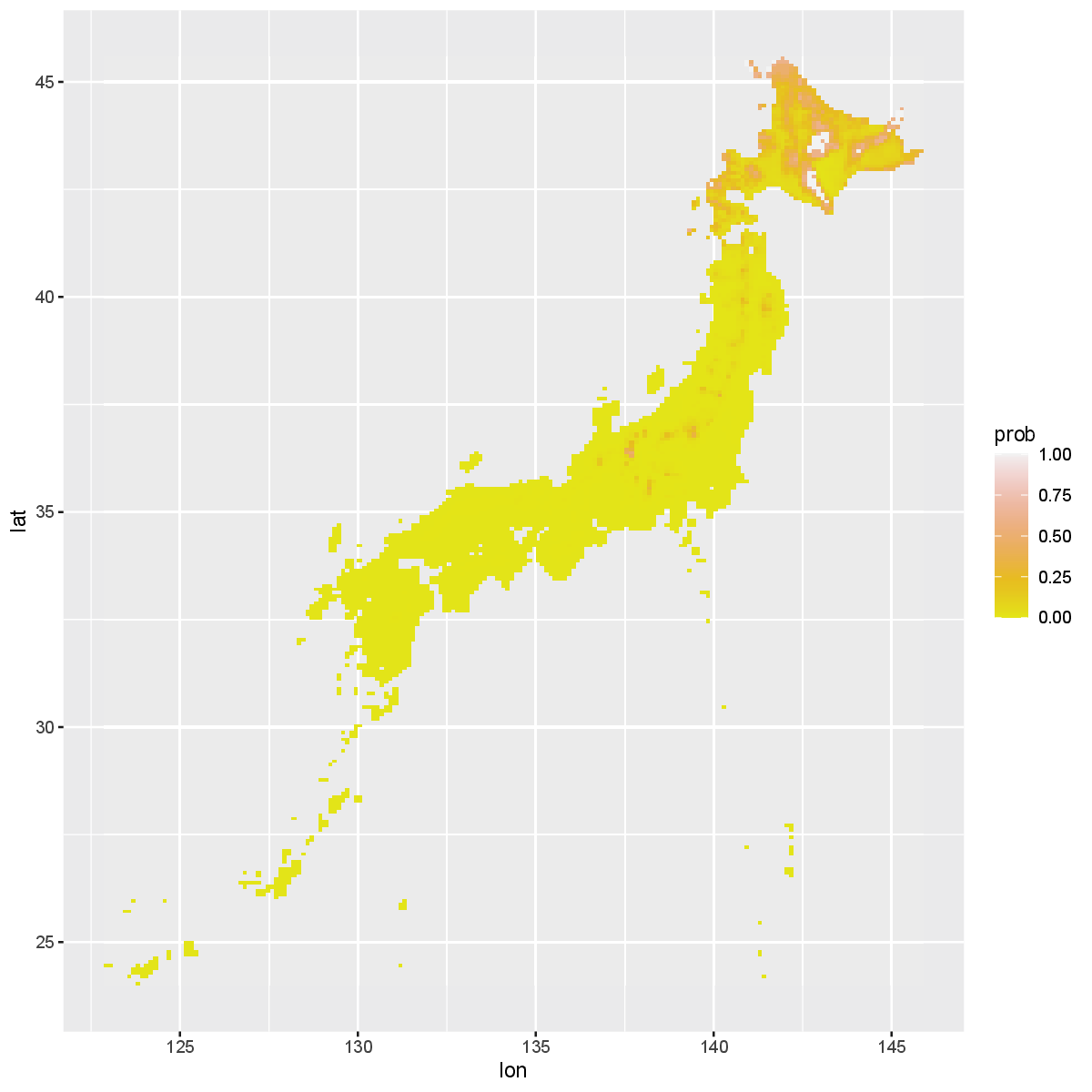}\\
(a) Maxent
  \end{center}
 \end{minipage}
 \begin{minipage}{0.5\hsize}
  \begin{center}
   \includegraphics[width=8cm,height=10cm]{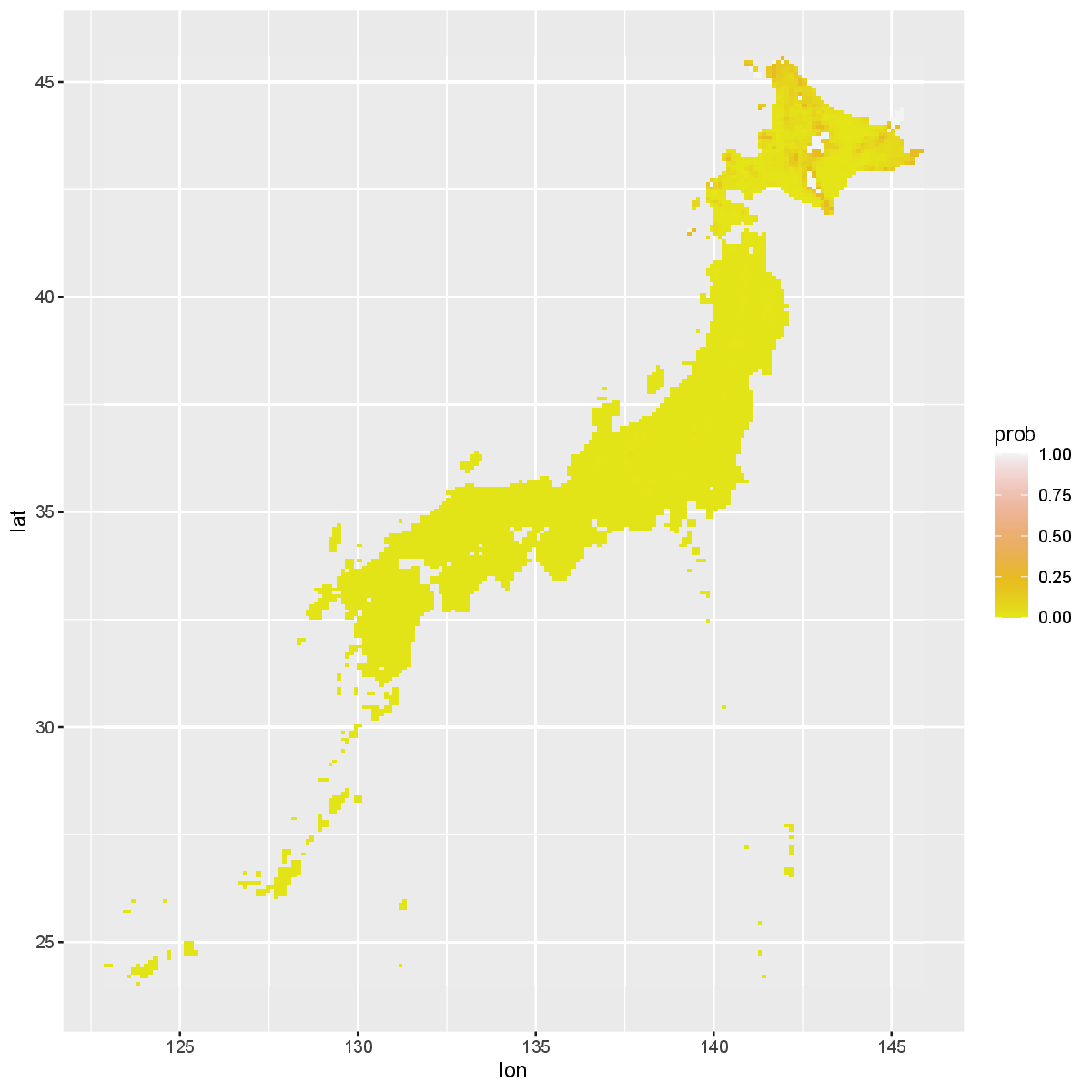}\\
(b) $\rm{rGM_{\gamma=-0.5}}$
  \end{center}
 \end{minipage}
 \begin{minipage}{0.5\hsize}

 \end{minipage}

 \begin{minipage}{0.5\hsize}
  \begin{center}
   \includegraphics[width=8cm,height=10cm]{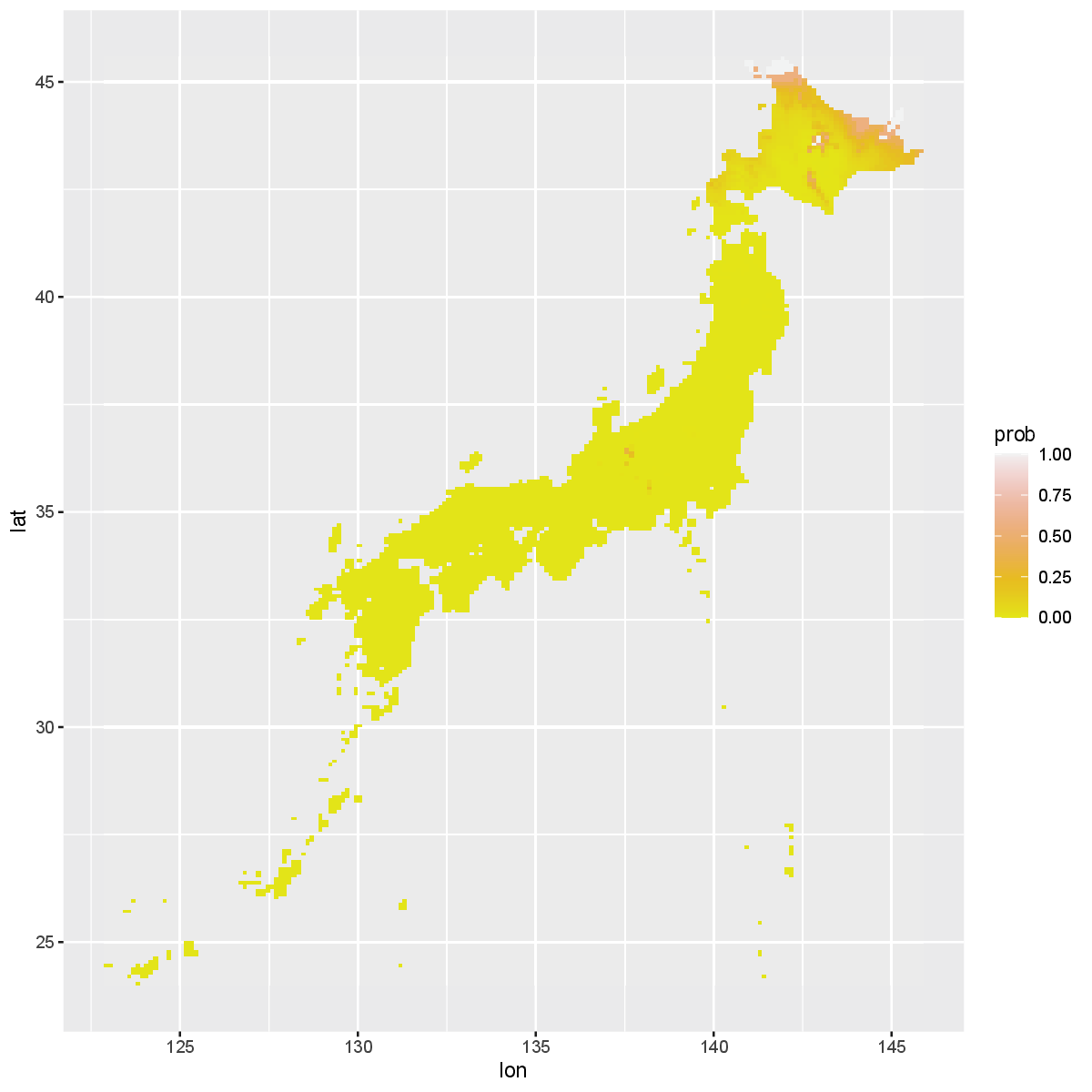}\\
(c) GM
  \end{center}
 \end{minipage}
  \begin{minipage}{0.5\hsize}
  \begin{center}
   \includegraphics[width=8cm,height=10cm]{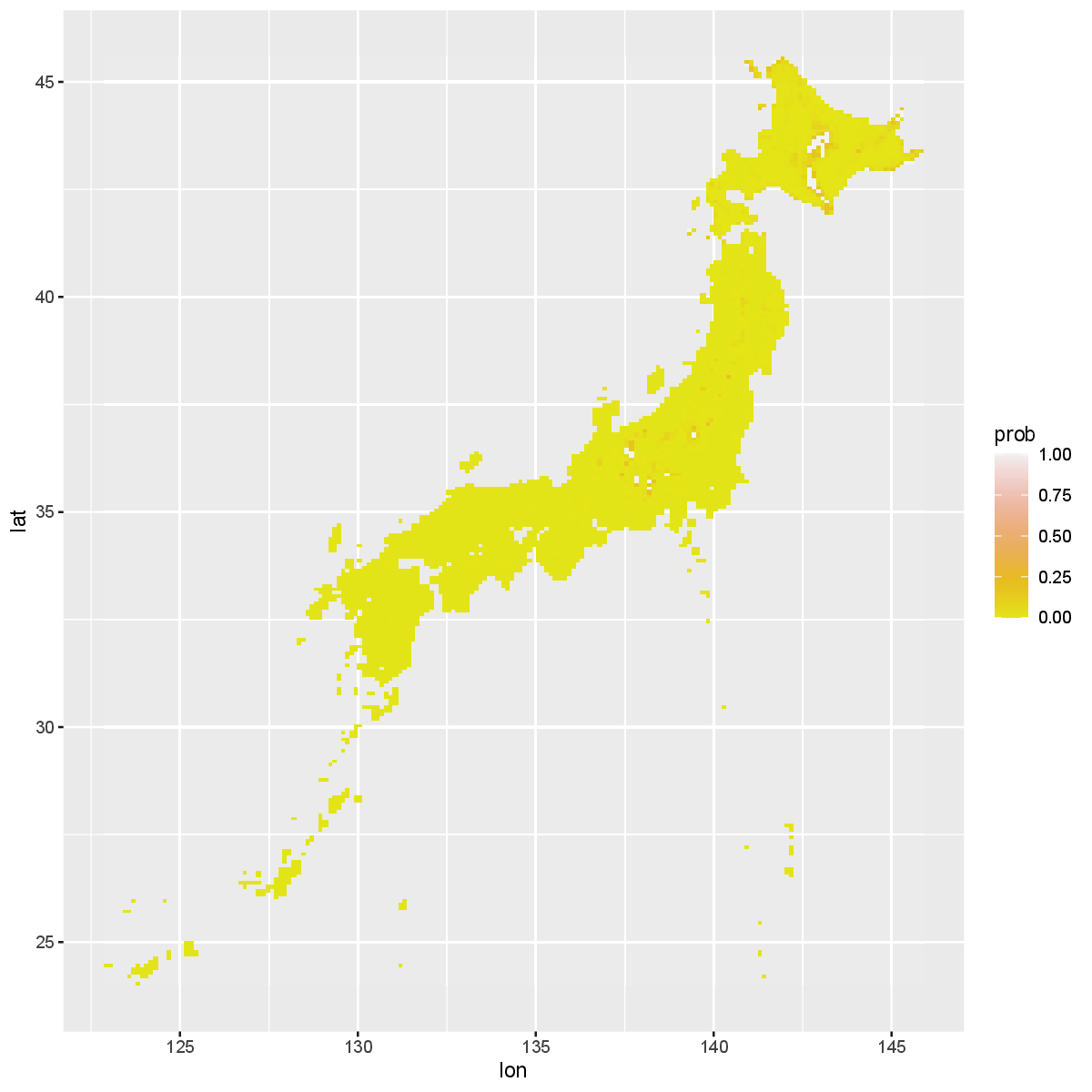}\\
(d) Fisher
  \end{center}
 \end{minipage}

\caption{Estimated habitat maps for {\it Patrinia sibirica} ($m=41$)}\label{fig_habitatS2}
\end{figure}

\begin{figure}[H]
 \begin{minipage}{0.5\hsize}
  \begin{center}
   \includegraphics[width=8cm,height=10cm]{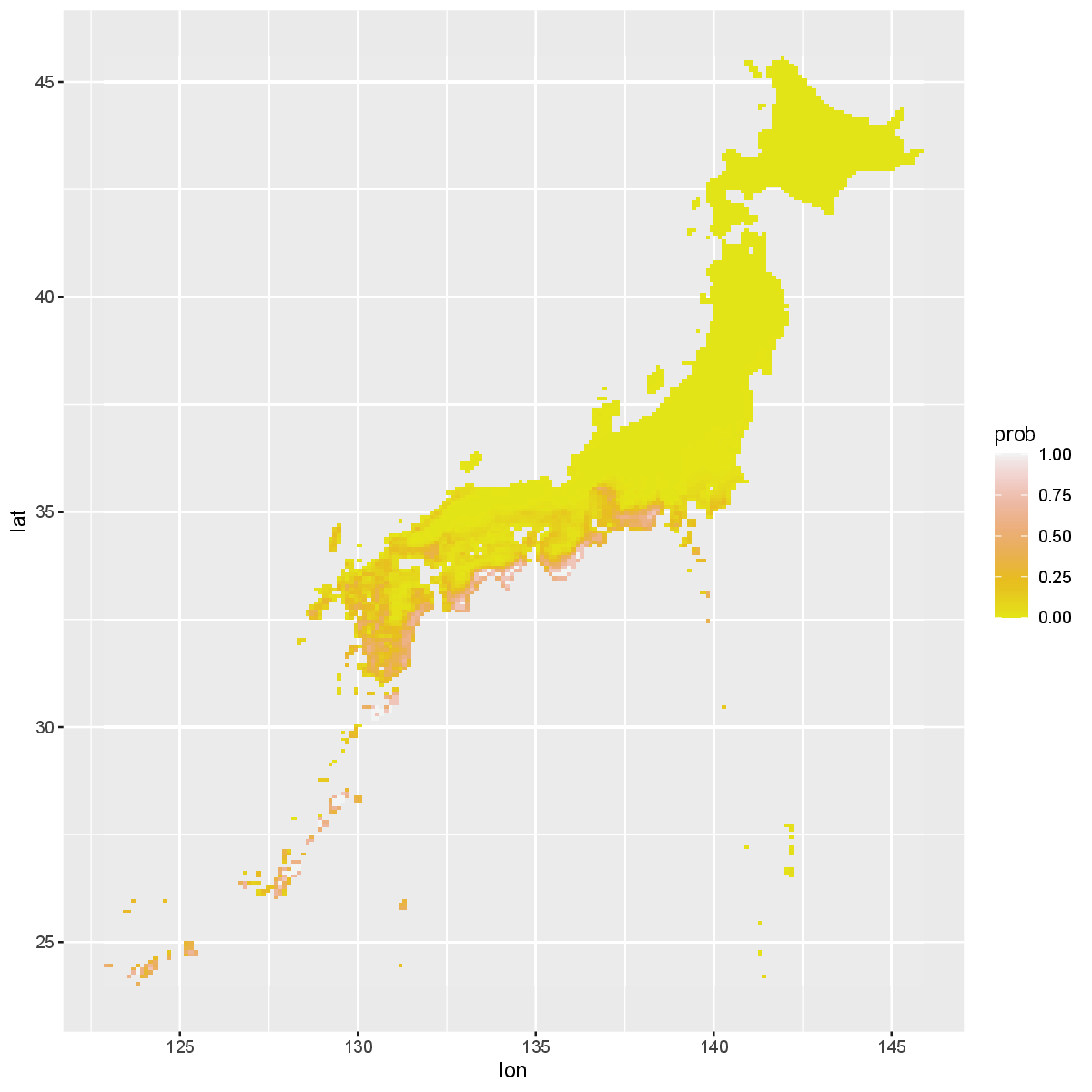}\\
(a) Maxent
  \end{center}
 \end{minipage}
 \begin{minipage}{0.5\hsize}
  \begin{center}
   \includegraphics[width=8cm,height=10cm]{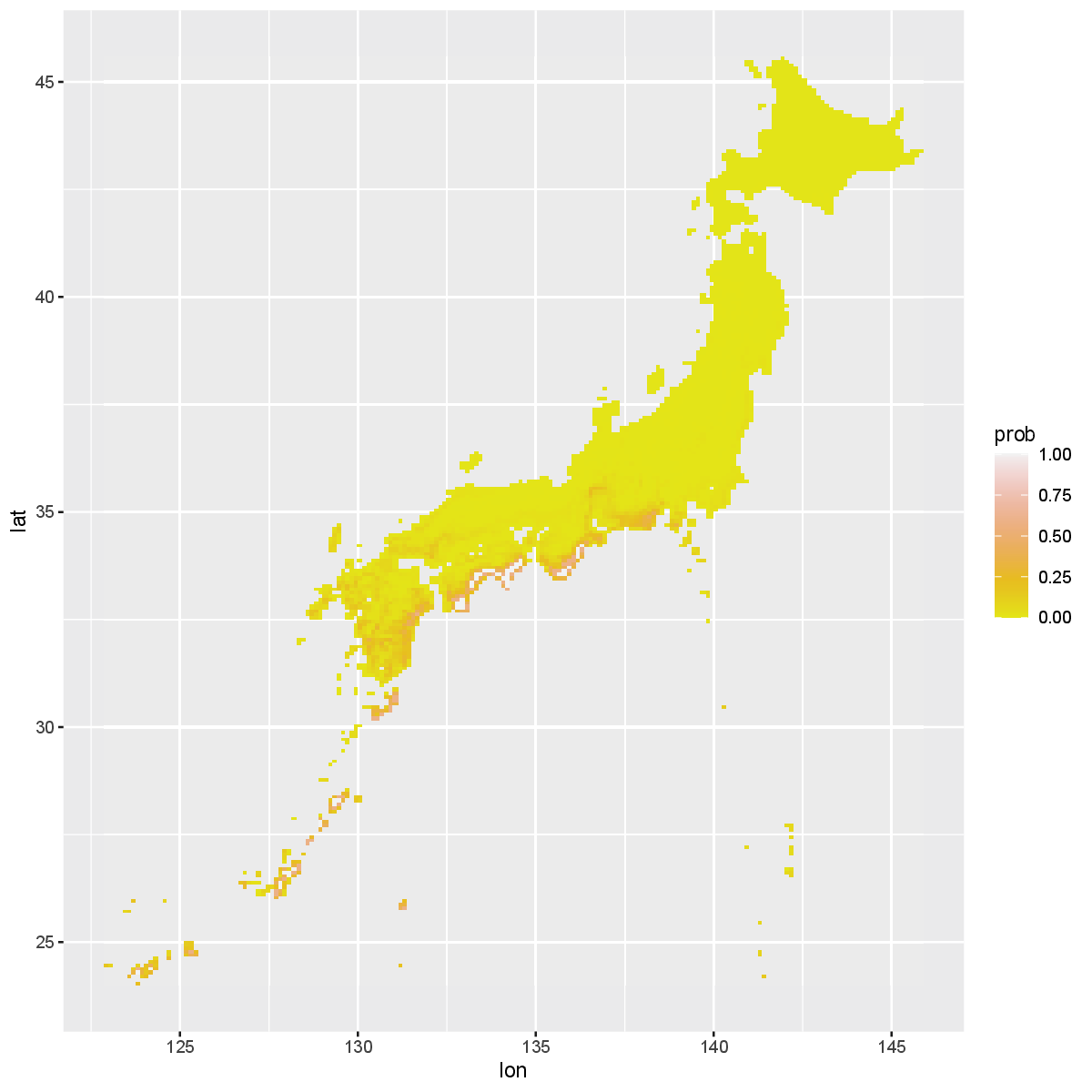}\\
(b) $\rm{rGM_{\gamma=-0.5}}$
  \end{center}
 \end{minipage}
 \begin{minipage}{0.5\hsize}

 \end{minipage}

 \begin{minipage}{0.5\hsize}
  \begin{center}
   \includegraphics[width=8cm,height=10cm]{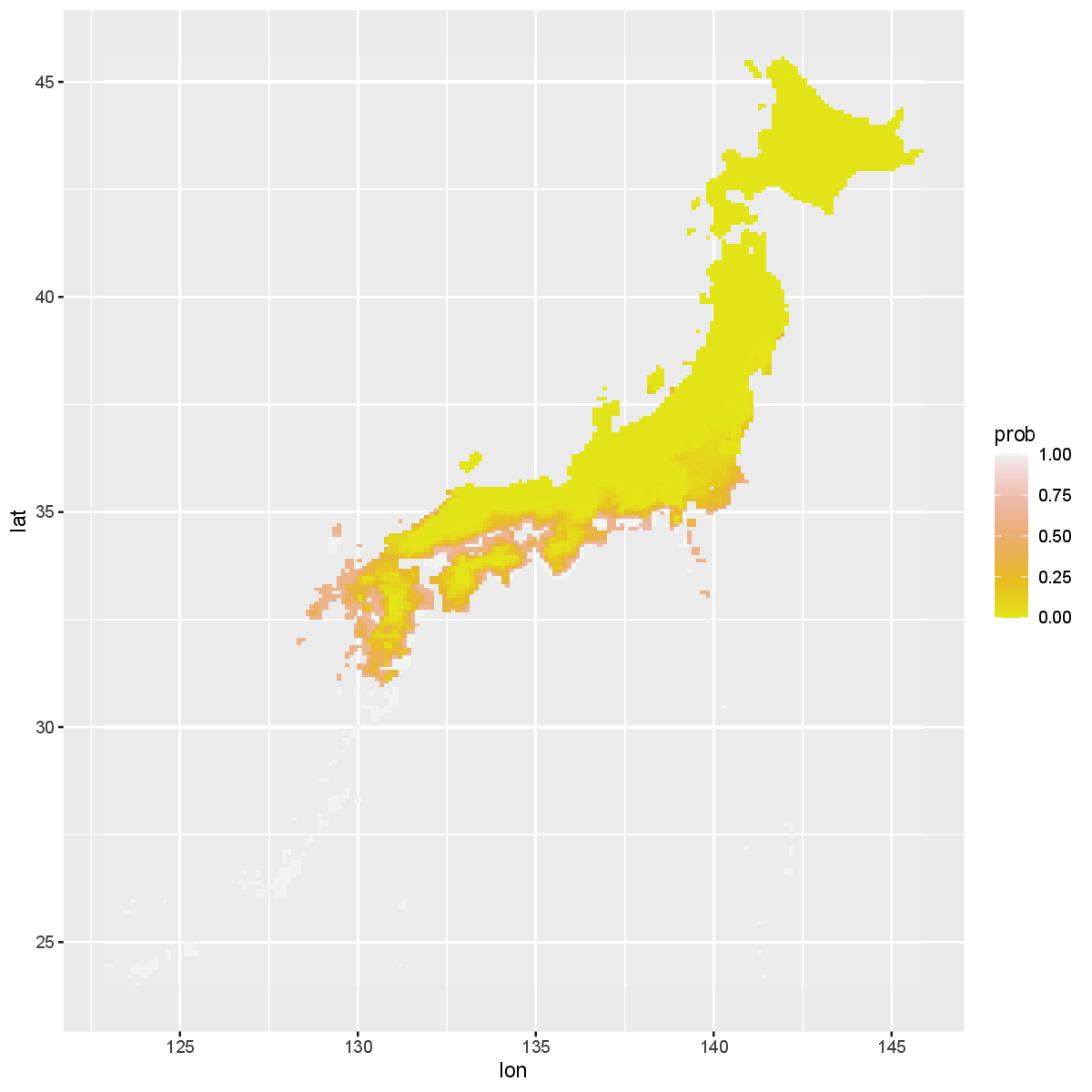}\\
(c) GM
  \end{center}
 \end{minipage}
  \begin{minipage}{0.5\hsize}
  \begin{center}
   \includegraphics[width=8cm,height=10cm]{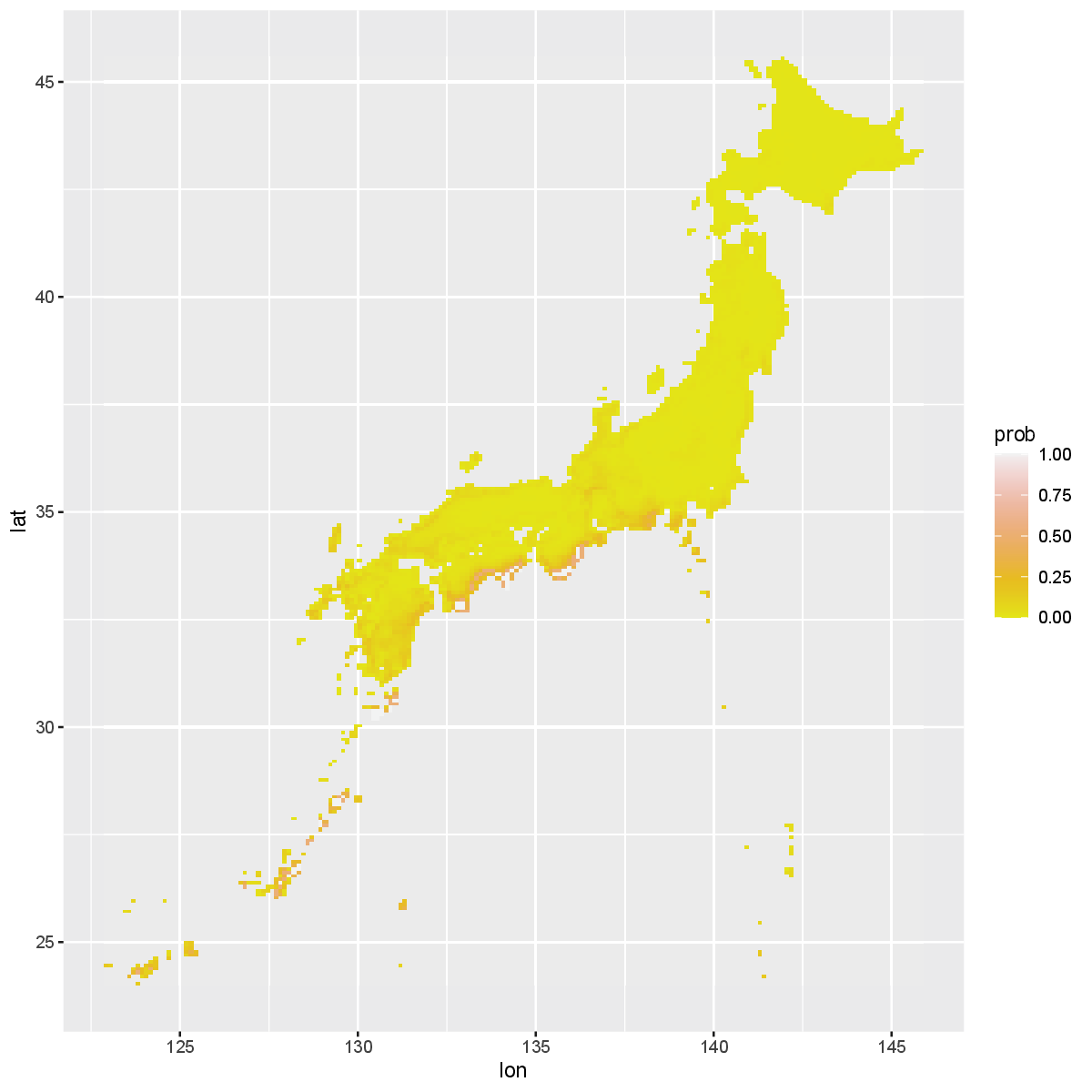}\\
(d) Fisher
  \end{center}
 \end{minipage}

\caption{Estimated habitat maps for {\it Diospyros morrisiana} ($m=127$)}\label{fig_habitatS3}
\end{figure}

\section{R code for a CBA package}\label{appendJ}
\noindent
\hrulefill
\begin{verbatim}
# Install a CBA package with dependencies
devtools::install_local("CBA_0.0.0.9000.tar.gz", dependencies = TRUE)
library(CBA)
help(CBA)
# Example
#Reproduce simulation results in Figure 1
sim() 
#Reproduce NCEAS data analysis in Figure 2
nceas() 
#Reproduce habitat maps in Figure 3
habitat() 
#Reproduce estimated coefficinets in Figure 4
boxplot0() 
#Reproduce paths about coefficients in Figure 5
make.path() 
#Maxent code to estimate species distribution for Pteridium aquilinum
maxent() 
#rGM code to estimate species distribution for Pteridium aquilinum
rgm() 
#GM code to estimate species distribution for Pteridium aquilinum
gm() 
#Gamma code to estimate species distribution for Pteridium aquilinum
gamma0() 
#Fisher code to estimate species distribution for Pteridium aquilinum
fisher0() 

#Reproduce habitat maps in Figure S5
habitat(sp=sp_hypericum_tosaense) 
#Reproduce habitat maps in Figure S6
habitat(sp=sp_patrinia_sibirica) 
#Reproduce habitat maps in Figure S7
habitat(sp=sp_diospyros_morrisiana)
\end{verbatim}
\hrulefill

\end{document}